%% file: index.tex
\documentclass[acmsmall,screen,nonacm]{acmart}
\input{preamble}
\AtBeginDocument{%
  }
\usepackage{xifthen}
\newboolean{showcomments}
\setboolean{showcomments}{false} 

\ifthenelse{\boolean{showcomments}}{
	\newcommand{\nbc}[3]{
		{\colorbox{#3}{\bfseries\sffamily\scriptsize\textcolor{white}{#1}}}
		{\textcolor{#3}{\sf\small$\langle$\textit{#2}$\rangle$}}}
}{
	\newcommand{\nbc}[3]{}
	
}

\newcommand{\learnedSolutions}{\textit{learnedSolutions}}
\newcommand{\learnedSpecs}{\textit{learnedSpecs}}
\newcommand{\unrealisableSpecifications}{\textit{unrealisableSpecifications}}
\newcommand{\candidateQueue}{\textit{candidateQueue}}
\newcommand{\newCounterPlaySets}{\textit{newCounterPlaySets}}

\newcommand{\learnAssumptionDegradations}{\textit{learnAssumptionDegradations}}
\newcommand{\learnGuaranteeDegradations}{\textit{learnGuaranteeDegradations}}
\newcommand{\extractCounterPlays}{\textit{extractCounterPlays}}

\newcommand{\isRealisable}{\textit{isRealisable}}
\newcommand{\spec}{\textit{spec}}
\newcommand{\isRedundantCandidate}{\textit{isRedundantCandidate}}
\newcommand{\trivialSolutions}{\textit{trivialSolutions}}

\newcommand{\asm}{\textit{asm}}
\newcommand{\trivialDegradation}{\textit{trivialDegradation}}

\newcommand{\ucsSet}{\textit{ucsSet}}
\newcommand{\minCorrectionSets}{\textit{minCorrectionSets}}

\newcommand{\trivialGD}{\textit{trivialGD}}

\newcommand{\getUnrealisableCores}{\textit{getUnrealisableCores}}
\newcommand{\getMinimalHittingSets}{\textit{getMinimalHittingSets}}

\newcommand{\invariant}{\textit{inv}}
\newcommand{\response}{\textit{res}}
\newcommand{\cs}{\textit{cs}}
\newcommand{\preferredDegradations}{\textit{preferredDegradations}}

\begin{document}

\title{Learning to adapt GR(1) specifications through degradation}
\titlenote{This is a preprint of a manuscript under review at ACM Transactions on Software Engineering and Methodology.}

\author{Tiberiu-Andrei Georgescu}
\email{tibi.geo@ic.ac.uk}
\orcid{0009-0006-7099-4153}
\affiliation{%
  \institution{Imperial College London}
  \city{London}
  \country{United Kingdom}}

\author{Dalal Alrajeh}
\affiliation{%
  \institution{Imperial College London}
  \city{London}
  \country{United Kingdom}}
\email{dalal.alrajeh04@ic.ac.uk}

\author{Sebastian Uchitel}
\affiliation{%
  \institution{Imperial College London}
  \city{London}
  \country{United Kingdom}}
\affiliation{%
  \institution{Universidad de Buenos Aires, CONICET, Universidad de San Andrés}
  \country{Argentina}
}
\email{s.uchitel@ic.ac.uk}

\renewcommand{\shortauthors}{Georgescu et al.}

\begin{abstract}
Reactive synthesis is a powerful tool for generating correct-by-construction controllers from formal specifications. GR(1) is an assume-guarantee specification framework that enables efficient synthesis, allowing synthesised controllers to be used in a wide array of applications. The limitation of such controllers is that, should they encounter environment behaviour unspecified in the assumptions of the specification, the specified system guarantees are no longer ensured. Our work proposes an approach based on oracle-guided inductive synthesis to adapt the specification to be consistent with the observed assumption violation, while degrading system guarantees as little as possible to maintain realisability. Our methodology discovers multiple potential solutions, so we propose a preference criteria, based on the ability of the specification to enable robustness under adaptation. Although our approach is capable of degrading the entire specification, for our case studies we successfully discover degradations that preserve the entire set of original guarantees. 
\end{abstract}



\maketitle

\section{Introduction}
\input{sections/1-introduction/index}

\section{Motivating Example}
\label{sec:minepump}

\input{sections/2-motivating_example/index}

\section{Background}
\label{sec:background}
\input{sections/3-background/3.0-introduction}

\input{sections/3-background/3.1-linear_temporal_logic_over_infinite_and_finite_traces}
\input{sections/3-background/3.2-gr1}

\section{Adaptation under Assumption Violations}
\label{sec:reality_integration}
\input{sections/4-reality_integration/4.0-introduction}
\input{sections/4-reality_integration/4.1-rint_problem_definition}
\input{sections/4-reality_integration/4.2-optimality_of_rint_solutions}

\input{sections/4-reality_integration/4.3-trivial_solution}
\input{sections/4-reality_integration/4.4-motivating_example}

\section{Learning Degradations}
\label{sec:learning_adaptations}

\input{sections/5-learning_adaptations/5.0-introduction}
\input{sections/5-learning_adaptations/5.1-syntactic_degradation_methods}
\input{sections/5-learning_adaptations/5.2-finer_grained_rint_solutions}
\input{sections/5-learning_adaptations/5.3-learning_degradations_for_reality_integration}


\section{Evaluation}
\label{sec:evaluation}

\input{sections/6-evaluation/6.0-introduction}
\input{sections/6-evaluation/6.1-methodology_of_evaluation}

\input{sections/6-evaluation/6.2-experiments}

\input{sections/6-evaluation/6.3-discussion}
\input{sections/6-evaluation/6.z-conclusion}

\section{Related Work}
\label{sec:related_work}
\input{sections/7-related_work/index}

\section{Conclusion}
\label{sec:conclusion}

\input{sections/8-conclusion/index}
\begin{acks}
    This work was supported by the UK Research and Innovation [grant number EP/S023356/1] in the UKRI Centre for Doctoral Training in Safe and Trusted Artificial Intelligence (\url{www.safeandtrustedai.org}).
\end{acks}

\bibliographystyle{ACM-Reference-Format}
\bibliography{index}

\clearpage

\appendix
\pagenumbering{Roman} 
\renewcommand{\thesection}{\Alph{section}} 
\section{Proofs}
\label{apx:sec:proofs}
    \input{appendices/A-proofs/A.1-counter_example_disjunction_of_assumptions}

\input{appendices/A-proofs/A.2-proof_of_soundness_of_baseline_algorithm}
    \input{appendices/A-proofs/A.3-finiteness_of_semantic_equivalence_classes}

\input{appendices/A-proofs/A.4-soundness_of_algorithm}

\section{Code}
    \input{appendices/B-code/B.1-helper_functions}

\section{Case Studies}
\label{apx:sec:case_studies}
    \input{appendices/C-case_studies/C.1-size_of_specifications}
    \input{appendices/C-case_studies/C.2-specifications_and_violation_traces}
    \input{appendices/C-case_studies/C.3-trivial_solutions}
    \input{appendices/C-case_studies/C.4-preferred_learned_solutions}
\end{document}

%% file: preamble.tex
\usepackage{algorithm}
\usepackage{algpseudocode}
\usepackage{aliascnt}
\usepackage{enumitem}
\usepackage{fancyvrb}
\usepackage{comment}
\usepackage{blindtext}
\usepackage{lipsum}
\usepackage{tfrupee}

\usepackage{graphicx}
\usepackage{xcolor}
\usepackage{svg}
\usepackage{tabularx}
\newcolumntype{L}{>{\raggedright\arraybackslash}X}
\newcolumntype{C}{>{\centering\arraybackslash}X}
\newcolumntype{R}{>{\raggedleft\arraybackslash}X}

\definecolor{codegreen}{rgb}{0,0.6,0}
\definecolor{codegray}{rgb}{0.5,0.5,0.5}
\definecolor{codepurple}{rgb}{0.58,0,0.82}
\definecolor{backcolour}{rgb}{0.95,0.95,0.92}

\definecolor{Teal}{RGB}{0,128,128}
\definecolor{NewBlue1}{RGB}{4,100,226}
\definecolor{NiceBlue}{RGB}{63,104,132}
\definecolor{DarkRed}{RGB}{14,53,59}
\definecolor{NewBlue2}{RGB}{62,100,125}
\definecolor{NewBlue3}{RGB}{44,100,128}

\usepackage{listings}
\usepackage{courier}
\lstdefinestyle{mystyle}{
  backgroundcolor=\color{backcolour},
  commentstyle=\color{codegreen},
  keywordstyle=\color{magenta},
  numberstyle=\tiny\color{codegray},
  stringstyle=\color{codepurple},
  basicstyle=\ttfamily\footnotesize,
  breaklines=true,
  captionpos=b,
  keepspaces=true,
  numbers=left,
  numbersep=5pt,
  showspaces=false,
  showstringspaces=false,
  showtabs=false,
  tabsize=2,
  xleftmargin=0.5cm,
  xrightmargin=-0.8cm,
  frame=lr,
  framerule=0pt
}
\lstdefinelanguage{Spectra}{
  morekeywords=[1]{spec,module},
  morekeywords=[2]{env, asm, assumption},
  morekeywords=[3]{sys, gar, guarantee},
  morekeywords=[4]{boolean, ini, alw, alwEv},
  sensitive=true,
  morecomment=[l]{//},
}

\lstdefinestyle{SpectraStyle}{
  language=Spectra,
  basicstyle=\ttfamily\small,
  keywordstyle=[1]\bfseries\color{blue},
  keywordstyle=[2]\bfseries\color{red},
  keywordstyle=[3]\bfseries\color{green!60!black},
  keywordstyle=[4]\bfseries\color{blue},
  commentstyle=\color{gray},
  numbers=left,
  numberstyle=\tiny,
  stepnumber=1,
  numbersep=8pt,
  frame=single,
  framesep=6pt,
  frame=single,
  framerule=1pt,
  backgroundcolor=\color{white},
  title=\strut Spectra,
  captionpos=t,
  showstringspaces=false,
  breaklines=true,
}

\newcommand\todo[1]{\textcolor{red}{\textit{TG TODO: #1}}}

\usepackage{amsmath}
\usepackage{amssymb}
\usepackage{mathtools}
\usepackage{bbold}
\usepackage{cancel}
\usepackage{amsthm}        
\newtheorem{theorem}{Theorem}[section]

\newtheorem{lemma}[theorem]{Lemma}
\newtheorem{definition}[theorem]{Definition}
\newtheorem{axiom}[theorem]{Axiom}

\newtheorem{example}[theorem]{Example}

\usepackage{tikz}
\usetikzlibrary{positioning,calc}

\usepackage{multirow}
\usepackage{booktabs}
\usepackage{colortbl}
\usepackage{xcolor}
\usepackage{amssymb}

\usepackage[normalem]{ulem}

%% file: sections/1-introduction/index.tex
Formal synthesis from temporal-logic specifications provides a rigorous approach for constructing reactive systems that are correct by construction with respect to their formally specified requirements. Consequently, stakeholders obtain stronger assurance that the generated controller satisfies the intended behavioural properties encoded within the specification~\cite{pnueli1989synthesis}. However, synthesising controllers from unrestricted Linear Temporal Logic (LTL) specifications is doubly exponential in complexity, limiting practicality for many real-world applications. Piterman et al.~\cite{piterman2006synthesis} introduced a restricted fragment of LTL known as Generalised Reactivity of rank~1 (GR(1) for short), enabling polynomial-time synthesis of reactive controllers and substantially improving the practical applicability of synthesis techniques. 

In GR(1) specifications, system behaviour is decomposed into two components: \textit{assumptions}, which characterize the expected behaviour of the \textit{environment}, and \textit{guarantees}, which characterize the required behaviour of the \textit{system}. GR(1) therefore follows an assume--guarantee structure, where realisability holds if and only if there exists a controller that can guarantee satisfaction of its specified properties whenever the environmental assumptions remain satisfied.

The implication above exposes a fundamental limitation of assume--guarantee specifications: once an environmental assumption is violated at runtime, the controller may produce arbitrary behaviour, as the correctness criterion the controller was synthesised for only applies when assumptions hold. Although techniques exist for improving the robustness of synthesized controllers under assumption violations~\cite{bloem2011specification}, such approaches do not generally provide precise control over which guarantees may degrade, nor the extent to which system behaviour may diverge from the intended specification. 

Classical work in reactive synthesis has emphasized that the validity of assume--guarantee reasoning fundamentally depends on the correctness and adequacy of the environmental assumptions themselves~\cite{zave1997four}. In practice, however, constructing sufficiently accurate and complete models of complex operational environments remains a significant challenge. Consequently, much of the existing literature focuses on specifying only those aspects of environmental behaviour deemed necessary for system synthesis or verification~\cite{ma2023using}. While practically motivated, this approach introduces the risk that assumptions are engineered primarily to support controller realisability rather than to faithfully characterize the true operational environment. As a result, synthesized controllers may exhibit formally correct behaviour with respect to the specification while remaining vulnerable to mismatches between assumed and actual runtime conditions.

In this work, we consider the problem of adapting reactive specifications under assumption violations, a setting commonly addressed in the literature through forms of specification degradation or repair. Intuitively, this problem captures how a reactive system should evolve when behaviours observed during deployment reveal environmental conditions that were not captured by the original specification. Rather than viewing adaptation as the synthesis of a single repaired specification, we instead consider the broader space of admissible adaptations obtained by integrating newly observed environmental behaviours into the environment assumptions and, where necessary, weakening or modifying the corresponding system guarantees to restore realisability.

Recent work by Buckworth et al.~\cite{buckworth2023adapting} addresses this challenge through the adaptation of environmental assumptions based on observed runtime behaviour. Starting from a controller synthesized from a GR(1) specification, they were able to successfully leverage inductive learning to weaken the original set of assumptions to make it consistent with  an environment violation, and use an oracle-guided inductive synthesis (OGIS) like loop \cite{jha2017theory} to converge to a realisable set of guarantees. 
%
However, the approach proposed by Buckworth et al.~\cite{buckworth2023adapting} produces a single adapted specification without explicitly considering the broader space of alternative adaptations that may also restore realisability. Moreover, the selected adaptation may not necessarily correspond to the solution that is closest to the original specification, nor the one that is most robust to future assumption violations. This limitation is closely related to challenges previously identified in robustness-oriented synthesis approaches~\cite{bloem2011specification}.

A central contribution of this paper is the formal characterization of the adaptation space and the development of principled criteria for evaluating alternative solutions. In particular, we investigate which adaptations should be preferred when multiple realisable solutions exist, including considerations such as preservation of the original specification intent, retention of critical guarantees, and robustness to future assumption violations. From this perspective, we develop a set of axioms defining preferred adaptations and use them to develop a learning-based algorithm capable of exploring a rich space of candidate adaptations. The proposed approach supports the discovery and comparison of multiple realisable alternatives while incorporating heuristics designed to prioritize preferred solutions. Experimental evaluation across multiple case studies demonstrates that the learning-based approach consistently outperforms an adaptation that operates by specification reduction,  dropping violated assumptions and non-realisable guarantees.


We summarize the contributions of this work as follows:
\begin{enumerate}
    \item We formalize the problem of adapting reactive specifications under observed mismatches between assumed and actual environmental behaviour.

    \item We introduce a formal preference framework for reasoning about alternative adaptation solutions, including criteria related to specification preservation and guarantee retention.

    \item We develop an adaptation algorithm based on specification reduction, and formally show that it generates realisable specifications consistent with observed environment behaviour.

    \item We develop an inductive learning-based adaptation algorithm that weakens specification formulas to explore a richer space of realisable adaptations, and formally show that the generated specifications are consistent with environment observations.

    \item We evaluate and analyze the proposed approaches across five case studies from the literature, examining the changes, tradeoffs, and preference of the resulting adapted specifications.
\end{enumerate}

Section~\ref{sec:minepump} introduces a motivating example based on the Arbiter case study \cite{piterman2006synthesis}. Section~\ref{sec:background} presents the formal background necessary for the remainder of the paper. In Section~\ref{sec:reality_integration}, we formally define the adaptation problem, introduce a preference framework for reasoning about alternative adaptations, and present a baseline algorithm based on specification-reduction operations for discovering realisable solutions. Section~\ref{sec:learning_adaptations} then considers formula-level degradation as a mechanism for solving the adaptation problem and presents an inductive learning-based adaptation algorithm for generating candidate specifications. Section~\ref{sec:evaluation} evaluates the proposed learning-based approach against a baseline adaptation algorithm, including comparative analysis of the resulting specifications and their logical relationships. Section~\ref{sec:related_work} reviews related work and situates our work on adaptation within the broader literature on reactive synthesis, specification adaptation, and robustness. Finally, Section~\ref{sec:conclusion} concludes with a discussion of current limitations, broader implications, and directions for future work.

%% file: sections/2-motivating_example/index.tex
We motivate the problem of specification adaptation under assumption violations using the Arbiter case study, discussed in \cite{piterman2006synthesis}, and further leveraged by subsequent literature on adaptation \cite{alur2013counter,cavezza2020minimal,buckworth2023adapting}. 


The arbiter constitutes the system to be synthesized, and controls access to an ARM Advanced High-Performance Bus (AHB), an on-chip communication protocol. The environment consists of multiple masters together with a third party (e.g. a user), although we simplify the case study to two masters only. Each master may request access to the resource by flipping the environment variable $r1$ or $r2$, representing their respective request switch. In response to this behaviour, the arbiter needs to eventually grant access to each requesting master by flipping the system variable $g1$ or $g2$, representing their respective grant switch. The arbiter is bound to never grant access to this shared resource to both masters at the same time, even if they make their requests concurrently. Finally, the third party has the power to lock the arbiter by flipping the environment variable $a$ off, which disallows the arbiter from granting access to the bus.

We model the Arbiter behaviour as a GR(1) specification, shown in Table~\ref{tab:arbiter}. In addition to the behaviour described earlier,  we consider the addition  of a further assumption  to ensure realisability of the specification:  the arbiter is always active. Consequently, the resulting system is realisable according to the  realisability condition of GR(1) presented in~\cite{bloem2012synthesis}. Therefore, an implementation can be synthesized that satisfies the specified guarantees.

\input{code/arbiter_spec}

Suppose, however, that during deployment,  the system exhibits an unexpected behaviour, in which the arbiter becomes inactive. A trace demonstrating a violation to the introduced environment assumption is one in which the environment deactivates the bus in the first time point. Consequently, the synthesized controller is no longer obliged to satisfy the specified guarantees. Although the controller may continue to exhibit operationally robust behaviour~\cite{bloem2014synthesizing}, there is no guarantee that the guarantees it preserves correspond to the intended or safety-critical system behaviour.

Prior work on specification repair~\cite{buckworth2023adapting} demonstrated that assumption-violating behaviours can be incorporated into a realisable specification through automated adaptation techniques. In particular, the arbiter specification in Table~\ref{tab:arbiter} was used as one of their case studies. Applying their repair method to this specification yields a repaired solution in which $assumption\_1$ and $guarantee\_3$ are both weakened to $\top$. The resulting specification is realisable while preserving a subset of the system's original guarantees.

We note, however, that this adaptation comes at a  cost. Most notably, the repaired specification completely disregards the environment variable $a$. Although this variable remains observable to the controller, it no longer influences the specification and is therefore ignored during synthesis. One might argue that the original assumption requiring the arbiter to remain continuously active was overly restrictive and ultimately responsible for the assumption violation. However, as we demonstrate shortly, this does not justify eliminating all dependence on $a$, as doing so discards behaviour that remains both meaningful and exploitable.

The problem above has two main causes:
\begin{enumerate}
    \item Premature convergence to the first feasible intermediate solution encountered; and
    \item An insufficiently expressive solution space, which restricts the exploration of alternative repairs.
\end{enumerate}

Addressing either of these limitations would increase the likelihood of identifying repairs that preserve a greater portion of the original specification. For example, the repair framework of~\cite{buckworth2023adapting} admits an alternative repair in which $assumption\_1$ is weakened to $\Box ((g1 \lor g2)\rightarrow a)$. Under this assumption, the arbiter may be deactivated only when neither requesting master has been granted access to the shared resource. Consequently, the specification remains realisable without weakening any of the system guarantees. However, if this repair, or one with similar properties, is not explored, the subsequent guarantee degradation procedure may unnecessarily weaken the specification despite the existence of repairs that preserve all original guarantees.

Alternatively, prior work does not consider degrading assumptions to weaker temporal patterns, such as justice or response properties. In the arbiter example, observing that the arbiter may become inactive need not imply that it can remain inactive indefinitely. Instead, $assumption\_1$ could be weakened from an invariant to the justice property $\Box\lozenge a$, requiring only that the arbiter eventually becomes active again. Although this repair lies outside the search space explored by prior work, it likewise restores realisability without requiring any degradation of the system guarantees.

The discussion above suggests that some repairs are inherently more desirable than others. While the preference is straightforward when comparing a repair that preserves all system guarantees with one that preserves only a subset, the duality of assume--guarantee specifications makes such comparisons considerably less clear in the general case. Existing repair approaches \cite{buckworth2023adapting,zhang2025behavioral} provide limited guidance on how alternative adaptations should be evaluated when both assumptions and guarantees may be modified. We therefore investigate the structure of the adaptation space to identify the properties that characterize desirable specification repairs. Building on these properties, we develop methods for systematically discovering preferred realisable adaptations from a broader space of candidate specifications.

More fundamentally, the observed violation exposes the need to reason about the quality of specification repairs, rather than merely their existence. Although multiple realisable adaptations may satisfy the observed environmental behaviour, they need not preserve the original specification equally well. This motivates the need for principled criteria for comparing candidate repairs and identifying those that minimise unnecessary degradation of the specification.

%% file: code/arbiter_spec.tex
\begin{table}[h]
\centering
\begin{tabular}{|p{2.5cm}|p{7cm}|p{3.5cm}|}
\hline
\textbf{Type \& ID} & \textbf{Informal requirement} & \textbf{GR(1) formalization} \\
\hline
guarantee\_1 & If master 1 requests to start communicating with the bus, access is eventually granted to master 1 & 
$\Box(\textit{r1}\rightarrow \lozenge\textit{g1})$\\
\hline
guarantee\_2 & If master 2 requests to start communicating with the bus, access is eventually granted to master 2& 
$\Box(\textit{r2}\rightarrow \lozenge\textit{g2})$ \\
\hline
guarantee\_3 & If inactive, the arbiter cannot give access to the bus to any master & 
$\Box(\neg\textit{a} \rightarrow (\neg\textit{g1}\land  \neg \textit{g2}))$ \\
\hline
guarantee\_4 & At most one master may be granted access to the bus at a time & 
$\Box(\neg\textit{g1} \lor \neg\textit{g2})$ \\
\hline
assumption\_1 & The bus is never locked & $\Box(a)$ \\
\hline
\end{tabular}
\caption{GR(1) specification for the arbiter example}
\label{tab:arbiter}
\end{table}

%% file: sections/3-background/3.0-introduction.tex
We present the theoretical background necessary for the remainder of this work. The preliminaries focus on the temporal-logic and reactive-synthesis foundations underlying our approach, including Linear Temporal Logic, GR(1) specifications, and notions of specification weakness, covered in Sections~\ref{subsec:linear_temporal_logic} and \ref{subsec:gr1}.

%% file: sections/3-background/3.1-linear_temporal_logic_over_infinite_and_finite_traces.tex
\subsection{Linear Temporal Logic over Infinite and Finite Traces}
\label{subsec:linear_temporal_logic}

Temporal logics model system evolution over time using Boolean variables. Let $\mathcal{V}$ denote a finite set of propositional variables capturing relevant abstractions of the system at a given instant in time. Using this set of variables, we define a state as follows:

\begin{definition}[State]
Given a finite set of Boolean variables $\mathcal{V}$, a state is a valuation of the variables in $\mathcal{V}$ at a particular instant in time. Equivalently, a state can be represented as a subset $p \subseteq \mathcal{V}$ containing exactly those variables assigned the value \textit{true}. We denote the set of all states over $\mathcal{V}$ by $S = 2^{\mathcal{V}}$.

\end{definition}


This representation captures the system state at a particular instant in time. In discrete-time settings, system behavior can then be modeled as a sequence of consecutive states, referred to as a trace:


\begin{definition}[Trace]
Given a set of Boolean variables $\mathcal{V}$, a trace is an ordered sequence of states over $\mathcal{V}$. A finite trace is an element of $(2^{\mathcal{V}})^*$, while an infinite trace is an element of $(2^{\mathcal{V}})^\omega$.
\end{definition}

Infinite traces correspond equivalently to $\omega$-words in the $\omega$-language literature. Throughout this work, we use the terms \textit{trace} and \textit{word} interchangeably. The notion is also equivalent to the concept of a \textit{run} in automata theory.


Linear Temporal Logic (LTL), introduced in~\cite{pnueli1977temporal}, provides a formal language for specifying and reasoning about temporal properties of execution traces.


\begin{definition}[Linear Temporal Logic]
Linear Temporal Logic (LTL) is a formal language for specifying temporal properties over reactive systems. Given a finite set of Boolean variables $\mathcal{V}$, LTL formulas are inductively generated by the following grammar:
\begin{equation*}
\varphi ::= \top 
\mid \bot 
\mid p 
\mid \neg \varphi 
\mid (\varphi \land \varphi)
\mid \bigcirc\,\varphi
\mid (\varphi \ U \ \varphi)
\end{equation*}
where $p \in \mathcal{V}$ is a propositional variable, $\top$ and $\bot$ denote the Boolean constants \textit{true} and \textit{false}, respectively, $\neg$ and $\land$ denote negation and conjunction, $\bigcirc$ denotes the \emph{next} temporal operator, and $U$ denotes the \emph{strong until} temporal operator.
Additional operators can be derived from this grammar, including $
\varphi \lor \psi \equiv \neg(\neg\varphi \land \neg\psi)
$, 
$\varphi \rightarrow \psi \equiv \neg\varphi \lor \psi $, 
eventually:
$
\lozenge\,\varphi \equiv \top \ U \ \varphi,
$
and globally:
$
\Box\,\varphi \equiv \neg \lozenge\,\neg\varphi.
$
\end{definition}

LTL formulas specify temporal properties over execution traces and are interpreted over infinite sequences of states. Let $\tau = s_0 s_1 s_2 \ldots$ denote an infinite trace over $2^{\mathcal{V}}$, and let $(\tau,i)\models\varphi$ denote that the formula $\varphi$ is satisfied at position $i \geq 0$ of the trace. The semantics of LTL formulas are defined inductively as follows:
\begin{align*}
(\tau, i) &\models p 
& \text{iff} & \quad p \in s_i
\qquad \text{for } p \in \mathcal{V} \\[2mm]
(\tau, i) &\models \neg \varphi
& \text{iff} & \quad (\tau, i) \not\models \varphi \\[2mm]
(\tau, i) &\models \varphi \land \psi
& \text{iff} & \quad (\tau, i) \models \varphi
\ \land \
(\tau, i) \models \psi \\[2mm]
(\tau, i) &\models \bigcirc\,\varphi
& \text{iff} & \quad (\tau, i+1) \models \varphi \\[2mm]
(\tau, i) &\models \varphi \ U \ \psi
& \text{iff} & \quad
\exists k \geq i.\;
(\tau,k)\models\psi
\ \land \
\forall j \in [i,k).\;
(\tau,j)\models\varphi
\end{align*}

For notational convenience, we additionally consider the past-time temporal operators \emph{yesterday} ($Y$) and \emph{historically} ($H$). The operator $Y\,\varphi$ specifies that $\varphi$ held at the previous position in the trace, while $H\,\varphi$ specifies that $\varphi$ has held at all preceding positions. Their semantics are defined as follows:
\begin{align*}
(\tau, i) \models Y\,\varphi
& \iff
(i > 0) \land (\tau, i-1) \models \varphi \\[2mm]
(\tau, i) \models H\,\varphi
& \iff
\forall j < i.\; (\tau,j)\models\varphi
\end{align*}


Finally, our setting additionally requires reasoning over finite traces. To this end, we use the following weak temporal operators from~\cite{de2014reasoning}. The \emph{weak next} operator is defined as
\[
X^\mathcal{W}\varphi \equiv \neg X \neg \varphi,
\]
and specifies that $\varphi$ holds at the next position in the trace, or vacuously when no successor position exists. Similarly, the \emph{weak until} operator is defined as
\[
\varphi_1 \ W \ \varphi_2 \equiv
(\varphi_1 \ U \ \varphi_2)
\lor
G\,\varphi_1,
\]
allowing $\varphi_1$ to hold indefinitely in the event that $\varphi_2$ never becomes true. Using weak until, we further define the \emph{weak eventually} operator:
\[
F^\mathcal{W}\varphi \equiv \top \ W \ \varphi,
\]
which expresses that $\varphi$ may eventually hold, including beyond the finite horizon of a trace.

An LTL formula induces a set of satisfying traces, commonly referred to as its language.



\begin{definition}[$\omega$-Language of an LTL Formula]
Let $\varphi$ be an LTL formula over a finite set of Boolean variables $\mathcal{V}$. The $\omega$-language induced by $\varphi$ is
\[
\mathcal{L}(\varphi)
=
\{
\tau \in (2^{\mathcal{V}})^\omega
\mid
\tau \models \varphi
\}.
\]
That is, $\mathcal{L}(\varphi)$ contains exactly the infinite traces that satisfy $\varphi$. Throughout this work, we use the term \emph{language} to refer to $\omega$-languages whenever no ambiguity arises.
\end{definition}


The language induced by an LTL formula captures the set of behaviours permitted by the specification. This naturally induces a notion of relative permissiveness between formulas: a formula is considered weaker if it admits a larger set of traces.


\begin{definition}[Weakness Relation over LTL Formulas]
\label{def:weakness_ltl}
Let $\varphi_1$ and $\varphi_2$ be LTL formulas over the same set of Boolean variables $\mathcal{V}$.

\begin{itemize}
    \item $\varphi_1 \equiv \varphi_2$ iff $\mathcal{L}(\varphi_1)=\mathcal{L}(\varphi_2)$.

    \item $\varphi_1 \succeq \varphi_2$ iff $\mathcal{L}(\varphi_1)\supseteq \mathcal{L}(\varphi_2)$. In this case, $\varphi_1$ is said to be \emph{weaker than or equivalent to} $\varphi_2$.

    \item $\varphi_1 \succ \varphi_2$ iff $\mathcal{L}(\varphi_1)\supset \mathcal{L}(\varphi_2)$. In this case, $\varphi_1$ is said to be \emph{strictly weaker than} $\varphi_2$.

    \item $\varphi_1 \preceq \varphi_2$ iff $\mathcal{L}(\varphi_1)\subseteq \mathcal{L}(\varphi_2)$. In this case, $\varphi_1$ is said to be \emph{stronger than or equivalent to} $\varphi_2$.

    \item $\varphi_1 \prec \varphi_2$ iff $\mathcal{L}(\varphi_1)\subset \mathcal{L}(\varphi_2)$. In this case, $\varphi_1$ is said to be \emph{strictly stronger than} $\varphi_2$.

\end{itemize}


The relation $\preceq$ is reflexive, antisymmetric, and transitive, and therefore defines a partial order over equivalence classes of LTL formulas under $\equiv$.
\end{definition}

Given a partially ordered set, one may identify elements that are minimal or maximal with respect to the underlying order relation.

\begin{definition}[Minimal and Maximal Elements]
Let $(S,\preceq)$ be a partially ordered set.

\begin{itemize}
    \item An element $e \in S$ is \emph{maximal} if there does not exist an element $e' \in S$ such that $e \prec e'$.

    \item An element $e \in S$ is \emph{minimal} if there does not exist an element $e' \in S$ such that $e' \prec e$.
\end{itemize}
\end{definition}

In the remainder of this paper, the terms \emph{minimal} and \emph{maximal} are used with respect to the partial order over LTL formulas introduced in Definition~\ref{def:weakness_ltl}, restricted to a given set of formulas under consideration.

%% file: sections/3-background/3.2-gr1.tex
\subsection{Generalised Reactivity of Rank 1}
\label{subsec:gr1}
Although LTL is highly expressive and can capture a broad range of temporal properties, including safety and liveness requirements, reactive synthesis from arbitrary LTL specifications is doubly exponential in the size of the formula~\cite{pnueli1989synthesis}. Consequently, practical synthesis approaches often restrict attention to computationally tractable fragments of LTL. One of the most widely adopted fragments is Generalised Reactivity of rank~1~\cite{bloem2012synthesis}, for which symbolic synthesis can be performed in polynomial time with respect to the size of the underlying game structure. In the remainder of this work, we restrict our attention to GR(1) specifications.

We present the definition of GR(1) specifications following~\cite{bloem2012synthesis}, while adopting the modular environment--system decomposition style of~\cite{maoz2020inherent}.

\begin{definition}[GR(1) Specification]
\label{def:gr1-spec}
Let $\mathcal{V}$ be a finite set of Boolean variables partitioned into environment variables $\mathcal{X}$ and system variables $\mathcal{Y}$, such that $\mathcal{X}\cap\mathcal{Y}=\emptyset$ and $\mathcal{X}\cup\mathcal{Y}=\mathcal{V}$.

A GR(1) specification is a tuple
$
\langle \mathcal{E}, \mathcal{S} \rangle,
$
where
$
\mathcal{E} =
\langle
\theta^e,
\rho^e,
\{J_1^e,\dots,J_m^e\}
\rangle
$
and
$
\mathcal{S} =
\langle
\theta^s,
\rho^s,
\{J_1^s,\dots,J_n^s\}
\rangle
$
denote the environment assumptions and system guarantees, respectively.

The components of the specification are defined as follows:

\begin{itemize}
    \item $\theta^e$ and $\theta^s$ are Boolean formulas representing the initial conditions of the environment and system, respectively. The formula $\theta^e$ is defined over $\mathcal{X}$, while $\theta^s$ is defined over $\mathcal{Y}$.

    \item $\rho^e$ and $\rho^s$ are the transition relations of the environment and system, respectively. The environment transition relation $\rho^e$ constrains transitions over current-state variables in $\mathcal{V}$ and next-state environment variables in $\mathcal{X}'$. The system transition relation $\rho^s$ constrains transitions over current-state variables in $\mathcal{V}$ and next-state variables in $\mathcal{X}' \cup \mathcal{Y}'$. These formulas specify the admissible evolution of the environment and system at each time step and are commonly referred to as \emph{safety constraints} or \emph{invariants}.

    \item $J_i^e$ and $J_j^s$ are liveness conditions over $\mathcal{V}$. The environment assumptions $J_i^e$ specify properties that the environment is required to satisfy infinitely often, while the system guarantees $J_j^s$ specify properties that the system must ensure infinitely often.
\end{itemize}
\end{definition}

A key advantage of the GR(1) fragment is that reactive synthesis can be performed in polynomial time whenever the specification is realisable. In this work, we adopt the notion of \emph{strict realisability}, defined below.

\begin{definition}[GR(1) Strict Realisability~\cite{bloem2012synthesis}]
\label{def:gr1-strict-realisability}
A GR(1) specification $\langle \mathcal{E},\mathcal{S}\rangle$ is said to be \emph{strictly realisable} if the following LTL formula is realisable:
\begin{align*}
    \phi :=\;&
    \Big(\theta^e \rightarrow \theta^s\Big)\ \land \nonumber\\
    &\Big(\theta^e \rightarrow \square\big(\mathbf{H}\rho^e \rightarrow \rho^s\big)\Big)\ \land \nonumber\\
    &\Big(\theta^e \land \square\rho^e \rightarrow \Big(
    \bigwedge_{i=1}^{m}\square\lozenge J_i^e
    \rightarrow
    \bigwedge_{j=1}^{n}\square\lozenge J_j^s
    \Big)\Big).
\end{align*}

Intuitively, the three conjuncts capture the following requirements:

\begin{itemize}
    \item If the environment satisfies its initial condition $\theta^e$, then the system must satisfy its initial condition $\theta^s$.

    \item Assuming the environment satisfies its transition relation $\rho^e$ up to the current time step, the system must satisfy its transition relation $\rho^s$ at the current step.

    \item If the environment satisfies its initial condition, continuously respects its transition relation, and satisfies all liveness assumptions $J_i^e$ infinitely often, then the system must satisfy all liveness guarantees $J_j^s$ infinitely often.
\end{itemize}
\end{definition}

Throughout this work, and in particular in the context of the Spectra synthesiser, we frequently reason about individual formulas that constitute a GR(1) specification. In Spectra, the environment assumptions and system guarantees are typically represented as collections of separate constraints rather than as monolithic formulas. For example, the system guarantees may contain multiple invariant constraints, such as $\varphi_1=\mathbf{G}\rho_1^s$ and $\varphi_2=\mathbf{G}\rho_2^s$, which together define the transition relation of the system specification $\mathcal{S}$ in a GR(1) specification $\langle \mathcal{E},\mathcal{S}\rangle$. Analogously, the environment assumptions can be represented as collections of individual formulas.

This representation naturally induces the interpretation that transition relations, such as $\rho^e$ and $\rho^s$, are conjunctions over sets of invariant constraints. The same observation applies to initial conditions $\theta^e$ and $\theta^s$, which may likewise be represented as conjunctions over sets of Boolean assertions.

To facilitate reasoning about individual assumptions and guarantees, we introduce the following syntactic extension of GR(1) specifications, referred to as the \emph{set representation}.

\begin{definition}[Set Representation of a GR(1) Specification]
\label{def:set-representation}
Let $\langle \mathcal{E},\mathcal{S}\rangle$ be a GR(1) specification. The \emph{set representation} of the specification is the tuple $\langle \mathcal{A},\mathcal{G}\rangle$, where $\mathcal{A}$ and $\mathcal{G}$ denote the sets of assumption and guarantee formulas, respectively.

\begin{itemize}
    \item $\mathcal{A}$ is the set of assumption formulas derived from the environment specification\\
    $
    \mathcal{E}=
    \langle
    \theta^e,
    \rho^e,
    \{J_1^e,\dots,J_m^e\}
    \rangle.
    $
    Formally,
    $
    \mathcal{A}
    =
    I^e
    \cup
    \{\mathbf{G}t^e \mid t^e \in T^e\}
    \cup
    \{\mathbf{GF}j^e \mid j^e \in L^e\},
    $
    where:
    \begin{itemize}
        \item $I^e$ is the set of conjuncts of the initial condition $\theta^e$;
        \item $T^e$ is the set of conjuncts of the transition relation $\rho^e$;
        \item $L^e=\{J_1^e,\dots,J_m^e\}$ is the set of environment liveness assumptions.
    \end{itemize}

    \item $\mathcal{G}$ is the set of guarantee formulas derived from the system specification
    $
    \mathcal{S}=
    \langle
    \theta^s,
    \rho^s,
    \{J_1^s,\dots,J_n^s\}
    \rangle.
    $
    Formally,
    $
    \mathcal{G}
    =
    I^s
    \cup
    \{\mathbf{G}t^s \mid t^s \in T^s\}
    \cup
    \{\mathbf{GF}j^s \mid j^s \in L^s\},
    $
    where $I^s$, $T^s$, and $L^s$ are defined analogously to $I^e$, $T^e$, and $L^e$.
\end{itemize}

The partition into sets of initial conditions, transition constraints, and liveness goals is inspired by the abstract syntax representation proposed in~\cite{maoz2020inherent}. However, the broader framework introduced therein is not required for the purposes of this work.

\end{definition}

Intuitively, the assumption set $\mathcal{A}$ captures the expected behaviour of the environment, whereas the guarantee set $\mathcal{G}$ captures the required behaviour of the system. This distinction reflects the asymmetric nature of reactive synthesis: system guarantees correspond to behaviours enforceable by the controller, while environment assumptions characterise expected behaviours of the external environment.

The set representation enables reasoning about GR(1) specifications at the level of individual assumptions and guarantees. In particular, it allows us to compare related specifications obtained through the removal or weakening of specific formulas.

For example, given a specification $\langle \mathcal{A}, \mathcal{G}\rangle$, one may derive a weaker specification by removing an assumption $\varphi \in \mathcal{A}$, yielding
$
\mathcal{A}' = \mathcal{A} \setminus \{\varphi\}.
$
Similarly, one may replace $\varphi$ with a weaker formula $\varphi'$ such that $\varphi' \succeq \varphi$, obtaining
$
\mathcal{A}'' = (\mathcal{A} \setminus \{\varphi\}) \cup \{\varphi'\}.
$

Throughout this work, we primarily use the set representation $\langle \mathcal{A},\mathcal{G}\rangle$ when reasoning about GR(1) specifications, since it facilitates discussion of individual constraints and specification modifications. The tuple representation $\langle \mathcal{E},\mathcal{S}\rangle$ is retained only in situations where the structured decomposition into initial conditions, transition relations, and liveness goals improves clarity or precision.

This style of representation has also been employed in prior work on specification analysis and realisability debugging~\cite{cimatti2008diagnostic}. 
Moving back to the original representation of GR(1), we present the formal background required to understand the GR(1) controller.

Reactive synthesis for GR(1) specifications is commonly formulated as a two-player game between the environment and the system, where the objective of the system is to satisfy its guarantees under the assumptions imposed on the environment. The corresponding game structure is defined as follows.

\begin{definition}[GR(1) Game~\cite{bloem2012synthesis}]
\label{def:gr1-game}
A GR(1) game is a tuple
$
\langle
\mathcal{X},
\mathcal{Y},
\theta^e,
\theta^s,
\rho^e,
\rho^s,
Win
\rangle,
$
where $\mathcal{X}$ and $\mathcal{Y}$ are the environment and system variables, $\theta^e$ and $\theta^s$ are the initial conditions, and $\rho^e$ and $\rho^s$ are the transition relations, as defined in Definition~\ref{def:gr1-spec}.

The winning condition of the game is given by the LTL formula
$
Win =
\Big(
\bigwedge_{i=1}^{m}\mathbf{GF}J_i^e
\Big)
\rightarrow
\Big(
\bigwedge_{j=1}^{n}\mathbf{GF}J_j^s
\Big),
$
which states that if the environment satisfies all of its liveness assumptions infinitely often, then the system must satisfy all of its liveness guarantees infinitely often.
\end{definition}

\begin{definition}[Game State and Projection]
    The state of a game is any valuation $\sigma\in 2^{\mathcal{V}}$ of the union of input and output variables $\mathcal{V}=\mathcal{X} \cup \mathcal{Y}$. For any subset $\mathcal{V}'\subseteq \mathcal{V}$, we call $\sigma|_{\mathcal{V}'}=\sigma \cap \mathcal{V}'$ the projection of $\sigma$ onto set $\mathcal{V}'$.
\end{definition}

To reason about executions of a GR(1) game, we define game states as valuations over the environment and system variables, together with the standard notion of variable projection.

\begin{definition}[Game State and Projection]
A \emph{game state} is a valuation
$
\sigma \in 2^{\mathcal{V}}
$
over the set of variables
$
\mathcal{V} = \mathcal{X} \cup \mathcal{Y}.
$
For any subset of variables $\mathcal{V}' \subseteq \mathcal{V}$, the \emph{projection} of a state $\sigma$ onto $\mathcal{V}'$ is defined as
$
\sigma|_{\mathcal{V}'}
=
\sigma \cap \mathcal{V}'.
$
That is, the projection retains only the variables in $\mathcal{V}'$ that are assigned the value \textit{true} in $\sigma$.
\end{definition}

Strategies provide a formal representation of how a player selects actions during the execution of a game. In the GR(1) setting, synthesis produces a finite-state strategy for the system player.

\begin{definition}[Finite-State Strategy]
A finite-state strategy is a tuple
$
\langle \Gamma,\gamma_0,f\rangle,
$
where:
\begin{itemize}
    \item $\Gamma$ is a finite set of memory states;
    \item $\gamma_0 \in \Gamma$ is the initial memory state;
    \item $f$ is a partial transition function
    $
    f:\Gamma\times Q \rightharpoonup \Gamma\times \Sigma,
    $
    where $Q$ is the set of game states and $\Sigma$ is the set of actions or outputs available to the player.
\end{itemize}

Given a memory state $\gamma \in \Gamma$ and a game state $q \in Q$, the function $f$ returns an updated memory state $\gamma' \in \Gamma$ together with an output action $o \in \Sigma$.
\end{definition}

When the memory structure $\Gamma$ and initial memory state $\gamma_0$ are clear from context, we refer to a finite-state strategy $\langle \Gamma,\gamma_0,f\rangle$ simply by its transition function $f$.


To relate strategies to executions of a GR(1) game, we next define when an infinite play is consistent with the choices prescribed by a strategy.

\begin{definition}[Play--Strategy Compliance~\cite{bloem2012synthesis}]
Let
$
\langle
\mathcal{X},
\mathcal{Y},
\theta^e,
\theta^s,
\rho^e,
\rho^s,
Win
\rangle
$
be a GR(1) game, and let
$
w=\sigma_0\sigma_1\sigma_2\dots \in (2^\mathcal{V})^\omega
$
be an infinite play.

\begin{itemize}
    \item The play $w$ is said to be \emph{compliant} with an environment strategy
    $
    \langle \Gamma^e,\gamma_0^e,f^e\rangle
    $
    with
    $
    f^e :
    \Gamma^e \times 2^\mathcal{V}
    \rightharpoonup
    \Gamma^e \times 2^\mathcal{X},
    $
    iff there exists a sequence of memory states
    $
    \gamma_0^e,\gamma_1^e,\gamma_2^e,\dots
    $
    such that, for every $i\geq 0$,
    $
    f^e(\gamma_i^e,\sigma_i)
    =
    (\gamma_{i+1}^e,\sigma_{i+1}|_\mathcal{X}).
    $

    \item Similarly, the play $w$ is said to be \emph{compliant} with a system strategy
    $
    \langle \Gamma^s,\gamma_0^s,f^s\rangle
    $
    with
    $
    f^s :
    \Gamma^s \times 2^\mathcal{V} \times 2^\mathcal{X}
    \rightharpoonup
    \Gamma^s \times 2^\mathcal{Y},
    $
    iff there exists a sequence of memory states
    $
    \gamma_0^s,\gamma_1^s,\gamma_2^s,\dots
    $
    such that, for every $i\geq 0$,
    $
    f^s(\gamma_i^s,\sigma_i,\sigma_{i+1}|_\mathcal{X})
    =
    (\gamma_{i+1}^s,\sigma_{i+1}|_\mathcal{Y}).
    $
\end{itemize}
\end{definition}

In the GR(1) setting, the environment and the system induce distinct classes of strategies corresponding to their respective roles in the game.

\begin{definition}[Environment and System Strategies]
Let
$
\langle
\mathcal{X},
\mathcal{Y},
\theta^e,
\theta^s,
\rho^e,
\rho^s,
Win
\rangle
$
be a GR(1) game.

An \emph{environment strategy} is a finite-state strategy whose transition function is of the form
$
f^e :
\Gamma^e \times 2^\mathcal{V}
\rightharpoonup
\Gamma^e \times 2^\mathcal{X},
$
where the strategy selects the next valuation of the environment variables.

A \emph{system strategy} is a finite-state strategy whose transition function is of the form
$
f^s :
\Gamma^s \times 2^\mathcal{V} \times 2^\mathcal{X}
\rightharpoonup
\Gamma^s \times 2^\mathcal{Y},
$
where the strategy selects the next valuation of the system variables after observing the next environment move.
\end{definition}

The asymmetry between the two definitions reflects the reactive nature of the game: the environment acts first at each step, while the system responds by choosing a compatible system action.

The objective of reactive synthesis is to construct a system strategy that guarantees satisfaction of the specification against all admissible behaviours of the environment. Such strategies are referred to as winning strategies.

\begin{definition}[Winning Strategy]
Let $\langle \mathcal{E},\mathcal{S}\rangle$ be a GR(1) specification.

A system strategy $f^s$ is said to be \emph{winning} iff every play compliant with $f^s$ satisfies the specification $\langle \mathcal{E},\mathcal{S}\rangle$.

A GR(1) specification is said to be \emph{realisable} iff there exists a winning strategy for the system. Equivalently, the corresponding GR(1) game is said to be \emph{winning for the system}.
\end{definition}

When a GR(1) specification is unrealisable, the environment can enforce a violation of the specification regardless of the behaviour of the system. The corresponding environment strategy is referred to as a counter-strategy.

\begin{definition}[Counter-Strategy]
\label{def:counter-strategy}
Let $\langle \mathcal{E},\mathcal{S}\rangle$ be a GR(1) specification. If the specification is unrealisable, then the corresponding GR(1) game is said to be \emph{winning for the environment}.

An environment strategy $c$ is called a \emph{counter-strategy} iff, for every system strategy $f^s$, every play $w$ compliant with both $c$ and $f^s$ violates the specification. Equivalently, every such play satisfies at least one of the following conditions:

\begin{itemize}
    \item The environment satisfies its initial condition, but the system violates its initial condition:
    $
    w[0]\models\theta^e
    $
    and
    $
    w[0]\not\models\theta^s.
    $

    \item The environment satisfies its transition relation at every step, but the system violates its transition relation at some step:
    $
    \forall i\geq0.\;(w[i],w[i+1])\models\rho^e
    $
    and
    $
    \exists i\geq0.\;(w[i],w[i+1])\not\models\rho^s.
    $

    \item The environment satisfies all of its liveness assumptions infinitely often, but the system fails to satisfy at least one of its liveness guarantees infinitely often:
    $
    \forall i\in\{1,\dots,m\}.\; w\models \mathbf{GF}J_i^e
    $
    and
    $
    \exists j\in\{1,\dots,n\}.\; w\not\models \mathbf{GF}J_j^s.
    $
\end{itemize}
\end{definition}

Counter-strategies characterise sets of environment behaviours that force the system to violate the specification. Individual executions induced by such strategies are referred to as counter-plays.

\begin{definition}[Counter-Play]
Let $c$ be a counter-strategy for an unrealisable GR(1) specification $\langle \mathcal{E},\mathcal{S}\rangle$. A \emph{counter-play} of $c$ is an infinite play that is compliant with the counter-strategy $c$.
\end{definition}

%% file: sections/4-reality_integration/4.0-introduction.tex
In this section, we formally define the adaptation problem considered throughout the paper. Section~\ref{subsec:rint_problem_definition} introduces the problem setting in which behaviours observed during deployment violate the assumptions of the original specification, requiring the synthesis of an adapted realisable specification that incorporates the newly observed environment behaviour while preserving as much of the original specification intent as possible.

Section~\ref{subsec:in_search_for_an_optimal_solution_to_the_rint_problem} then investigates the problem of comparing alternative adaptations. In particular, we identify desirable properties of adaptation solutions, analyse trade-offs between competing objectives, and derive a set of axiomatic criteria characterising preferred adaptations. We additionally discuss circumstances under which different solutions cannot be meaningfully ordered without additional contextual information or design priorities.

Finally, Section~\ref{subsec:trivial-solution} presents a baseline adaptation procedure based on specification-reduction operations. We formally prove that the proposed algorithm always produces valid solutions to the adaptation problem and satisfies the optimality criteria established earlier in the section.

%% file: sections/4-reality_integration/4.1-rint_problem_definition.tex
\subsection{Problem Definition}
\label{subsec:rint_problem_definition}


Reactive synthesis under assumption violations has previously been studied through approaches based on specification degradation and repair~\cite{buckworth2023adapting}. In this setting, the deployed system encounters behaviours of the environment that are not captured by the assumptions of the original specification, causing the theoretical guarantees of the synthesised controller to no longer apply. The objective of adaptation is therefore to recover a realisable specification that accounts for the newly observed behaviour while preserving as much of the original system intent as possible.

We model such unexpected observations using the notion of a \emph{violation trace}.

\begin{definition}[Violation Trace]
Let $\mathcal{A}$ be a set of environment assumptions over a set of environment variables $\mathcal{X}$. A finite trace
$
\tau\in(2^\mathcal{X})^*
$
is called a \emph{violation trace} iff
$
\tau\not\models\mathcal{A}.
$

When a violation trace is obtained from the execution of a deployed system, we may also refer to it as an \emph{observed behaviour} or \emph{newly observed behaviour}.
\end{definition}

Given a violation trace, the goal of adaptation is to revise the original specification so that the newly observed behaviour becomes admissible while maintaining as much of the original specification structure and system behaviour as possible. This leads to the following problem definition.

\begin{definition}[Adaptation Problem]
\label{def:rca}
Let
$
\varphi=\langle\mathcal{A},\mathcal{G}\rangle
$
be a realisable specification, and let
$
\tau\in(2^\mathcal{X})^*
$
be a violation trace such that
$
\tau\not\models\mathcal{A}.
$

An \emph{adaptation} of the tuple $(\varphi,\tau)$ is a specification
$
\varphi'=\langle\mathcal{A}',\mathcal{G}'\rangle
$
such that:

\begin{enumerate}
    \item $\tau\models\mathcal{A}'$;

    \item $\mathcal{A}\prec\mathcal{A}'$;

    \item $\mathcal{G}\preceq\mathcal{G}'$;

    \item $\varphi'$ is realisable.
\end{enumerate}

We denote by
$
AllSolutions(\varphi,\tau)
$
the set of all specifications satisfying the conditions above. If
$
\varphi'\in AllSolutions(\varphi,\tau),
$
we say that $\varphi'$ is a solution to the adaptation problem induced by $(\varphi,\tau)$.
\end{definition}



Intuitively, Property~1 ensures that the newly observed behaviour is incorporated into the adapted environment model. Property~2 requires the adapted assumptions to be weaker than the original assumptions, thereby preserving all previously admissible environment behaviours while additionally admitting the observed violation trace. Property~3 requires the adapted guarantees to preserve as much of the original system behaviour as possible by ensuring that all previously guaranteed behaviours remain permitted by the adapted specification. Finally, Property~4 ensures that the resulting specification remains realisable and therefore admits a correct-by-construction implementation.


A trivial solution to the adaptation problem is obtained by replacing both the assumptions and guarantees with $\top$, thereby allowing all environment and system behaviours. Although such a specification satisfies the formal conditions above, it discards all meaningful behavioural guarantees and is therefore uninformative in practice. Consequently, the central challenge of adaptation lies not merely in recovering realisability, but in identifying preferred solutions that preserve as much of the original specification intent as possible.


%% file: sections/4-reality_integration/4.2-optimality_of_rint_solutions.tex
\subsection{Preference over Comparable Adaptations}
\label{subsec:in_search_for_an_optimal_solution_to_the_rint_problem}


The adaptation problem introduced in Section~\ref{subsec:rint_problem_definition} admits, in general, multiple valid solutions. Consequently, realisability alone is insufficient for determining which adapted specification should be preferred in practice. In this subsection, we investigate how alternative adaptations may be compared and ordered according to desirable behavioural properties.

Our objective is to derive principled criteria for selecting preferred adaptations while preserving as much of the original specification intent as possible. To this end, we first analyse pairwise comparisons between candidate solutions and study how differences in assumptions and guarantees influence adaptation quality. Building upon this analysis, we derive a collection of preference axioms that induce a partial order over the space of valid adaptations. We refer to the ordering induced between two adaptation solutions as \emph{preference}, formalised below.

\begin{definition}[Specification Preference]
\label{def:preference}
Let $(\varphi,\tau)$ be an adaptation problem, and let
$
\varphi_1,\varphi_2 \in AllSolutions(\varphi,\tau)
$
be two valid adaptation solutions.

We use the following notation:
\begin{itemize}
    \item $\varphi_1 \sqsubset \varphi_2$ to denote that $\varphi_1$ is preferred to $\varphi_2$;

    \item $\varphi_1 \sqsupset \varphi_2$ to denote that $\varphi_2$ is preferred to $\varphi_1$.
\end{itemize}
\end{definition}

The preference criteria developed in this subsection induce a partial order over the space of adaptation solutions, allowing the extraction of preferred solutions from a candidate solution set.

\subsubsection{Reasoning about comparable specifications}

Consider an adaptation problem $(\varphi,\tau)$ together with two semantically distinct adaptation solutions
$
\varphi_1'=\langle\mathcal{A}_1',\mathcal{G}_1'\rangle
$
and
$
\varphi_2'=\langle\mathcal{A}_2',\mathcal{G}_2'\rangle.
$

To compare alternative adaptations, we use the logical weakness relation introduced in Definition~\ref{def:weakness_ltl}. We begin by considering the following representative cases involving logically comparable assumptions and guarantees (incomparable cases are considered  in Section~\ref{subsubsec:dealing_with_uncomparability}):
\begin{enumerate}
    \item $\mathcal{A}_1'\equiv \mathcal{A}_2'$ and $\mathcal{G}_1'\prec \mathcal{G}_2'$;

    \item $\mathcal{A}_1'\prec \mathcal{A}_2'$ and $\mathcal{G}_1'\equiv \mathcal{G}_2'$;

    \item $\mathcal{A}_1'\prec \mathcal{A}_2'$ and $\mathcal{G}_1'\prec \mathcal{G}_2'$;

    \item $\mathcal{A}_1'\succ \mathcal{A}_2'$ and $\mathcal{G}_1'\prec \mathcal{G}_2'$.
\end{enumerate}



When two adaptation solutions differ only in their guarantees, existing work on specification repair and degradation favours the  stronger guarantee~\cite{cavezza2017interpolation,maoz2019symbolic,buckworth2023adapting}. This preference follows naturally from the fact that
$
\mathcal{G}\prec\mathcal{G}_1'\prec\mathcal{G}_2'.
$
Consequently, $\varphi_1'$ preserves more of the original system behaviour encoded in $\mathcal{G}$ and introduces fewer additional behaviours beyond those intended by the original specification. Since stronger guarantees correspond to smaller deviations from the original specification intent, we adopt the following preference axiom.

\begin{axiom}
\label{axiom:preference_1}
Let $(\varphi,\tau)$ be an adaptation problem, and let
$
\varphi_1'=\langle\mathcal{A}_1',\mathcal{G}_1'\rangle
$
and
$
\varphi_2'=\langle\mathcal{A}_2',\mathcal{G}_2'\rangle
$
be two adaptation solutions.
If the assumptions are equivalent and the guarantees of $\varphi_1'$ are stronger than those of $\varphi_2'$, i.e.,
$
\mathcal{A}_1'\equiv\mathcal{A}_2'
$
and
$
\mathcal{G}_1'\prec\mathcal{G}_2',
$
then $\varphi_1'$ is preferred to $\varphi_2'$. 
\end{axiom}

Existing work on assumption refinement typically favours the weakest assumptions capable of preserving a fixed set of guarantees~\cite{cavezza2017interpolation,gaaloul2020mining,shalom2023my,maoz2019symbolic}. Under this interpretation, the preferred adaptation would, therefore, be the one with the weakest assumption set.

However, the preference criteria used in classical assumption refinement do not directly transfer to the adaptation setting considered in this work. Assumption refinement typically seeks to maximise the set of admissible environment behaviours while preserving a fixed set of guarantees. In contrast, our setting begins from an already deployed specification whose assumptions have been violated by newly observed behaviour. The objective is therefore not unrestricted weakening of the environment model, but rather controlled adaptation that incorporates the observed behaviour while preserving as much of the original specification intent as possible.

To illustrate the implications of this distinction, consider extending the comparison setting from two adaptation solutions to three solutions. Let
$
\varphi_1'=\langle\mathcal{A}_1',\mathcal{G}_1'\rangle,
$
$
\varphi_2'=\langle\mathcal{A}_2',\mathcal{G}_2'\rangle,
$
and
$
\varphi_3'=\langle\mathcal{A}_3',\mathcal{G}_3'\rangle
$
be three adaptation solutions satisfying the following properties:
\begin{enumerate}
    \item $\mathcal{A}_2'\succ \mathcal{A}_1'$ and  $\mathcal{A}_3'\succ \mathcal{A}_1'$;
    \item $\mathcal{A}_3'$ and $\mathcal{A}_2'$ are incomparable under the weakness relation;

    \item $\mathcal{G}_3'\equiv\mathcal{G}_1'\equiv\mathcal{G}_2'\equiv\mathcal{G}$.
\end{enumerate}


The implication relations between the assumptions are visualised in Figure~\ref{fig:overfitting_assumption_weakening}. Consider a deployed reactive system whose controller has been synthesised from a specification
$
\varphi=\langle\mathcal{A},\mathcal{G}\rangle.
$
Suppose that, during execution, the system observes a trace
$
\tau
$
such that
$
\tau\not\models\mathcal{A},
$
i.e., the observed behaviour violates the assumptions of the original specification.

In this setting, two important sources of uncertainty arise:
\begin{itemize}
    \item the true dynamics of the environment are generally unknown and may not be fully characterisable;

    \item the observed violation trace $\tau$ does not necessarily characterise the complete set of environment behaviours excluded by $\mathcal{A}$.
\end{itemize}


Existing work on specification learning and repair typically seeks to identify the weakest set of assumptions under which a fixed set of guarantees remains realisable. The underlying objective is to maximise the robustness of the synthesised system by enlarging the set of admissible environment behaviours while preserving the intended guarantees. Intuitively, if the same guarantees can be enforced under weaker assumptions, then the resulting specification is considered more robust, since the system can operate correctly under a broader range of environmental conditions.

While this perspective is well-motivated, it becomes less clear in the adaptation setting considered here. In particular, temporal logic specifications characterise infinite sets of behaviours, and logical weakness induces only a partial order over these sets. Consequently, multiple incomparable assumption sets may exist that are each minimal or maximally permissive with respect to a given guarantee set. In such cases, selecting between incomparable weakenings based solely on permissiveness may admit behaviours unsupported by the observed violation trace. We refer to this phenomenon as \emph{overgeneralisation}.


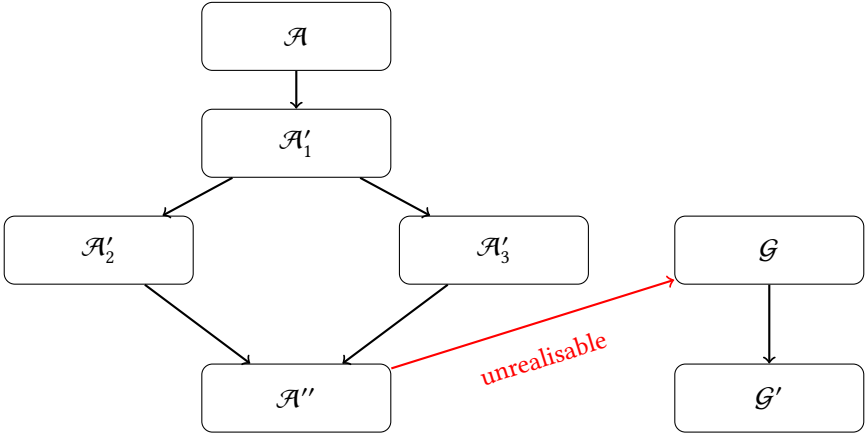
\begin{figure}
\centering
\begin{tikzpicture}[
    node distance=0.5cm and 0.1cm,
    every node/.style={draw, rounded corners, minimum width=2.5cm, minimum height=0.9cm, align=center},
    arr/.style={->, thick},
    box/.style={draw=red, thick},
    unrealisable/.style={->, thick,color=red},
    realisable/.style={->, thick,color=green},
    labelnode/.style={draw=none, rectangle, inner sep=1pt}
]

\node (A) {$\mathcal{A}$};
\node (Ap) [below=of A] {$\mathcal{A}_1'$};

\node (A1pp) [below left=of Ap] {$\mathcal{A}_2'$};
\node (A2pp) [below right=of Ap] {$\mathcal{A}_3'$};

\node (Appp) [below=1.5cm of $(A1pp)!0.5!(A2pp)$] {$\mathcal{A}''$};

\node (G) [right=5cm of $(A1pp)!0.5!(A2pp)$] {$\mathcal{G}$};
\node (Gp) [right=5cm of $(Appp)!0.5!(Appp)$] {$\mathcal{G}'$};

\draw[arr] (A) -- (Ap);
\draw[arr] (Ap) -- (A1pp);
\draw[arr] (Ap) -- (A2pp);
\draw[arr] (A1pp) -- (Appp);
\draw[arr] (A2pp) -- (Appp);

\draw[arr] (G) -- (Gp);

\draw[unrealisable] (Appp) -- node[labelnode, sloped, below, allow upside down] {unrealisable} (G);

\end{tikzpicture}
\caption{Illustration of successive specification weakenings during adaptation. $\mathcal{A}$ and $\mathcal{G}$ denote the original assumption and guarantee sets, respectively. The sets $\mathcal{A}_1'$, $\mathcal{A}_2'$, $\mathcal{A}_3'$, and $\mathcal{A}''$ represent weakened assumption sets derived from $\mathcal{A}$, while $\mathcal{G}'$ denotes a weakened guarantee set derived from $\mathcal{G}$. Directed edges represent the strict weakness relation $\prec$. Nodes with no directed path between them are incomparable under the weakness relation.}
\label{fig:overfitting_assumption_weakening}
\end{figure}

To illustrate the risks of overgeneralisation in the adaptation setting, we return to Figure~\ref{fig:overfitting_assumption_weakening}. Suppose the adaptation problem $(\varphi,\tau)$ admits multiple solutions requiring only assumption weakening to recover realisability. Formally,
$
AllSolutions(\varphi,\tau)
\supseteq
\{\varphi_1',\varphi_2',\varphi_3'\},
$
where
$
\varphi_1'=\langle\mathcal{A}_1',\mathcal{G}\rangle,
$
$
\varphi_2'=\langle\mathcal{A}_2',\mathcal{G}\rangle,
$
and
$
\varphi_3'=\langle\mathcal{A}_3',\mathcal{G}\rangle.
$
The logical relationships between the assumption sets $\mathcal{A}_1'$, $\mathcal{A}_2'$, and $\mathcal{A}_3'$ are illustrated in Figure~\ref{fig:overfitting_assumption_weakening}.

Suppose that, guided solely by the preference for weaker assumptions under equivalent guarantees, the adaptation process selects $\varphi_3'$ (the same argument applies symmetrically to $\varphi_2'$). Consider now the later observation of a second violation trace $\tau_2$ such that
$
\tau_2 \models \mathcal{A}_2'
$
but
$
\tau_2 \not\models \mathcal{A}_3'.
$

In this setting, the previously selected adaptation $\varphi_3'$ no longer accommodates the newly observed behaviour. Consequently, the adaptation process must further weaken the assumptions to some set
$
\mathcal{A}'' \succ \mathcal{A}_3'
$
such that
$
\tau_2 \models \mathcal{A}''.
$
Since $\mathcal{A}_2'$ already admits $\tau_2$, it follows that the new assumption set must additionally satisfy
$
\mathcal{A}'' \succ \mathcal{A}_2'.
$

The difficulty is that the resulting specification
$
\langle\mathcal{A}'',\mathcal{G}\rangle
$
may no longer remain realisable. Recovering realisability may therefore require an additional weakening of the guarantees to some
$
\mathcal{G}' \succ \mathcal{G},
$
yielding a realisable specification
$
\varphi''=\langle\mathcal{A}'',\mathcal{G}'\rangle.
$

Intuitively, this sequence of adaptations illustrates the risk of prematurely selecting overly permissive assumptions based on limited observations. Although $\varphi_3'$ initially appeared preferable under a weakest-assumption criterion, it may ultimately force stronger degradation of guarantees during subsequent adaptations. In contrast, the alternative solution
$
\varphi_2'=\langle\mathcal{A}_2',\mathcal{G}\rangle
$
would already accommodate the second violation trace without requiring any degradation of the guarantees.

Consequently, adaptation procedures that prioritise aggressively weak assumptions based on limited observations are susceptible to overgeneralisation, potentially leading to poorer long-term preservation of system behaviour.

Two alternative adaptation strategies are worth discussing.
First, one could attempt to refine assumptions after additional violations are observed, rather than continuing assumption weakening. Existing approaches to specification refinement achieve this either through reasoning over the current guarantee set~\cite{cavezza2017interpolation,gaaloul2020mining} or through the incorporation of additional behavioural data~\cite{keegan2020control}. However, the former risks overfitting assumptions to the current guarantees, while the latter may require substantially more observational data than is realistically available during deployment.
Second, in the setting of Figure~\ref{fig:overfitting_assumption_weakening}, one may consider the disjunctive assumption set
$
\mathcal{A}_2' \lor \mathcal{A}_3'
$
together with the original guarantees $\mathcal{G}$. Although such a construction may preserve realisability in some cases, the resulting specification does not, in general, remain within the GR(1) fragment. A counterexample is provided in Appendix~\ref{apx:counter_example}.
These observations further motivate a conservative adaptation strategy based on gradual weakening of assumptions while preserving stronger guarantees whenever possible.

We therefore obtain the following preference axiom.

\begin{axiom}
\label{axiom:preference_2}
Let $(\varphi,\tau)$ be an adaptation problem, and let
$
\varphi_1'=\langle\mathcal{A}_1',\mathcal{G}_1'\rangle
$
and
$
\varphi_2'=\langle\mathcal{A}_2',\mathcal{G}_2'\rangle
$
be two adaptation solutions.

If the guarantees are equivalent and the assumptions of $\varphi_1'$ are stronger than those of $\varphi_2'$, i.e.,
$
\mathcal{G}_1'\equiv\mathcal{G}_2'
$
and
$
\mathcal{A}_1'\prec\mathcal{A}_2',
$
then $\varphi_1'$ is preferred to $\varphi_2'$. Formally,
$
(\mathcal{A}_1'\prec\mathcal{A}_2')
\land
(\mathcal{G}_1'\equiv\mathcal{G}_2')
\Rightarrow
\varphi_1' \sqsubset \varphi_2'.
$
\end{axiom}

Axioms~\ref{axiom:preference_1} and~\ref{axiom:preference_2} immediately yield the following combined preference criterion.

\begin{axiom}
\label{axiom:preference_3}
Let $(\varphi,\tau)$ be an adaptation problem, and let
$
\varphi_1'=\langle\mathcal{A}_1',\mathcal{G}_1'\rangle
$
and
$
\varphi_2'=\langle\mathcal{A}_2',\mathcal{G}_2'\rangle
$
be two adaptation solutions.

If both the assumptions and guarantees of $\varphi_1'$ are stronger than those of $\varphi_2'$, i.e.,
$
\mathcal{A}_1'\prec\mathcal{A}_2'
$
and
$
\mathcal{G}_1'\prec\mathcal{G}_2',
$
then $\varphi_1'$ is preferred to $\varphi_2'$. Formally,
$
(\mathcal{A}_1'\prec\mathcal{A}_2')
\land
(\mathcal{G}_1'\prec\mathcal{G}_2')
\Rightarrow
\varphi_1' \sqsubset \varphi_2'.
$
\end{axiom}


Finally, the remaining comparison case can be reduced to the previous preference criteria through the following lemma.

\begin{lemma}
\label{lemma:exists_a_third_solution}
Let $(\varphi,\tau)$ be an adaptation problem. Suppose
$
\varphi_1'=\langle\mathcal{A}_1',\mathcal{G}_1'\rangle
$
and
$
\varphi_2'=\langle\mathcal{A}_2',\mathcal{G}_2'\rangle
$
are adaptation solutions satisfying
$
\mathcal{A}_1'\succ\mathcal{A}_2'
$
and
$
\mathcal{G}_1'\prec\mathcal{G}_2'.
$
Then
$
\varphi_3'=\langle\mathcal{A}_2',\mathcal{G}_1'\rangle
$
is also a valid adaptation solution.
\end{lemma}

\begin{proof}
Since $\varphi_2'$ is a valid adaptation solution,
$
\tau\models\mathcal{A}_2'
$
and
$
\mathcal{A}\prec\mathcal{A}_2'.
$
Moreover, since
$
\mathcal{A}_2'\prec\mathcal{A}_1'
$
and $\varphi_1'$ is realisable, monotonicity of realisability under strengthening of assumptions implies that
$
\langle\mathcal{A}_2',\mathcal{G}_1'\rangle
$
is also realisable.

Finally, because $\varphi_1'$ is a valid adaptation solution,
$
\mathcal{G}\preceq\mathcal{G}_1'.
$
Therefore,
$
\varphi_3'=\langle\mathcal{A}_2',\mathcal{G}_1'\rangle
$
satisfies all conditions of the adaptation problem and is consequently a valid adaptation solution.
\end{proof}

Lemma~\ref{lemma:exists_a_third_solution} reduces the remaining comparison case to Axiom~\ref{axiom:preference_1}, since $\varphi_3'$ and $\varphi_2'$ differ only in their guarantees.

%
%

\subsubsection{Reasoning about Incomparable Adaptations}
\label{subsubsec:dealing_with_uncomparability}

Logical weakness induces only a partial order over specifications. Consequently, not all adaptation solutions can be compared directly using implication-based preference criteria. In particular, distinct adaptations may permit different, non-overlapping classes of behaviours while remaining equally valid solutions to the adaptation problem.

In such cases, preference cannot generally be established without introducing additional domain knowledge, engineering priorities, or external optimisation criteria. We therefore analyse the following representative cases involving incomparable assumptions and guarantees.

Let
$
\varphi_1=\langle\mathcal{A}_1,\mathcal{G}_1\rangle
$
and
$
\varphi_2=\langle\mathcal{A}_2,\mathcal{G}_2\rangle
$
be two solutions to the adaptation problem $(\varphi,\tau)$. We consider the following cases:
\begin{enumerate}
    \item $\mathcal{A}_1$ and $\mathcal{A}_2$ are incomparable, and $\mathcal{G}_1\equiv\mathcal{G}_2$;

    \item $\mathcal{A}_1$ and $\mathcal{A}_2$ are incomparable, and $\mathcal{G}_1\prec\mathcal{G}_2$;

    \item $\mathcal{A}_1\equiv\mathcal{A}_2$, and $\mathcal{G}_1$ and $\mathcal{G}_2$ are incomparable;

    \item $\mathcal{A}_1\prec\mathcal{A}_2$, and $\mathcal{G}_1$ and $\mathcal{G}_2$ are incomparable;

    \item $\mathcal{A}_1$ and $\mathcal{A}_2$ are incomparable, and $\mathcal{G}_1$ and $\mathcal{G}_2$ are incomparable.
\end{enumerate}

Reasoning about incomparable adaptations may require considering additional candidate solutions beyond the original pair under comparison. This is already evident in the first case above, where the existence of a third adaptation directly influences the preference ordering.

To analyse Case~1, recall that solutions to the adaptation problem satisfy $\mathcal{A}\prec\mathcal{A}_1$ and $\mathcal{A}\prec\mathcal{A}_2$. Since logical weakness corresponds to language inclusion, it follows that $\mathcal{A}\preceq(\mathcal{A}_1 \cup \mathcal{A}_2)$, where the union operator denotes conjunction over the combined assumption sets.

Furthermore, because $\tau\models\mathcal{A}_1$ and $\tau\models\mathcal{A}_2$, the violation trace $\tau$ also satisfies the combined assumption set $\mathcal{A}_1\cup\mathcal{A}_2$. Since $\tau\not\models\mathcal{A}$, we obtain $\mathcal{A}\prec(\mathcal{A}_1\cup\mathcal{A}_2)$.

Consequently, the specification $\varphi'=\langle\mathcal{A}_1\cup\mathcal{A}_2,\mathcal{G}\rangle$ may itself constitute a valid adaptation solution. If realisable, Axiom~\ref{axiom:preference_2} implies that $\varphi'$ is preferred to both $\varphi_1$ and $\varphi_2$, since it preserves strictly stronger assumptions while maintaining the same guarantees. This observation yields the following lemma.

\begin{lemma}
\label{lemma:exists_a_third_solution_uncomparable}
Let $(\varphi,\tau)$ be an adaptation problem. Suppose
$\varphi_1'=\langle\mathcal{A}_1',\mathcal{G}'\rangle$ and
$\varphi_2'=\langle\mathcal{A}_2',\mathcal{G}'\rangle$
are specification adaptations whose assumption sets are incomparable under the weakness relation. Then
$\varphi_3'=\langle\mathcal{A}_1'\cup\mathcal{A}_2',\mathcal{G}'\rangle$
is also a valid specification adaptation.
\end{lemma}


Lemma~\ref{lemma:exists_a_third_solution_uncomparable} shows that the first incomparable case reduces to the comparable setting addressed by Axiom~\ref{axiom:preference_1}. In particular, the construction of the specification adaptation $\varphi_3'$ yields a specification with strictly stronger assumptions while preserving the same guarantees.
The second incomparable case admits a similar analysis.

%

\begin{lemma}
\label{lemma:exists_a_third_solution_incomparable_2}
Let $(\varphi,\tau)$ be an adaptation problem. Suppose
$\varphi_1'=\langle\mathcal{A}_1',\mathcal{G}_1'\rangle$
and
$\varphi_2'=\langle\mathcal{A}_2',\mathcal{G}_2'\rangle$
are specification adaptations such that the assumption sets $\mathcal{A}_1'$ and $\mathcal{A}_2'$ are incomparable under the weakness relation, and
$\mathcal{G}_1'\prec\mathcal{G}_2'$.
Then
$\varphi_3'=\langle\mathcal{A}_1'\cup\mathcal{A}_2',\mathcal{G}_1'\rangle$
is also a valid specification adaptation.
\end{lemma}


Lemma~\ref{lemma:exists_a_third_solution_incomparable_2} combines the constructions of Lemmas~\ref{lemma:exists_a_third_solution} and~\ref{lemma:exists_a_third_solution_uncomparable}. Consequently, the second incomparable case also reduces to the comparable setting addressed by Axiom~\ref{axiom:preference_1}. In particular, the construction of $\varphi_3'$ yields a specification adaptation with strictly stronger assumptions than both $\varphi_1'$ and $\varphi_2'$, while preserving the stronger guarantee set $\mathcal{G}_1'$.


Finally, we argue that the remaining cases involving incomparable guarantee sets cannot, in general, be resolved through logical weakness alone. Distinguishing between such specification adaptations requires extending the problem setting with additional information, such as behavioural preferences, domain-specific priorities, or constraints identifying undesirable behaviours. In the absence of such information, no principled basis exists for preferring one incomparable guarantee set over another, and both adaptations must therefore be retained as valid candidate solutions.
This observation motivates the following final axiom.


\begin{axiom}
\label{axiom:preference_4}
Let $(\varphi,\tau)$ be an adaptation problem, and let
$\varphi_1'=\langle\mathcal{A}_1',\mathcal{G}_1'\rangle$
and
$\varphi_2'=\langle\mathcal{A}_2',\mathcal{G}_2'\rangle$
be two specification adaptations.

If the guarantee sets $\mathcal{G}_1'$ and $\mathcal{G}_2'$ are incomparable under the weakness relation, then neither solution is preferred to the other. Formally,
$
\mathcal{G}_1'
\text{ and }
\mathcal{G}_2'
\text{ incomparable}
\Rightarrow
(\varphi_1' \not\sqsubset \varphi_2')
\land
(\varphi_1' \not\sqsupset \varphi_2').
$
\end{axiom}


\begin{example}
Consider the following two solutions for an adaptation problem:
$\varphi_1'=\langle\mathcal{A}',\mathcal{G}_1'\rangle$
and
$\varphi_2'=\langle\mathcal{A}',\mathcal{G}_2'\rangle$,
where
$\mathcal{A}'=\emptyset$,
$\mathcal{G}_1'=\{\square(a\rightarrow \bigcirc c)\}$,
and
$\mathcal{G}_2'=\{\square(b\rightarrow \bigcirc \neg c)\}$.

Suppose these adaptations originate from the adaptation problem generated by the specification $\varphi=\langle\mathcal{A},\mathcal{G}\rangle$, where the violation trace is $\tau=[(a,b,c)]$, the assumptions are $\mathcal{A}=\{\square(\neg a\lor \neg b)\}$, and the original guarantees satisfy $\mathcal{G}=\mathcal{G}_1'\cup\mathcal{G}_2'$. Assume further that $a$ and $b$ are environment variables, while $c$ is a system variable.

At the purely symbolic level, there is insufficient information to determine whether $\varphi_1'$ or $\varphi_2'$ should be preferred. Both adaptations preserve different subsets of the original guarantees, and neither dominates the other under the logical preference criteria introduced previously.

Now consider the following semantic interpretation:
$a=\textit{highwater}$,
$b=\textit{methane}$,
and
$c=\textit{pump}$.
This corresponds to the Minepump specification introduced in Section~\ref{sec:minepump}, where $\textit{highwater}$ indicates dangerous water levels, $\textit{methane}$ indicates the presence of explosive gas, and $\textit{pump}$ denotes activation of the mine pump.

Under this interpretation, a clear preference emerges. Specifically, $\varphi_2'$ becomes preferable because it preserves the safety property preventing pump activation during methane leaks, thereby avoiding a potentially catastrophic explosion. Crucially, however, this preference depends entirely on semantic knowledge external to the formal specification itself. An adaptation algorithm operating purely at the logical level has no access to such contextual information and therefore cannot soundly discard either solution. Consequently, both candidate adaptations should be preserved for subsequent analysis by a domain-aware stakeholder.
\end{example}

We conclude that Axioms~\ref{axiom:preference_1}--\ref{axiom:preference_4} collectively induce a preference ordering over specification adaptations. We summarise these criteria in the following definition.

\begin{definition}[Preference between Specification Adaptations]
\label{def:preference_rint_solutions}
Let $(\varphi,\tau)$ be an adaptation problem, and let
$\varphi_1=\langle\mathcal{A}_1,\mathcal{G}_1\rangle$
and
$\varphi_2=\langle\mathcal{A}_2,\mathcal{G}_2\rangle$
be two semantically distinct specification adaptations in
$AllSolutions((\varphi,\tau))$.

We say that $\varphi_1$ is preferred to $\varphi_2$, denoted
$\varphi_1\sqsubset\varphi_2$,
iff one of the following conditions holds:
\begin{itemize}
    \item $(\mathcal{A}_1\equiv\mathcal{A}_2)
    \land
    (\mathcal{G}_1\prec\mathcal{G}_2)$;

    \item $(\mathcal{A}_1\prec\mathcal{A}_2)
    \land
    (\mathcal{G}_1\equiv\mathcal{G}_2)$;

    \item $(\mathcal{A}_1\prec\mathcal{A}_2)
    \land
    (\mathcal{G}_1\prec\mathcal{G}_2)$.
\end{itemize}

Under all other circumstances, $\varphi_1$ and $\varphi_2$ are considered incomparable with respect to preference.
\end{definition}

%

\subsubsection{Extracting Preferred Specifications}

An adaptation procedure may discover multiple candidate specification adaptations for a given adaptation problem $(\varphi,\tau)$. Let
$
\Phi \subseteq AllSolutions((\varphi,\tau))
$
denote the set of adaptations generated by such a procedure.

Using the preference relation from Definition~\ref{def:preference_rint_solutions}, the corresponding preference structure can be represented as the partially ordered set
$
(\Phi,\sqsubset).
$
Preferred adaptations can then be obtained by extracting the minimal elements of this partial order, corresponding to adaptations for which no strictly preferred alternative exists within $\Phi$.

%

However, Lemmas~\ref{lemma:exists_a_third_solution}, \ref{lemma:exists_a_third_solution_uncomparable}, and~\ref{lemma:exists_a_third_solution_incomparable_2} reveal an important limitation of this approach: the set $\Phi$ may not yet contain all relevant preferred adaptations. In particular, the previous lemmas show that combining assumptions from multiple solutions can yield new solutions that are strictly preferred according to the preference relation.

This observation leads to the following theorem.

\begin{theorem}[Most Preferred Assumption]
Let $\Phi\subseteq AllSolutions((\varphi,\tau))$ be a set of specification adaptations for an adaptation problem $(\varphi,\tau)$, where each solution $\varphi_i'=\langle\mathcal{A}_i',\mathcal{G}_i'\rangle$ belongs to $\Phi$.

Define $\mathcal{A}^*=\bigcup_{i\in|\Phi|}\mathcal{A}_i'$ to be the union of all assumption sets appearing in $\Phi$. We call $\mathcal{A}^*$ the \textit{most preferred assumption}. Then the following properties hold:
\begin{enumerate}
    \item $\mathcal{A}^*$ is stronger than or equivalent to every assumption set appearing in $\Phi$. Formally, $\forall i\in|\Phi|.\;\mathcal{A}^*\preceq\mathcal{A}_i'$;

    \item For every adaptation in $\Phi$, replacing its assumption set with $\mathcal{A}^*$ yields another valid specification adaptation. Formally, $\forall i\in|\Phi|.\;\langle\mathcal{A}^*,\mathcal{G}_i'\rangle \in AllSolutions((\varphi,\tau))$;

    \item For every adaptation in $\Phi$, the specification obtained using $\mathcal{A}^*$ is preferred to or equivalent to the original adaptation. Formally, $\forall i\in|\Phi|.\;\langle\mathcal{A}^*,\mathcal{G}_i'\rangle \sqsubseteq \varphi_i'$.
\end{enumerate}
\end{theorem}

This theorem allows every adaptation in $\Phi$ to be rewritten using the most preferred assumption $\mathcal{A}^*$, yielding a new set of solutions $\Phi'\subseteq AllSolutions((\varphi,\tau))$. Consequently, preference reasoning can be reduced to comparisons between guarantee sets alone, for which Axioms~\ref{axiom:preference_1} and~\ref{axiom:preference_4} are sufficient to identify the preferred adaptations.

\input{code/filter_max_algorithm}

We therefore propose Algorithm~\ref{alg:filterPreferred} for extracting the preferred specification adaptations from a candidate set $\Phi$. Lines~\ref{line:filter_asms_init}--\ref{line:filter_for_spec_end} separate the assumption and guarantee sets of each solution into the collections \textit{asms} and \textit{gars}, respectively. The minimal elements of these collections are then extracted in Lines~\ref{line:filter_find_min_asms} and~\ref{line:filter_find_min_gars} using the $\Call{findMinimal}{}$ procedure.


In Line~\ref{line:filter_get_asm_star_as_union}, the most preferred assumption $\mathcal{A}^*$ is constructed through the union of the minimal assumption sets contained in $minAsms$. Subsequently, Lines~\ref{line:filter_filtered_specs_init}--\ref{line:filter_for_min_gar_end} generate the preferred specification adaptations by combining $\mathcal{A}^*$ with each minimal guarantee set in $minGars$, yielding solutions of the form $\langle\mathcal{A}^*,\mathcal{G}'\rangle$. Each generated solution is added to the set $filteredSpecs$, which is returned in Line~\ref{line:filter_return_filtered_specs}.

The procedure $\Call{setUnion}{}$ corresponds to the standard set-union operator $\bigcup$, while $\Call{findMinimal}{ }$, defined in Algorithm~\ref{alg:filterMinimalFormulas}, returns the minimal elements of the partial order $(formulas,\prec)$. Additional implementation details and pseudocode are provided in Appendix~\ref{apx:helper_functions}.


%% file: code/filter_max_algorithm.tex
\begin{algorithm}
\caption{Filter Preferred Specifications}
\label{alg:filterPreferred}
\begin{algorithmic}[1]
\Function{filterPreferred}{$specs$}
\State $asms \gets []$ \label{line:filter_asms_init}
\State $gars \gets []$ \label{line:filter_gars_init}
\For{$spec \in specs$} \label{line:filter_for_spec_start}
    \State $\Call{add}{asms, spec.\mathcal{A}}$ \label{line:filter_add_asm}
    \State $\Call{add}{gars, spec.\mathcal{G}}$
    \label{line:filter_add_gar}
\EndFor \label{line:filter_for_spec_end}
\State $minAsms \gets \Call{findMinimal}{asms}$ \label{line:filter_find_min_asms}
\State $minGars \gets \Call{findMinimal}{gars}$ \label{line:filter_find_min_gars}
\State $\mathcal{A}^*\gets \Call{setUnion}{minAsms}$ \label{line:filter_get_asm_star_as_union}
\State $filteredSpecs \gets []$ \label{line:filter_filtered_specs_init}
\For{$\mathcal{G}' \in minGars$} \label{line:filter_for_min_gar_start}
    \State $\Call{add}{filteredSpecs, \langle\mathcal{A}^*,\mathcal{G}'\rangle}$ \label{line:filter_add_new_preferred_spec}
\EndFor \label{line:filter_for_min_gar_end}
\State \Return $filteredSpecs$ \label{line:filter_return_filtered_specs}
\EndFunction
\end{algorithmic}
\end{algorithm}

%% file: sections/4-reality_integration/4.3-trivial_solution.tex
\subsection{Baseline Adaptation Procedure}
\label{subsec:trivial-solution}

The previous section established the preference criteria governing specification degradations for adaptation problems. Intuitively, preferred adaptations preserve assumptions and guarantees as strongly as possible while still incorporating the observed violation behaviour and maintaining realisability.


To construct specification adaptations satisfying the requirements, we consider a baseline procedure based on removal of formulas responsible for the observed violation or for subsequent unrealisability. The approach is inspired by techniques for computing minimal unsatisfiable subsets of constraints~\cite{liffiton2008algorithms}, and by related work on synthesis under unrealizability~\cite{yang2024synthesizing}.

If removing the assumptions violated by the observed trace immediately yields a realisable specification, then adaptation is straightforward. We therefore focus on the more interesting setting in which the original assumptions are already minimally sufficient with respect to the guarantees (i.e. any assumption removed would make the specification unrealisable). Such assumption sets may, for example, be obtained using techniques such as those proposed in~\cite{shalom2023my}. Under this assumption, removing any violated assumption necessarily produces an unrealisable specification, requiring subsequent degradation of the guarantees to restore realisability.


The first step in adapting a specification is to identify the assumptions violated by the observed trace. These assumptions must be removed or weakened in order for the adapted specification to admit the newly observed behaviour.

\begin{definition}[Violation Set]
\label{def:violation_set}
Given a set of assumptions $\mathcal{A}$ and a violation trace $\tau$, the \textit{violation set} $A_\tau\subseteq\mathcal{A}$ is the subset of assumptions violated by $\tau$. Formally,
$
A_\tau=\{a\in\mathcal{A}\mid \tau\not\models a\}.
$
\end{definition}


After removing the violation set from the assumptions, the resulting specification becomes
$\langle\mathcal{A}\setminus A_\tau,\mathcal{G}\rangle$.
Since the original assumptions are assumed to be minimally sufficient, this specification is necessarily unrealisable. Consequently, restoring realisability requires degrading the guarantee set.

Following the intuition of~\cite{yang2024synthesizing}, we seek adaptations that remove as few guarantees as possible. To identify such adaptations, we first characterise the minimally unfulfillable subsets of guarantees, commonly referred to as unrealisable cores.

\begin{definition}[Unrealisable Core]
Let
$\varphi=\langle\mathcal{A},\mathcal{G}\rangle$
be an unrealisable specification. A set of guarantees
$\mathcal{G}^{uc}\subseteq\mathcal{G}$
is an \textit{unrealisable core} iff:
\begin{enumerate}
    \item the specification
    $\langle\mathcal{A},\mathcal{G}^{uc}\rangle$
    is unrealisable; and

    \item removing any guarantee from $\mathcal{G}^{uc}$ restores realisability. Formally,
    $
    \forall g_i\in\mathcal{G}^{uc}.\;
    \langle\mathcal{A},\mathcal{G}^{uc}\setminus\{g_i\}\rangle
    \text{ is realisable}.
    $
\end{enumerate}
\end{definition}

%

To minimise the number of guarantees removed during degradation, we require the notion of a minimal hitting set over the unrealisable cores.

\begin{definition}[Hitting Set~\cite{liffiton2008algorithms}]
\label{def:hitting_set}
Given a collection of sets $F=\{S_1,S_2,\dots,S_n\}$, a set $\mathcal{H}(F)$ is called a \textit{hitting set} of $F$ iff $\forall S_i\in F.\;\mathcal{H}(F)\cap S_i\neq\emptyset$. That is, $\mathcal{H}(F)$ contains at least one element from every set in $F$.

A hitting set $\mathcal{H}_{min}(F)$ is called \textit{minimal} iff no strict subset $H\subset\mathcal{H}_{min}(F)$ is itself a hitting set.
\end{definition}

%

In our setting, we compute the set of unrealisable cores
$\mathcal{C}=\{\mathcal{G}^{uc}_1,\mathcal{G}^{uc}_2,\ldots,\mathcal{G}^{uc}_k\}\subseteq\mathcal{P}(\mathcal{G})$
using the Punch algorithm \footnote{Experimental results show the implementation of the Punch algorithm may return non-minimal sets of guarantees as unrealisable cores. We contacted the SYNTECH team with respect to this matter, and will treat their implementation as sound for the rest of the paper.} proposed in~\cite{maoz2021unrealizable}. From this collection, we compute a minimal hitting set and remove it from the guarantee set, yielding the degraded guarantees
$\mathcal{G}'=\mathcal{G}\setminus\mathcal{H}_{min}(\mathcal{C})$.
The resulting specification
$\langle\mathcal{A}\setminus A_\tau,\mathcal{G}\setminus\mathcal{H}_{min}(\mathcal{C})\rangle$
is then realisable.

\begin{definition}[Optimal Trivial Degradation]
\label{def:optimal_trivial_solution}
Given the adaptation problem $(\langle \mathcal{A},\mathcal{G}\rangle,\tau)$, an \textit{optimal trivial degradation} is the specification
$\langle\mathcal{A}\setminus A_\tau,\mathcal{G}\setminus\mathcal{H}_{min}(\mathcal{C})\rangle$,
where:
\begin{itemize}
    \item $A_\tau\subseteq \mathcal{A}$ is the violation set induced by $\tau$;

    \item $\mathcal{C}=\{\mathcal{G}^{uc}_1,\mathcal{G}^{uc}_2,\ldots,\mathcal{G}^{uc}_k\}\subseteq\mathcal{P}(\mathcal{G})$ is the set of unrealisable cores of the specification $\langle \mathcal{A}\setminus A_\tau,\mathcal{G}\rangle$;

    \item $\mathcal{H}_{min}(\mathcal{C})$ is a minimal hitting set of $\mathcal{C}$.
\end{itemize}
\end{definition}

%

Note that multiple minimal hitting sets may exist for a given collection of unrealisable cores. Consequently, the adaptation problem may admit multiple optimal trivial solutions.


\input{code/trivial_rint_algorithm}
\input{code/trivial_gw_algorithm}

Algorithm~\ref{alg:trivialRInt} implements the baseline adaptation procedure described above and computes all optimal trivial solutions. Lines~\ref{line:trivialRInt_init_asm_set_empty}--\ref{line:trivialRInt_rename_asms_prime} identify the assumptions violated by the trace and remove them from the original assumption set. At the end of the loop in Line~\ref{line:trivialRInt_for_asm_end}, the set $vAsms$ corresponds to the violation set $A_\tau$, and the assignment in Line~\ref{line:trivialRInt_rename_asms_prime} constructs the degraded assumption set $\mathcal{A}'=\mathcal{A}\setminus A_\tau$.
Lines~\ref{line:trivialRInt_is_realisable}--\ref{line:trivialRInt_return_asm_prime_only} then check whether the resulting specification is realisable. If so, the unique optimal trivial solution is returned directly. Otherwise, Algorithm~\ref{alg:trivialGW} is invoked in Line~\ref{line:trivialRInt_call_trivial_GD} to compute the set of optimal trivial adaptations obtained through guarantee degradation.

Algorithm~\ref{alg:trivialGW} assumes that the degraded assumption set no longer violates the observed trace, but that the resulting specification remains unrealisable. In line~\ref{line:gd_get_ucs}, an external procedure such as the Punch algorithm from the Spectra toolbox~\cite{maoz2021spectra} is used to compute all unrealisable cores of the specification. Subsequently, line~\ref{line:gd_get_min_hitting_sets} computes the corresponding minimal correction sets through a minimal hitting-set procedure.
Lines~\ref{line:gd_for_guc_start}--\ref{line:gd_for_guc_end} then construct a realisable specification adaptation for each discovered correction set by removing the associated guarantees. The resulting set of optimal trivial solutions is returned in line~\ref{line:gd_return_solutions}.


From the discussion above, together with the construction implemented in Algorithm~\ref{alg:trivialRInt}, we obtain the following result.

\begin{theorem}[Soundness and Minimality of the Optimal Trivial Degradation]
\label{thm:soundness_minimality_optimal_trivial_solution}
Let $\langle\mathcal{A}\setminus A_\tau,\mathcal{G}\setminus\mathcal{H}_{min}(\mathcal{C})\rangle$ be an optimal trivial solution to the adaptation problem $(\langle \mathcal{A},\mathcal{G}\rangle,\tau)$ as defined in Definition~\ref{def:optimal_trivial_solution}. Then:
\begin{enumerate}
\renewcommand{\labelenumi}{(\roman{enumi})}
    \item the violation trace $\tau$ satisfies the degraded assumptions, i.e. $\tau\models \mathcal{A}\setminus A_\tau$;

    \item the specification $\langle\mathcal{A}\setminus A_\tau,\mathcal{G}\setminus\mathcal{H}_{min}(\mathcal{C})\rangle$ is realisable;

    \item the pair $\langle A_\tau,\mathcal{H}_{min}(\mathcal{C})\rangle$ is minimal with respect to the number of assumptions and guarantees removed.
\end{enumerate}
\end{theorem}

Properties~(i) and~(ii) are sufficient to show that the resulting specification is a valid solution for the adaptation problem, since the degraded assumptions and guarantees are trivially weaker than their original counterparts. We prove property (iii) of Theorem~\ref{thm:soundness_minimality_optimal_trivial_solution} in Appendix~\ref{apx:proof_soudness_trivial_algorithm}.

%% file: code/trivial_rint_algorithm.tex
\begin{algorithm}
\caption{Trivial Degradation}
\label{alg:trivialRInt}
\begin{algorithmic}[1]
\Require $\langle\mathcal{A}, \mathcal{G}\rangle$ (realisable GR(1) specification), $\tau$ (violation trace)
\Ensure $\trivialSolutions$, set of optimal trivial degradations

\Function{\trivialDegradation}{$\langle\mathcal{A},\mathcal{G}\rangle,\tau$}
\State $A_\tau \gets \emptyset$\label{line:trivialRInt_init_asm_set_empty}

\For{$\asm \in \mathcal{A}$}\label{line:trivialRInt_for_asm_start}
    \If{$\tau \not\models \asm$}\label{line:trivialRInt_check_tau_violates_asm}
        \State $A_\tau \gets A_\tau \cup \{\asm\}$\label{line:trivialRInt_add_v_asm_to_set}
    \EndIf
\EndFor\label{line:trivialRInt_for_asm_end}

\State $\mathcal{A}' \gets \mathcal{A}\setminus A_\tau$\label{line:trivialRInt_rename_asms_prime}

\If{$\Call{\isRealisable}{\langle\mathcal{A}',\mathcal{G}\rangle}$}\label{line:trivialRInt_is_realisable}
    \State \Return $[\langle\mathcal{A}',\mathcal{G}\rangle]$ \label{line:trivialRInt_return_asm_prime_only}
\Else
    \State \Return $\Call{\trivialGD}{\langle\mathcal{A}',\mathcal{G}\rangle}$\label{line:trivialRInt_call_trivial_GD}
\EndIf\label{line:trivialRInt_is_realisable_end}

\EndFunction
\end{algorithmic}
\end{algorithm}

%
%

%% file: code/trivial_gw_algorithm.tex
\begin{algorithm}
\caption{Trivial Guarantee Degradation}\label{alg:trivialGW}
\begin{algorithmic}[1]
\Require $\langle\mathcal{A}, \mathcal{G}\rangle$ (unrealisable GR(1) specification)
\Ensure $\trivialSolutions$, set of optimal trivial degradations

\Function{\trivialGD}{$\langle\mathcal{A},\mathcal{G}\rangle$}

\State $\trivialSolutions \gets [\ ]$

\State $\ucsSet \gets \Call{\getUnrealisableCores}{\langle\mathcal{A},\mathcal{G}\rangle}$\label{line:gd_get_ucs}

\State $\minCorrectionSets \gets \Call{\getMinimalHittingSets}{\ucsSet}$\label{line:gd_get_min_hitting_sets}

\For{$\mathcal{G}^{uc} \in \minCorrectionSets$}\label{line:gd_for_guc_start}
    \State $\mathcal{G}' \gets \mathcal{G}\setminus\mathcal{G}^{uc}$\label{line:gd_get_gars_prime}
    \State $\Call{append}{\trivialSolutions,\ \langle\mathcal{A},\mathcal{G}'\rangle}$\label{line:gd_save_solution}
\EndFor\label{line:gd_for_guc_end}

\State \Return $\trivialSolutions$\label{line:gd_return_solutions}

\EndFunction
\end{algorithmic}
\end{algorithm}

%
%
%

%% file: sections/4-reality_integration/4.4-motivating_example.tex
\subsection{Case Study: Minepump}
\label{subsec:motivating_example_minepump_2}

We now illustrate the behaviour and limitations of the baseline adaptation procedure presented in Algorithm~\ref{alg:trivialRInt} using the Minepump case study~\cite{kramer1983conic}, a canonical benchmark in the reactive-systems and formal-methods literature. Although comparatively simple, the Minepump specification captures the core interaction between environmental assumptions, safety-critical guarantees, and unexpected runtime behaviours, making it well suited for studying adaptation under assumption violations.

The Minepump system models a coal mine in which water may accumulate over time and must be removed using a pump. The $\textit{pump}$ constitutes the system to be synthesized; the environment is characterized by the water and methane sensor readings, $\textit{highwater}$ and $\textit{methane}$. However, methane may also accumulate inside the mine, and activating the pump in the presence of methane introduces a risk of explosion. The controller must therefore balance competing safety requirements under changing environmental conditions.

The corresponding GR(1) specification is shown in Table~\ref{tab:minepump}. Under the stated assumptions, the specification is realisable according to the standard GR(1) realisability criteria~\cite{bloem2012synthesis}, and a correct-by-construction controller can therefore be synthesized.

\input{code/minepump_spec}

Suppose now that, during deployment, the system observes a previously unmodelled behaviour in which both methane and dangerous water levels occur simultaneously. Such a trace violates $assumption\_1$, which explicitly excludes this configuration. As a consequence, the correctness guarantees established for the synthesized controller no longer necessarily hold. More fundamentally, the observation motivates the problem of how the specification itself should be adapted in response to newly observed environmental behaviour.

\subsubsection{Baseline Adaptation Applied}

Consider the specification $\varphi=\langle\mathcal{A},\mathcal{G}\rangle$ encoded in Table~\ref{tab:minepump}, where
$\mathcal{A}=\{\neg \textit{highwater} \land \neg \textit{methane},\Box(\neg\textit{highwater}\lor\neg\textit{methane}),\Box(\mathbf{Y}\textit{pump}\land\textit{pump}\rightarrow\bigcirc\neg\textit{highwater})\}$
and
$\mathcal{G}=\{\neg\textit{pump},\Box(\textit{highwater}\rightarrow\textit{pump}),\Box(\textit{methane}\rightarrow\neg\textit{pump})\}$.

For readability, we abbreviate these formulas as
$\mathcal{A}=\{\textit{assumption\_0},\textit{assumption\_1},\textit{assumption\_2}\}$
and
$\mathcal{G}=\{\textit{guarantee\_0},\textit{guarantee\_1},\textit{guarantee\_2}\}$.

Assume the system observes the violation trace
$\tau=[\{\textit{highwater},\textit{methane},\textit{pump}\}]$,
which violates $\textit{assumption\_1}$. Algorithm~\ref{alg:trivialRInt} therefore removes this assumption, producing the intermediate specification
$\varphi'=\langle\mathcal{A}',\mathcal{G}\rangle$,
where
$\mathcal{A}'=\{\textit{assumption\_0},\textit{assumption\_2}\}$.

The resulting specification is unrealisable. Computing the unrealisable cores yields two minimal correction sets:
$\{\textit{guarantee\_1}\}$
and
$\{\textit{guarantee\_2}\}$.
Consequently, the baseline adaptation procedure produces two optimal trivial solutions:
$\varphi_1=\langle\mathcal{A}',\{\textit{guarantee\_0},\textit{guarantee\_1}\}\rangle$
and
$\varphi_2=\langle\mathcal{A}',\{\textit{guarantee\_0},\textit{guarantee\_2}\}\rangle$.

Neither solution is entirely satisfactory after synthesis:
\begin{enumerate}
    \item $\varphi_1$ yields a controller that keeps the pump continuously active;

    \item $\varphi_2$ yields a controller that keeps the pump permanently inactive.
\end{enumerate}

The issue is not merely practical, but structural. Since the baseline procedure operates solely through formula removal, it cannot recover behaviours that remain compatible with the newly observed environment dynamics. As a result, important behavioural structure encoded in the original specification is lost unnecessarily.

Consider $\varphi_2$ as an example. The guarantee
$\Box(\textit{methane}\rightarrow\neg\textit{pump})$
must be preserved, since activating the pump in the presence of methane remains unsafe. At the same time, the original guarantee
$\Box(\textit{highwater}\rightarrow\textit{pump})$
can no longer be maintained universally, because the newly admitted environment behaviour allows both $\textit{highwater}$ and $\textit{methane}$ to hold simultaneously.

However, the original intent of the guarantee can still be partially preserved by strengthening its antecedent, and therefore weakening the overall formula. In particular, the specification
$\varphi_2'=\langle\mathcal{A}',\{\textit{guarantee\_0},\textit{guarantee\_2},\Box(\textit{highwater}\land\neg\textit{methane}\rightarrow\textit{pump})\}\rangle$
remains realisable while preserving substantially more of the intended controller behaviour. An analogous weakening can similarly be applied to $\varphi_1$.

These alternative adaptations preserve much more of the behavioural structure of the original specification while modifying behaviour only in the previously unmodelled edge case introduced by the violation trace. More generally, they illustrate an important limitation of purely removal-based adaptation procedures: optimal solutions may require weakening formulas rather than discarding them entirely.

This observation motivates the richer adaptation methodology developed in Section~\ref{sec:learning_adaptations}, where we consider formula-level weakening as a mechanism for exploring a substantially broader space of specification adaptations.

%% file: code/minepump_spec.tex
\begin{table}[h]
\centering
\begin{tabular}{|p{2.5cm}|p{5.5cm}|p{5.4cm}|}
\hline
\textbf{Type \& ID} & \textbf{Informal requirement} & \textbf{GR(1) formalization} \\
\hline
assumption\_0 & Start with low water and no methane & $\neg \textit{highwater} \land \neg \textit{methane}$ \\
\hline
\textbf{assumption\_1} & There can never be both high levels of water and dangerous levels of methane & 
$\Box(\lnot\textit{highwater} \lor \neg\textit{methane})$ \\
\hline
assumption\_2 & If pump is turned on for two time steps, water levels lower & 
$\Box(\textbf{Y}\textit{pump} \land\textit{pump} \rightarrow \bigcirc\neg\textit{highwater})$ \\
\hline
guarantee\_0 & Start with pump off & 
$\neg\textit{pump}$\\
\hline
\textbf{guarantee\_1} & Whenever water exceeds threshold, turn pump on & 
$\Box(\textit{highwater} \rightarrow \textit{pump})$ \\
\hline
\textbf{guarantee\_2} & Whenever methane levels are dangerous, turn pump off & 
$\Box(\textit{methane} \rightarrow \neg\textit{pump})$ \\
\hline
\end{tabular}
\caption{GR(1) specification for the mine pump example}
\label{tab:minepump}
\end{table}

%% file: sections/5-learning_adaptations/5.0-introduction.tex
Section~\ref{sec:reality_integration} introduced the adaptation problem, formally characterised the space of admissible adaptations, and established preference criteria for comparing alternative solutions. We additionally presented a baseline degradation procedure based on formula removal and showed that it produces sound minimal adaptations restoring realisability after assumption violations.

The removal-based methodology, however, explores only a restricted portion of the degradation space. In many cases, stronger degradations can be avoided by weakening formulas rather than discarding them entirely, thereby preserving substantially more of the original specification behaviour.

In this section, we therefore investigate finer-grained degradation procedures based on formula-level weakening. Section~\ref{subsec:syntactic_degradation_methods} introduces syntactic weakening operators for generating degraded formulas and presents a collection of degradation strategies. Section~\ref{subsec:finer_grained_solutions_rint} then shows how formula-level degradation can be incorporated into a degradation procedure based on the OGIS framework. Finally, Section~\ref{subsec:learning_degradations_for_reality_integration} presents an implementation capable of systematically exploring the resulting degradation space and generating multiple candidate specification adaptations. Proofs of correctness and termination are deferred to Appendix~\ref{apx:sec:proofs}.

%% file: sections/5-learning_adaptations/5.1-syntactic_degradation_methods.tex
\subsection{Syntactic Formula Weakening}
\label{subsec:syntactic_degradation_methods}

The baseline adaptation procedure presented earlier operates exclusively through formula removal. While effective for restoring realisability, such an approach may discard substantially more behavioural structure than necessary. Prior work on specification weakening~\cite{zhang2025behavioral,buckworth2023adapting} has, therefore, considered finer-grained transformations capable of generating weaker formulas through modifications to their syntactic structure. This process is commonly referred to as \textit{syntactic degradation}.

\begin{definition}[Syntactic Degradation]
Given a formula $\varphi$, the process of transforming its syntactic structure to obtain a weaker formula $\varphi'$ such that $\varphi'\succ\varphi$ is called \textit{syntactic formula weakening}.
\end{definition}

To some extent, the baseline degradation procedure introduced in Section~\ref{subsec:trivial-solution} already constitutes a form of syntactic formula weakening. Removing formulas from the assumption or guarantee sets is equivalent to deleting conjuncts from the strict realisability formula represented by the GR(1) specification in set form (cf. Definition~\ref{def:gr1-strict-realisability}).

However, simple formula removal provides only a coarse-grained degradation mechanism. While conjunct deletion can also be applied within invariant and justice formulas, such transformations remain limited in the range of behavioural degradations they can express. Consequently, we instead focus on adaptation procedures based on  weakening formulas through local syntactic modifications rather than removed entirely.

We begin by focusing on the degradation of invariants. The corresponding treatment for justice goals is analogous, and we highlight the relevant differences where appropriate. To enable syntactic weakening, we represent invariants as implications between formulas in disjunctive normal form (DNF) under the scope of the globally operator:
\begin{align*}
    \varphi^{\invariant} = \Box\Big(\bigvee_{i\in m} \bigwedge_{j\in n_i}\alpha_{ij} \rightarrow \bigvee_{k\in p} \bigwedge_{l\in q_k}\beta_{kl}\Big)
\end{align*}
where each $\alpha$ is a literal that may optionally be preceded by the previous operator $\mathbf{Y}$, and each $\beta$ is a literal that may optionally be preceded by the next operator $\mathbf{X}$. This representation forms the basis of the syntactic augmentation procedures used for invariant degradation.

For justice goals, the structure is similar. However, since GR(1) synthesis tools such as Spectra support response-style justice specifications~\cite{maoz2021spectra,maoz2015gr,dwyer1999patterns}, we additionally consider response formulas of the form:
\begin{align*}
    \varphi^{\response} = \Box\Big(\bigvee_{i\in m} \bigwedge_{j\in n_i}\alpha_{ij} \rightarrow \textcolor{red}{\lozenge}\bigvee_{k\in p} \bigwedge_{l\in q_k}\beta_{kl}\Big)
\end{align*}
The only structural difference between $\varphi^{\invariant}$ and $\varphi^{\response}$ is the presence of the eventuality operator $\lozenge$ in the consequent of the response formula.

\subsubsection{Types of Weakening}
\label{sec:weakening-types}

Given a formula $\varphi$, we consider three classes of syntactic degradation: antecedent weakening, consequent weakening, and conversion to response or justice formulas. For each degradation operator introduced below, it follows directly from the rules of LTL entailment that the resulting formula is weaker than the original specification.

\paragraph{Antecedent Weakening}

Consider an invariant of the form $\Box(A\rightarrow B)$. Such a formula becomes weaker whenever the antecedent $A$ becomes stronger, since the implication is then required to hold in fewer situations. Because the antecedent is represented in disjunctive normal form (DNF), we preserve the syntactic structure of the formula by strengthening individual disjuncts through conjunction with additional literals. Intuitively, this restricts the set of behaviours under which the consequent must hold.
We therefore define the following degradation operator.

\begin{definition}[Antecedent Weakening]
\label{def:antecedent_weakening}
Formula $\bar\varphi^{\invariant}=AW(\varphi^{\invariant},x,\gamma)$ is the antecedent weakening of $\varphi^{\invariant}$ obtained by conjoining the unary temporal literal $\gamma$ to the $x$th disjunct of the antecedent of $\varphi^{\invariant}$, where $x\in[1..m]$. The resulting formula has the form:
\begin{align*}
    \bar\varphi^{\invariant} \Coloneqq \Box\Big[\textcolor{red}{(\gamma \land\bigwedge_{j\in n_x}\alpha_{xj})\lor} \big(\bigvee_{i\in m,i\not=x} \bigwedge_{j\in n_i}\alpha_{ij}\big) \rightarrow \bigvee_{k\in p} \bigwedge_{l\in q_k}\beta_{kl}\Big]
\end{align*}
\end{definition}

This methodology extends naturally to response formulas. For justice goals, we observe that every justice formula of the form $\Box\lozenge\varphi$ can equivalently be rewritten as the response formula $\Box(\top\rightarrow\lozenge\varphi)$. Antecedent weakening can then be applied by conjoining a unary temporal literal $\gamma$ to the antecedent $\top$ introduced by this transformation.

\paragraph{Consequent Weakening}

An alternative mechanism for weakening $\varphi^{\invariant}$ is to relax its consequent. Intuitively, adding additional disjuncts to the consequent enlarges the set of behaviours satisfying the implication, thereby weakening the overall specification. We formalise this transformation below.

\begin{definition}[Consequent Weakening]
\label{def:consequent_weakening}
Formula $\bar\varphi^{\invariant}=CW(\varphi^{\invariant},\Gamma)$ is the consequent weakening of $\varphi^{\invariant}$ obtained by introducing the conjunction of unary temporal literals $\Gamma$ as an additional disjunct in the consequent of $\varphi^{\invariant}$. The resulting formula has the following form:
\begin{align*}
    \bar\varphi^{\invariant} \Coloneqq \Box\Big(\bigvee_{i\in m} \bigwedge_{j\in n_i}\alpha_{ij} \rightarrow \big(\bigvee_{k\in p} \bigwedge_{l\in q_k}\beta_{kl}\big)\textcolor{red}{\;\lor\;\Gamma} \Big)
\end{align*}
\end{definition}

The same transformation applies analogously to response and justice formulas, with the additional disjunct introduced within the scope of the eventuality operator $\lozenge$.

\paragraph{Conversion to Response or Justice}

During the development of our degradation framework, and in particular while analysing the Arbiter case study, we identified the need for a third class of weakening operator. Unlike antecedent and consequent weakening, this transformation applies specifically to invariants and captures situations in which behaviour originally specified as an immediate obligation is more appropriately interpreted as an eventual response.

Intuitively, some invariants may be overly restrictive because they require the consequent to hold immediately after the antecedent becomes true. In practice, however, the intended behaviour may only require eventual satisfaction. This motivates a degradation operator that relaxes invariants into response-style properties by introducing an eventuality operator in the consequent.

\begin{definition}[Conversion to Response or Justice]
\label{def:conversion_to_response_or_justice}
We define the formula $\bar\varphi^{\invariant}=EV(\varphi^{\invariant})$ as the conversion of the invariant $\varphi^{\invariant}$ into a response formula by prepending the consequent of $\varphi^{\invariant}$ with the eventuality operator $\lozenge$. The resulting formula has the form:
\begin{align*}
    \bar\varphi^{\invariant} \Coloneqq \Box\Big(\bigvee_{i\in m} \bigwedge_{j\in n_i}\alpha_{ij} \rightarrow \textcolor{red}{\lozenge}\bigvee_{k\in p} \bigwedge_{l\in q_k}\beta_{kl}\Big)
\end{align*}

We refer to this transformation as \textit{conversion to justice} in the special case where the invariant $\varphi^{\invariant}$ has no antecedent. In such cases, the resulting formula has the form:
\begin{align*}
    \bar\varphi^{\invariant} \Coloneqq \Box \textcolor{red}{\lozenge}\bigvee_{k\in p} \bigwedge_{l\in q_k}\beta_{kl}
\end{align*}
\end{definition}

\subsubsection{Weakening Search Space}

The degradation operators introduced in Definitions~\ref{def:antecedent_weakening}, \ref{def:consequent_weakening}, and~\ref{def:conversion_to_response_or_justice} collectively induce a search space of weaker formulas for a given GR(1) specification component.

\begin{definition}[Weakening Search Space]
\label{def:weakening_search_space}
Let $\varphi$ be a formula following one of the syntactic structures introduced in Section~\ref{subsec:syntactic_degradation_methods} (that is, either $\varphi^{\invariant}$ or $\varphi^{\response}$). We define the weakening search space of $\varphi$, denoted $Space(\varphi)$, as the set of formulas $\varphi'\succ\varphi$ obtainable through applications of the weakening operators $AW()$, $CW()$, and $EV()$. Formally:
\begin{align*}
    Space(\varphi)=&
    \{\varphi'\mid \forall \gamma,\forall x\in[1..m].\ \varphi'=AW(\varphi,x,\gamma)\}\ \cup\\
    &\{\varphi'\mid \forall \Gamma.\ \varphi'=CW(\varphi,\Gamma)\}\ \cup\\
    &\{EV(\varphi)\}
\end{align*}
\end{definition}

The search space $Space(\varphi)$ provides the foundation for exploring alternative specification adaptations beyond simple formula removal. Importantly, however, the degradation operators defined above are neither unique nor exhaustive. They represent only one possible family of weakening transformations and do not guarantee discovery of the globally optimal degradation for a given adaptation problem.

Consequently, the challenge is not merely to generate weaker formulas, but rather to systematically explore the resulting degradation space in a manner that prioritises preferred adaptations according to the preference framework introduced in Section~\ref{sec:reality_integration}. The following sections develop a degradation procedure capable of leveraging these operators to discover higher-quality solutions than those produced by the baseline removal-based approach.

%% file: sections/5-learning_adaptations/5.2-finer_grained_rint_solutions.tex
\subsection{Specification Degradation via Formula Weakening}
\label{subsec:finer_grained_solutions_rint}

Recall that a solution for an adaptation problem $(\langle\mathcal{A},\mathcal{G}\rangle,\tau)$ is a realisable specification $\langle\mathcal{A}',\mathcal{G}'\rangle$ such that $\tau\models\mathcal{A}'$, where the degraded assumptions and guarantees are weaker than, or equivalent to, the original specification components. The baseline procedure in Algorithm~\ref{alg:trivialRInt} exploited these properties by constructing adaptations through formula removal from the sets $\mathcal{A}$ and $\mathcal{G}$.

As illustrated in the Minepump case study, however, purely removal-based adaptations explore only a restricted portion of the degradation space. In many situations, substantially better adaptations can be obtained by weakening formulas directly rather than discarding them entirely. This motivates the introduction of formula-level degradation operators within specifications.
We therefore introduce the following notation.

\begin{definition}[Formula-level Degradation]
Let $\mathcal{A}$ be a set of formulas. Suppose $a\in\mathcal{A}$ and let $a'$ be a degradation of $a$ such that $a\prec a'$. We define the weakening operation
\[
\Call{weaken}{\mathcal{A},a,a'}
=
\mathcal{A}\setminus\{a\}\cup\{a'\}
=
\mathcal{A}'.
\]
Intuitively, the operation replaces the original formula $a$ with a weaker formula $a'$, thereby generating a degraded specification component $\mathcal{A}'$. In our implementation, we usually select $a'$ from the space of weaker formulas from Definition~\ref{def:weakening_search_space}, so $a'\in Space(a)$.
\end{definition}

The introduction of formula-level weakening reveals that the baseline procedure in Algorithm~\ref{alg:trivialRInt} implicitly solves two distinct adaptation tasks: assumption adaptation and guarantee adaptation. We formalise these tasks separately below.

\begin{definition}[Assumption Adaptation]
\label{def:rint_a}
Let $\varphi=\langle\mathcal{A},\mathcal{G}\rangle$ be a specification and let $\tau$ be a violation trace. The assumption adaptation task consists of constructing a weakened assumption set $\mathcal{A}'$ such that:
\begin{enumerate}
    \item $\tau\models\mathcal{A}'$;

    \item $\mathcal{A}\prec\mathcal{A}'$.
\end{enumerate}
\end{definition}

The resulting specification $\langle\mathcal{A}',\mathcal{G}\rangle$ may still be unrealisable. However, the weakened assumptions at least incorporate the observed violation trace into the environment model. The remaining task is therefore to construct a suitable weakened guarantee set restoring realisability, which motivates the complementary adaptation task below.

\begin{definition}[Guarantee Adaptation]
\label{def:rint_g}
Let $\varphi=\langle\mathcal{A},\mathcal{G}\rangle$ be a specification, let $\tau$ be a violation trace, and let $\mathcal{A}'$ be a solution to the assumption adaptation task. The guarantee adaptation task consists of constructing a weakened guarantee set $\mathcal{G}'$ such that:
\begin{enumerate}
    \item $\mathcal{G}\preceq\mathcal{G}'$;

    \item the specification $\langle\mathcal{A}',\mathcal{G}'\rangle$ is realisable.
\end{enumerate}
\end{definition}

This decomposition highlights that the adaptation problem can be viewed as the sequential composition of assumption adaptation followed by guarantee adaptation.

\subsubsection{Learning Assumption Degradations}

The objective of assumption adaptation is to construct a weakened assumption set $\mathcal{A}'$ that admits the observed violation trace while preserving as much of the original environment model as possible. 


To construct a weakened assumption set using only degradation operations, we first identify the assumptions violated by the trace, exactly as in Section~\ref{subsec:trivial-solution}. Weakening assumptions already satisfied by the violation trace would be unnecessary, since they already admit the observed behaviour. Consequently, the adaptation procedure focuses exclusively on formulas belonging to the violation set $\mathcal{A}_\tau$ from Definition~\ref{def:violation_set}.

The viability of this approach follows from the observation below.

\begin{lemma}[Existence of Assumption Degradations]
\label{lem:each_asm_has_rint_degradation_solution}
Given a violation set $\mathcal{A}_\tau\subseteq\mathcal{A}$ for a violation trace $\tau$, every violated assumption admits a weaker formula satisfied by the trace. Formally,
\[
\forall a\in\mathcal{A}_\tau.\ \exists a'\succ a.\ \tau\models a'.
\]
\end{lemma}

The lemma follows immediately from the observation that every assumption admits a trivial weakening to $\top$, which is satisfied by any trace, including $\tau$. This observation further extends to the assumption set as a whole, leading to the theorem below.


\begin{theorem}[Existence of Violating-Assumption Degradations]
\label{th:deg_all_violating_assumptions}
Given a set of assumptions $\mathcal{A}$, a violation trace $\tau\not\models\mathcal{A}$, and the corresponding violation set $\mathcal{A}_\tau$, there exists a set of degraded assumptions $\mathcal{A}_\tau'=\{a_1',\dots,a_n'\}$
such that:
\begin{enumerate}
    \item for every violated assumption $a_i\in\mathcal{A}_\tau$, the corresponding degraded assumption satisfies $a_i\prec a_i'$;

    \item the degraded assumption set satisfies the violation trace, i.e.,   $\tau\models\mathcal{A}_\tau'$.
\end{enumerate}
\end{theorem}


Lemma~\ref{lem:each_asm_has_rint_degradation_solution} shows that every violated assumption admits at least one satisfying degradation. Consequently, assumption adaptation can be formulated as the task of discovering an appropriate degraded replacement for each formula in the violation set $\mathcal{A}_\tau$, thereby constructing a weakened assumption set satisfying the observed trace.

More specifically, the objective is to construct a degraded assumption set of the form
$\mathcal{A}'=
\mathcal{A}\setminus\mathcal{A}_\tau \cup \mathcal{A}_\tau'$,
where every formula in $\mathcal{A}_\tau'$ is a degradation of a corresponding violated assumption in $\mathcal{A}_\tau$, and the resulting assumption set satisfies the violation trace, that is, $\tau\models\mathcal{A}'$.

\subsubsection{Learning Guarantee Degradations}

Unlike assumption adaptation, guarantee adaptation cannot be guided directly by trace satisfaction alone. While assumption degradations can be validated locally against the observed violating traces, guarantee degradations must additionally preserve realisability of the resulting specification, introducing a substantially more difficult global synthesis problem.

Oracle-Guided Inductive Synthesis (OGIS)~\cite{jha2017theory} is a synthesis framework in which a learner proposes candidate solutions from a hypothesis space, and an oracle checks whether each candidate satisfies the desired correctness condition. If a candidate fails, the oracle returns information, typically in the form of a counterexample \cite{alur2013counter}, that guides the next synthesis step.
In this work, we employ OGIS as the underlying framework for solving the assumption adaptation problem. It is particularly well suited to our setting because it separates candidate generation from formal verification, allowing degradation candidates to be refined iteratively through interaction with correctness oracles.

Our setting, however, differs from the classical OGIS formulation in an important respect. Traditional OGIS approaches typically assume a monotonic learning process, in which newly acquired examples progressively reduce the admissible synthesis space. In contrast, adaptation under runtime assumption violations is inherently non-monotonic: newly observed traces may invalidate previously preferred degradations, requiring the adaptation procedure to revisit and revise earlier weakening decisions.

Consequently, rather than learning degradations from a fixed set of examples in a single step, our methodology incrementally accumulates violating traces and counterexamples throughout execution. The guarantee adaptation task can therefore be viewed as an OGIS loop over the degradation search space introduced in Section~\ref{subsec:syntactic_degradation_methods}. The inductive component generates weakened guarantees through applications of the syntactic degradation operators and the $\Call{weaken}{}$ transformation, while the oracle component verifies whether the resulting guarantees satisfy the accumulated violation traces and preserve the structural requirements of the specification. Counterexamples returned by the oracle are then incorporated into subsequent OGIS iterations, progressively refining the degradation search process.


Consider a specification $\varphi'=\langle\mathcal{A}',\mathcal{G}\rangle$ obtained after adapting the assumptions of some specification $\varphi$ in response to a violation trace $\tau$. Suppose further that $\varphi'$ is unrealisable. From the discussion in Section~\ref{subsec:trivial-solution}, unrealisability implies the existence of at least one unrealisable core within the guarantee set.

Each unrealisable core corresponds to one or more counter-strategies demonstrating how the environment can force the system to violate the specification. Intuitively, a counter-strategy identifies behaviours under which the current guarantees cannot all be enforced simultaneously. Guided by the OGIS paradigm, we therefore use counter-strategies as oracle feedback: given a counter-strategy $\cs$, the OGIS loop searches for guarantee degradations that eliminate the ability of $\cs$ to force a guarantee violation.

We therefore focus on degrading guarantees so that the behaviours enforced by a counter-strategy are no longer considered violating behaviours. To achieve this, we extract a finite counter-play $\rho\in\Sigma^*$ compliant with the selected counter-strategy $\cs$, such that $\rho\not\models\mathcal{G}$, and search for a degraded guarantee set $\mathcal{G}'\succ\mathcal{G}$ satisfying $\rho\models\mathcal{G}'$.

This leads to the following learning problem.

\begin{definition}[Counter-Play-Guided Guarantee Adaptation]\label{def:rint_g_ct}
Let $\mathcal{G}$ be the guarantee set of an unrealisable specification $\varphi=\langle\mathcal{A}',\mathcal{G}\rangle$. Furthermore, let $C\subset\Sigma^*$ be a finite set of counter-plays representing behaviours that the adapted specification should permit.

The counter-play-guided guarantee adaptation task consists of constructing a degraded guarantee set $\mathcal{G}'$ such that:
\begin{enumerate}
    \item $\mathcal{G}\prec\mathcal{G}'$;

    \item $\forall \rho\in C.\ \rho\models\mathcal{G}'$.
\end{enumerate}
\end{definition}

We allocate the task of discovering one or more counter-plays to the oracle component of the OGIS loop. Once extracted, these counter-plays guide the guarantee adaptation process in a manner closely analogous to the role of violation traces during assumption adaptation. In particular, counter-plays identify guarantees whose current behaviour prevents the specification from admitting the behaviours enforced by the environment counter-strategy.

This observation allows us to derive the following dual of Lemma~\ref{lem:each_asm_has_rint_degradation_solution}.

\begin{lemma}[Existence of Guarantee Degradations]
\label{lem:each_gar_has_rint_degradation_solution}
Given a counter-play violation set $\mathcal{G}_C\subseteq\mathcal{G}$ for a set of counter-plays $C$, every violating guarantee admits a weaker formula satisfied by all counter-plays in $C$. Formally,
\[
\forall g\in\mathcal{G}_C.\ \exists g'\succ g.\ \forall \rho\in C.\ \rho\models g'.
\]
\end{lemma}

The lemma follows immediately from the observation that every guarantee admits a trivial weakening to $\top$, which is satisfied by every counter-play. This observation extends naturally to the entire guarantee set, leading to the theorem below.

%
%

\begin{theorem}[Existence of Violating-Guarantee Degradations]
\label{th:deg_all_violating_guarantees}
Given a set of guarantees $\mathcal{G}$, a set of counter-plays $C$ such that $\forall \rho\in C.\ \rho\not\models\mathcal{G}$, and the corresponding counter-play violation set $\mathcal{G}_C$, there exists a set of degraded guarantees $\mathcal{G}_C'$ such that:
\begin{enumerate}
    \item for every violating guarantee $g_i\in\mathcal{G}_C$, the corresponding degraded guarantee satisfies $g_i\prec g_i'$;

    \item the degraded guarantee set satisfies every counter-play in $C$, that is, $$\forall \rho\in C.\ \rho\models\mathcal{G}_C'.$$
\end{enumerate}
\end{theorem}

Note that the construction above does not necessarily imply that the resulting specification $\varphi''=\langle\mathcal{A}',\mathcal{G}'\rangle$ is realisable. If $\varphi''$ remains unrealisable, the guarantee adaptation procedure can be applied iteratively over the degraded guarantee set $\mathcal{G}'$, using a newly generated set of counter-plays $C'$ returned by the oracle component.

More specifically, the OGIS loop searches for a further degradation $\mathcal{G}''\succ\mathcal{G}'$ such that the specification $\varphi'''=\langle\mathcal{A}',\mathcal{G}''\rangle$ becomes realisable. If realisability is achieved, the OGIS loop terminates. Otherwise, the process continues iteratively until a sufficiently weak guarantee set is discovered that restores realisability, potentially including the empty guarantee set in the worst case.

A proof of soundness and termination for this iterative OGIS loop is provided in Appendix~\ref{apx:proof_soundness_cegis}.

%% file: sections/5-learning_adaptations/5.3-learning_degradations_for_reality_integration.tex
\subsection{Synthesis of Preferred Degradations}
\label{subsec:learning_degradations_for_reality_integration}
Recall Definition~\ref{def:preference_rint_solutions}, which characterises preference between degraded specifications. In particular, preferred solutions correspond to specifications preserving the strongest possible guarantees while remaining incomparable only when semantically unavoidable.

Consequently, the OGIS loop must explore the degradation space sufficiently thoroughly to identify the strongest incomparable guarantee degradations obtainable through syntactic weakening. Since different degradation paths may preserve different subsets of system behaviour, restricting the search prematurely risks overlooking preferred adaptations and converging toward unnecessarily weak specifications.

The objective of the OGIS loop presented in this section is therefore not merely to restore realisability, but to systematically explore the degradation space generated by the syntactic weakening operators introduced earlier, while using the preference ordering to guide the extraction of the most desirable degraded specifications.

\input{code/learning_rint_algorithm} 

We now introduce Algorithm~\ref{alg:learningRInt}, which employs a breadth-first OGIS-based degradation loop to compute preferred degraded specifications.

In line~\ref{line:ogis_learn_asm_deg}, the learner generates candidate assumption degradations satisfying the assumption adaptation conditions from Definition~\ref{def:rint_a}. Lines~\ref{line:ogis_gen_aw_start}--\ref{line:ogis_gen_aw_end} then partition the generated specifications into two sets: the set \textit{learnedSolutions}, containing specifications that are already realisable after assumption adaptation, and the set \textit{unrealisableSpecifications}, containing the remaining unrealisable candidates.

In line~\ref{line:ogis_check_early_stop}, the algorithm checks whether any realisable specifications were discovered. If so, further guarantee adaptation is unnecessary. The algorithm therefore returns, in line~\ref{line:ogis_early_stop}, the merged preferred realisable degraded specifications extracted using the $\Call{merge}{}$ procedure in Appendix~\ref{apx:helper_functions} in Algorithm~\ref{alg:mergeSpecs} and the $\Call{filterPreferred}{}$ procedure introduced in Algorithm~\ref{alg:filterPreferred}.

If no such specifications are discovered, the algorithm extracts the strongest intermediate specification in line~\ref{line:ogis_first_ispec}. Observe that, at this stage, a call to $\Call{filterPreferred}{}$ returns only a single specification, since all candidates still preserve the original guarantee set $\mathcal{G}$. Because this intermediate specification is unrealisable, the oracle component then extracts multiple sets of counter-plays in line~\ref{line:ogis_first_set_cp}, following the optimisation strategy discussed in Section~\ref{subsec:finer_grained_solutions_rint}. 

In lines~\ref{line:ogis_init_queue_start}--\ref{line:ogis_init_queue_end}, the algorithm constructs the initial synthesis frontier by inserting tuples composed of an unrealisable specification $spec$ together with an associated set of counter-plays $C$ into the $candidateQueue$.

Starting from line~\ref{line:ogis_cpgis_start}, the algorithm enters the counter-play-guided guarantee adaptation loop. While there is some candidate tuple, we extract it on line~\ref{line:ogis_pop_tuple} from the queue,
then apply on line~\ref{line:ogis_learn_guarantee_degradation} a counter-play guided guarantee degradation step over the intermediate specification $iSpec$ and the set of counter-plays $C$, as per Definition~\ref{def:rint_g_ct}, to learn a new set of possible solutions to the adaptation problem. 

Lines~\ref{line:ogis_learned_specs_check_start}--\ref{line:ogis_learned_specs_check_end} then evaluate each specification in the resulting set \textit{learnedSpecs}. If a specification is realisable, it is added to the set \textit{learnedSolutions} in line~\ref{line:ogis_add_learned_solution}. Otherwise, the oracle component extracts new counter-play sets for the unrealisable specification in line~\ref{line:ogis_cpgis_get_set_cp}. Finally, lines~\ref{line:ogis_iterate_cps_start}--\ref{line:ogis_iterate_cps_end} extend the synthesis frontier by inserting newly generated specification/counter-play tuples back into the queue, allowing the OGIS loop to continue refining guarantee degradations iteratively.

Eventually, the OGIS loop terminates once all relevant degradations have been explored, at which point the queue becomes empty after line~\ref{line:ogis_cpgis_end}. As proved in Appendix~\ref{apx:proof_soundness_cegis}, 
every exploration branch eventually converges to a realisable specification. 
Finally, line~\ref{line:ogis_end} applies the $\Call{filterPreferred}{}$ and $\Call{merge}{}$ procedures in succession to the accumulated set \textit{learnedSolutions}, returning the merged preferred degraded specifications.

The OGIS loop can be further optimised by exploiting the preference ordering introduced earlier. Since the objective is to discover specifications preserving the strongest distinct guarantee sets, the exploration queue need not retain candidates that cannot contribute to preferred adaptations.

More specifically, the $candidateQueue$ can be extended with a redundancy check in line~\ref{line:ogis_redundancy_check}, preventing reinsertion of candidate tuples composed of equivalent intermediate specifications together with identical, or sufficiently similar, counter-play sets. In such cases, the resulting OGIS branches are expected to explore equivalent degradation regions, making repeated exploration unnecessary.

Additionally, if an intermediate specification is already weaker than some specification contained in $learnedSolutions$, then it need not be inserted into the queue in line~\ref{line:ogis_add_queue}. Any adaptations generated from that specification would eventually be discarded by the $\Call{filterPreferred}{}$ procedure, since stronger preferred adaptations have already been discovered.

These optimisations substantially reduce redundant exploration while preserving completeness with respect to the preferred degraded specifications returned by the OGIS loop.

%% file: code/learning_rint_algorithm.tex
\newcommand\doubleplus{+\kern-1.3ex+\kern0.8ex}
\newcommand\mdoubleplus{\ensuremath{\mathbin{+\mkern-10mu+}}}

\begin{algorithm}
\caption{OGIS-Based Synthesis of Preferred Degradations}
\label{alg:learningRInt}
\begin{algorithmic}[1]
\Require $\langle\mathcal{A}, \mathcal{G}\rangle$ (realisable GR(1) specification), $\tau$ (violation trace), $AW$ and $CW$ (syntactic degradation operators)
\Ensure $\preferredDegradations=\{\langle\mathcal{A}',\mathcal{G}'\rangle\}$, a set of preferred realisable degraded specifications

\Function{synthesisePreferredDegradations}{$\langle\mathcal{A},\mathcal{G}\rangle,\tau$} 
\State $\learnedSpecs \gets \Call{\learnAssumptionDegradations}{\langle\mathcal{A},\mathcal{G}\rangle,\tau}$ \label{line:ogis_learn_asm_deg}
\State $\learnedSolutions \gets [\ ]$ \label{line:ogis_gen_aw_start}
\State $\unrealisableSpecifications \gets [\ ]$

\For{$\spec \in \learnedSpecs$}
    \If{$\Call{\isRealisable}{\spec}$}
        \State \Call{append}{$\learnedSolutions,\ \spec$}
    \Else
        \State \Call{append}{$\unrealisableSpecifications,\ \spec$}
    \EndIf
\EndFor \label{line:ogis_gen_aw_end}

\If{$\learnedSolutions \neq \emptyset$}\label{line:ogis_check_early_stop}
    \State \Return $\textsc{merge}\left(\Call{filterPreferred}{\learnedSolutions}\right)$\label{line:ogis_early_stop}
\EndIf

\State $\spec \gets \Call{filterPreferred}{\unrealisableSpecifications}[0]$\label{line:ogis_first_ispec}
\State $\newCounterPlaySets \gets \Call{\extractCounterPlays}{\spec}$\label{line:ogis_first_set_cp}
\State $\candidateQueue \gets [\ ]$\label{line:ogis_init_queue_start}

\For{$C \in \newCounterPlaySets$}
    \State \Call{append}{$\candidateQueue,\ (\spec, C)$}
\EndFor\label{line:ogis_init_queue_end}

\While{$\candidateQueue \neq \emptyset$}\label{line:ogis_cpgis_start}
    \State $(iSpec, C) \gets \Call{pop}{\candidateQueue}$\label{line:ogis_pop_tuple}
    \State $\learnedSpecs \gets \Call{\learnGuaranteeDegradations}{iSpec, C}$\label{line:ogis_learn_guarantee_degradation}

    \For{$\spec \in \learnedSpecs$}\label{line:ogis_learned_specs_check_start}
        \If{$\Call{\isRealisable}{\spec}$}
            \State \Call{append}{$\learnedSolutions,\ \spec$}\label{line:ogis_add_learned_solution}
        \Else
            \State $\newCounterPlaySets \gets \Call{\extractCounterPlays}{\spec}$\label{line:ogis_cpgis_get_set_cp}
            \For{$C \in \newCounterPlaySets$} \label{line:ogis_iterate_cps_start}
                \If{$\neg\Call{\isRedundantCandidate}{\candidateQueue, \learnedSolutions, \spec, C}$}\label{line:ogis_redundancy_check}
                    \State \Call{append}{$\candidateQueue,\ (\spec, C)$}\label{line:ogis_add_queue}
                \EndIf
            \EndFor\label{line:ogis_iterate_cps_end}
        \EndIf
    \EndFor \label{line:ogis_learned_specs_check_end}
\EndWhile \label{line:ogis_cpgis_end}

\State \Return $\textsc{merge}\left(\Call{filterPreferred}{\learnedSolutions}\right)$ \label{line:ogis_end}
\EndFunction
\end{algorithmic}
\end{algorithm}

%% file: sections/6-evaluation/6.0-introduction.tex
In Sections~\ref{sec:reality_integration} and~\ref{sec:learning_adaptations}, we introduced the adaptation problem for GR(1) specifications together with two adaptation methodologies: a removal-based degradation approach, referred to as the \textit{trivial degradation algorithm}, and a finer-grained syntactic degradation approach based on OGIS loops, and post-processing of the results. Both approaches were shown to compute at least one satisfying degraded specification for the adaptation problem.

In this section, we evaluate the practical benefits and trade-offs of the proposed methodologies. Section~\ref{subsub:methodology_of_evaluation} presents the evaluation methodology, including the preference criteria used to assess the proposed approach. Section~\ref{subsec:experiments} introduces the case studies, describes the experimental setup, and discusses the results obtained. Finally, Section~\ref{subsec:discussion_eval} summarises the findings and discusses their implications.


%% file: sections/6-evaluation/6.1-methodology_of_evaluation.tex
\subsection{Methodology}
\label{subsub:methodology_of_evaluation}

We refer to the degraded specifications computed by the trivial degradation procedure from Section~\ref{subsec:trivial-solution} as \textit{trivial degradations}. Similarly, we refer to the degraded specifications computed by the OGIS-based syntactic degradation approach from Section~\ref{sec:learning_adaptations} as \textit{learned degradations}.


The methodologies introduced in Sections~\ref{sec:reality_integration} and~\ref{sec:learning_adaptations} are guaranteed to produce satisfying degradations for the adaptation problem. Consequently, the evaluation focuses not on correctness, but on assessing the relative quality and usefulness of the degraded specifications produced by each approach. Specifically, we investigate whether the learned degradations are
preferable, according to
Definition~\ref{def:preference_rint_solutions}, to trivial degradations.



To this end, we evaluate the proposed approaches on a collection of case studies adapted from the degradation benchmarks of Buckworth et al. \cite{buckworth2023adapting}. Each case study consists of:
\begin{enumerate}
    \item a specification $\varphi=\langle\mathcal{A},\mathcal{G}\rangle$;
    \item a violation trace $\tau$ such that $\tau\not\models\mathcal{A}$.
\end{enumerate}

For each benchmark, we compute two sets of degraded specifications:
\begin{enumerate}
 \item $\Phi_t=\Call{\trivialDegradation}{\varphi,\tau}$, the set of trivial degradations;
\item $\Phi_l=\Call{\textit{learnPreferredDegradations}}{\varphi,\tau}$, the set of learned degradations.

\end{enumerate}

We are interested in determining whether, for every element in the set of trivial solutions, there exists a preferred specification in the set of learned solutions. Formally, we seek to establish the validity of the following property:
\begin{align*}
\forall\varphi_t\in\Phi_t.\exists\varphi_l\in\Phi_l.\varphi_l\sqsubset \varphi_t \qquad  (\textit{preference criteria})
\end{align*}

\noindent
where $\sqsubset$ denotes preference between two specifications, and is defined in Definition~\ref{def:preference}. If the formula above is not true in any case study, then we are looking at some form of failure of our learning-based algorithm in Section~\ref{sec:learning_adaptations}, which we will investigate on a case-by-case basis.

We encode the specifications in the GR(1) language Spectra, and leverage the tools and code base of the SYNTECH lab for synthesis and diagnosis \cite{maoz2021spectra}. This choice infers an implementation detail of note, with respect to the formulation of response patterns $\Box (p\rightarrow \lozenge s)$. When discussing pure GR(1) syntax, this behaviour is unsupported by the original rules of realisability \cite{bloem2012synthesis}, which would disallow our third syntactic degradation methodology from Definition~\ref{def:conversion_to_response_or_justice}. 

%% file: sections/6-evaluation/6.2-experiments.tex
\subsection{Experimental Results}
\label{subsec:experiments}


The results of the evaluation are presented in Tables~\ref{tab:rq1} and~\ref{tab:ordering}. In Table~\ref{tab:rq1}, the columns \textit{NumT} and \textit{NumL} report the number of degradations returned by the trivial and learning-based adaptation algorithms, respectively. Across all case studies, the learning-based algorithm identifies a single preferred degradation.


Table~\ref{tab:ordering} provides a qualitative comparison between the trivial and learned degradations. The learned adaptations avoid guarantee weakening entirely for the Arbiter and Traffic Single case studies, demonstrating that fine-grained assumption weakening can preserve the complete set of system guarantees. Furthermore, the \textit{Assumptions Order} column of Table~\ref{tab:ordering} shows that the learned degradations consistently retain stronger assumptions than their trivial counterparts. This reflects the locality of the observed assumption violations: since the violation trace affects only a subset of the behaviour captured by the original assumptions, the learning-based approach weakens only the affected portion rather than discarding the assumption altogether.


The \textit{Preference} column of Table~\ref{tab:ordering} reports the ordering between the trivial and learned degradations. In every case study, the learned degradation is preferred according to Definition~\ref{def:preference_rint_solutions}, thereby satisfying the preference criteria. The complete degraded specifications are provided in Appendices~\ref{apx:trivial_solutions} and~\ref{apx:preferred_solutions}.

\begin{table}[h]
\centering

\begin{minipage}[t]{0.33\textwidth}
\centering
\scriptsize
\setlength{\tabcolsep}{3pt}
\begin{tabular}{@{}lcc c@{}}
\toprule
Case Study & NumT & NumL & Outcome \\
\midrule
Arbiter         & 1 & 1 & Success \\
Lift            & 1 & 1 & Success \\
Minepump        & 2 & 1 & Success \\
Traffic Single  & 1 & 1 & Success \\
Traffic Updated & 1 & 1 & Success \\
\bottomrule
\end{tabular}
\caption{Adaptation counts and preference-criteria outcome per case study.}
\label{tab:rq1}
\end{minipage}\hfill
\begin{minipage}[t]{0.63\textwidth}
\centering
\scriptsize
\setlength{\tabcolsep}{2pt}
\begin{tabular*}{\linewidth}{@{\extracolsep{\fill}}lccc@{}}
\toprule
Case study & Assumptions order & Guarantees order & Preference \\
\midrule
Arbiter &
$\mathcal{A}_{o}\prec\mathcal{A}_{l}\prec\mathcal{A}_{t}$ &
$\mathcal{G}_{o}\equiv\mathcal{G}_{l}\prec\mathcal{G}_{t}$ & $\varphi_l\sqsubset \varphi_t$ \\
Lift &
$\mathcal{A}_{o}\prec\mathcal{A}_{l}\prec\mathcal{A}_{t}$ &
$\mathcal{G}_{o}\equiv\mathcal{G}_{l}\equiv\mathcal{G}_{t}$ & $\varphi_l\sqsubset \varphi_t$ \\
Traffic Single &
$\mathcal{A}_{o}\prec\mathcal{A}_{l}\prec\mathcal{A}_{t}$ &
$\mathcal{G}_{o}\equiv\mathcal{G}_{l}\prec\mathcal{G}_{t}$ & $\varphi_l\sqsubset \varphi_t$ \\
Traffic Updated &
$\mathcal{A}_{o}\prec\mathcal{A}_{l}\prec\mathcal{A}_{t}$ &
$\mathcal{G}_{o}\equiv\mathcal{G}_{l}\equiv\mathcal{G}_{t}$ & $\varphi_l\sqsubset \varphi_t$ \\
\hline
Minepump &
$\mathcal{A}_{o}\prec\mathcal{A}_{l}\prec\mathcal{A}_{t}$ &
$\mathcal{G}_{o}\equiv\mathcal{G}_{l}\prec\mathcal{G}_{t1}$ & $\varphi_l\sqsubset \varphi_t$ \\
&
 &
$\mathcal{G}_{o}\equiv\mathcal{G}_{l}\prec\mathcal{G}_{t2}$ & $\varphi_l\sqsubset \varphi_t$ \\
\bottomrule
\end{tabular*}
\caption{Relative weakness ordering between original, learned and trivial specification components, alongside preference between learned and trivial solution(s).}
\label{tab:ordering}
\end{minipage}

\end{table}
We have to note separately on the fact that the minepump example returns a single solution that, semantically, yields stronger guarantees than both trivial solutions we discover. We illustrate the relationship between the trivial specifications and the learned specification in Figure~\ref{fig:minepump_solution_comparison}. We identify the two trivial solutions as $\varphi_{t1}=\langle\mathcal{A}_{t1},\mathcal{G}_{t1}\rangle$ and $\varphi_{t2}=\langle\mathcal{A}_{t2},\mathcal{G}_{t2}\rangle$, and similarly identify the learned solution as $\varphi_{l}=\langle\mathcal{A}_{l},\mathcal{G}_{l}\rangle$. The interested reader may find the raw solutions in Appendix~\ref{apx:trivial_solution_minepump} and \ref{apx:preferred_learned_minepump} for the trivial and learned solutions respectively. We see that the sets of assumptions are comparable to each other, that the assumptions are equivalent between $\varphi_{t1}$ and $\varphi_{t2}$. Finally, we see that the learned assumptions are strictly stronger for the learned solutions, and the learned guarantees are strictly stronger than either of the trivial guarantees. This satisfies the preference criteria, since we needed to show there is a stronger (or equivalent) set of guarantees in the learned solutions for each trivial set of guarantees. Since the learned solution is strictly preferable to every specification in the set of trivial solutions, then we consider the minepump case study a success for the preference criteria, and report it in Table~\ref{tab:rq1}.

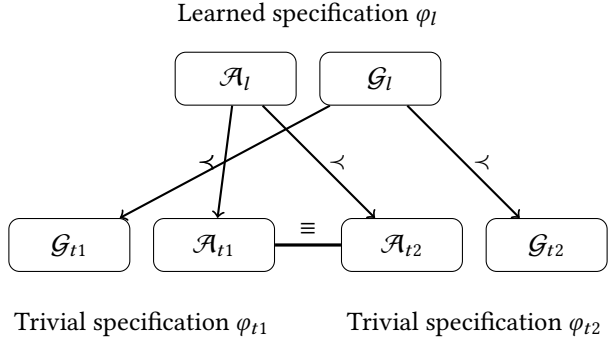
\begin{figure}[t]
\centering
\begin{tikzpicture}[
    spec/.style={
        draw,
        rounded corners,
        minimum width=1.6cm,
        minimum height=0.7cm,
        inner sep=2pt,
        align=center
    },
    dom/.style={->, thick}
]

\coordinate (top) at (0,0);
\coordinate (botL) at (-2.2,-2.2);
\coordinate (botR) at (2.2,-2.2);

\node[spec] (Al) at ($(top)+(-0.95,0)$) {$\mathcal{A}_{l}$};
\node[spec, right=3mm of Al] (Gl) {$\mathcal{G}_{l}$};

\node[spec] (Gt1) at ($(botL)+(-0.95,0)$) {$\mathcal{G}_{t1}$};
\node[spec, right=3mm of Gt1] (At1) {$\mathcal{A}_{t1}$};

\node[spec] (At2) at ($(botR)+(-0.95,0)$) {$\mathcal{A}_{t2}$};
\node[spec, right=3mm of At2] (Gt2) {$\mathcal{G}_{t2}$};

\draw[very thick] (At1) -- node[above] {$\equiv$} (At2);

\draw[dom] (Al) -- node[left] {$\prec$} (At1);
\draw[dom] (Al) -- node[right] {$\prec$} (At2);

\draw[dom] (Gl) -- node[left] {$\prec$} (Gt1);
\draw[dom] (Gl) -- node[right] {$\prec$} (Gt2);

\node[above=2mm of $(Al.north)!0.5!(Gl.north)$]
    {Learned specification $\varphi_l$};

\node[below=4mm of $(Gt1.south)!0.5!(At1.south)$]
    {Trivial specification $\varphi_{t1}$};

\node[below=4mm of $(At2.south)!0.5!(Gt2.south)$]
    {Trivial specification $\varphi_{t2}$};

\end{tikzpicture}
\caption{Comparison between learned and trivial solutions of the minepump case study. Solid arrows denote weakness-based ordering ($\prec$), vertical solid lines denote semantic equivalence ($\equiv$), and lack of any line denotes incomparability.}
\label{fig:minepump_solution_comparison}
\end{figure}

\subsubsection{Discussion on scale}
The procedure we developed in Algorithm~\ref{alg:learningRInt} identifies, over the course of the run, five types of specifications:
\begin{enumerate}
    \item Unrealisable specifications, during the BFS exploration phase, when the algorithm considers a degradation that yields a new syntactically unique specification that is not realisable;
    \item Realisable specifications, during the BFS exploration, when a degradation step modifies either the original specification, or an unrealisable intermediate specification, into a realisable specification. We also call this the set of solutions;
    \item Unique specifications, filtered after the exploration phase. Any group of semantically equivalent specifications is trimmed down to one. This procedure returns a set of pair-wise semantically distinct specifications;
    \item Preferred specifications, when non-preferred solutions (as per Definition~\ref{def:preference_rint_solutions}) are being filtered out of the pool of semantically unique solutions;
    \item Final specifications, when the assumptions and guarantees of the preferred solutions are merged and manipulated to return the best set of solutions.
\end{enumerate}
For each case study, we report the number of specifications encountered of each category in Table~\ref{tab:total}. Every column coincides with a category above. The last column, "Trivial", reports on the amount of specifications our baseline adaptation procedure discovers, for the sake of comparing.

\begin{table}[h]
\centering
\begin{tabular}{l|rr|rrr|c}
\toprule
\multirow{2}*{Case Study} & \multicolumn{2}{c|}{Explored} & \multirow{2}*{Unique} & \multirow{2}*{Preferred} & \multirow{2}*{Final} & \multirow{2}*{Trivial}\\
 & Unrealisable & Realisable &  &  & &  \\
\midrule
Lift & 0 & 21 & 17 & 11 & 1 & 1\\
Traffic Updated & 0 & 21 & 17 & 10 & 1 & 1 \\
Traffic Single & 9 & 3 & 2 & 2 & 1 & 1\\
Arbiter & 285 & 546 & 308 & 6 & 1 & 1 \\
Minepump & 809 & 3762 & 1218 & 17 & 1 & 2 \\
\bottomrule
\end{tabular}
\caption{Amount of specifications discovered at each stage in the algorithm.}
\label{tab:total}
\end{table}

We observe the learning procedure, for each case study, starts with a large set of explored solutions (i.e. realisable specifications), only to converge onto one or a few final solutions. The Explored Unrealisable column highlights the potential of the algorithm to discover a set of satisfiable solutions immediately, through the innate filtering that the ILP backbone of the algorithm does in order to degrade the set of assumptions in a way that the violation trace is satisfied. This fact is shown in the Lift and Traffic Updated rows, where the set of explored unrealisable specifications is $0$. Other than the simple case study of Traffic Single, we see an increase on the realisable specifications side compared to the unrealisable side, with a significant increase for the complex case studies of Arbiter and Minepump. This is likely caused by the traversal methodology: for each syntactically distinct unrealisable specification discovered, there could be several solutions discovered, and several continuations from that point. We produce a snippet of our visualisation program highlighting this fact in Figure~\ref{fig:learning_split}, where at least 3 realisable specifications and at least 3 unrealisable specifications are discovered from node 81 during traversal.
    
\begin{figure}
    \centering
    \includegraphics[width=0.5\linewidth]{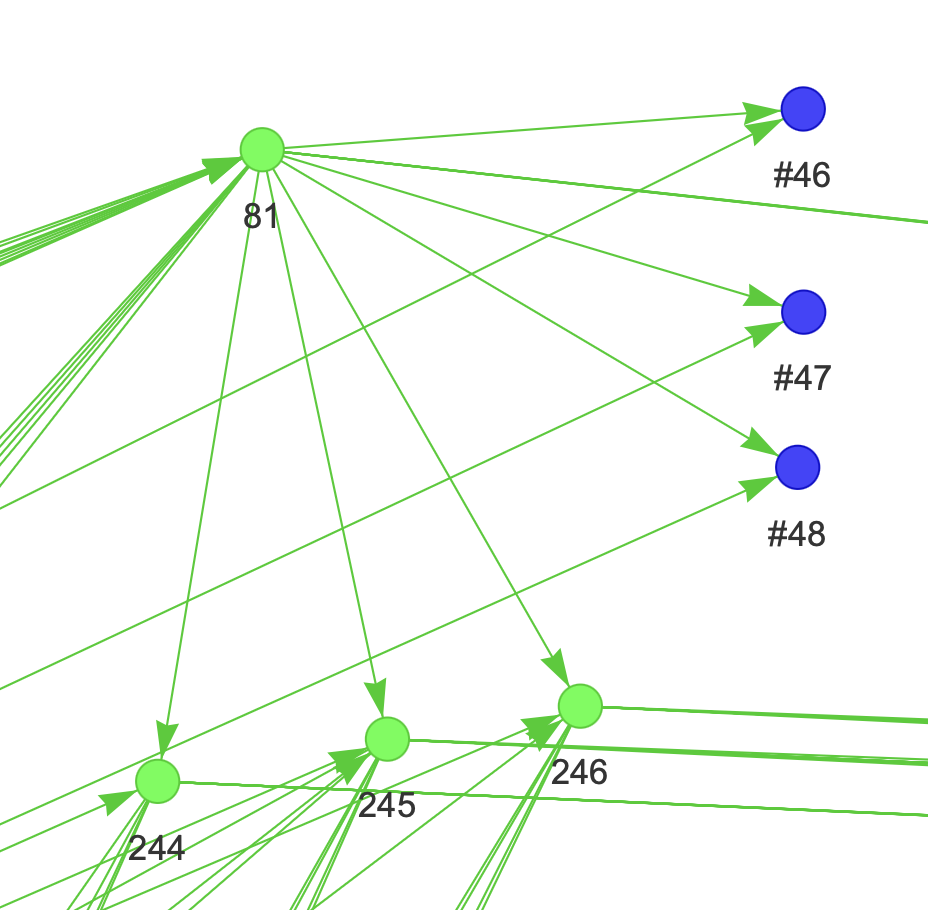}
    \caption{A snippet of the exploration phase during learning for the Minepump case study. Every green node represents a new unrealisable specification discovered, every blue node represents a new realisable specification discovered. An arrow between two nodes represents a learning step.}
    \label{fig:learning_split}
\end{figure}
Out of the set of realisable specifications explored, we find a relatively small drop is followed when filtering out all extraneous specifications, those being any syntactically distinct, but semantically equivalent specification. We see the number reported in the "Unique" column. Although we see a smaller drop in the simpler case studies, for the complex studies, we see a sharper drop between the unique specifications, and the preferred maximal specifications. For all case studies, this number is around $10$ or lower, with the exception of Minepump, which is the only case study containing preferred specifications which require guarantee degradation. Finally, we observe that all case studies can be merged into a single final most preferred specification. 

%% file: sections/6-evaluation/6.3-discussion.tex
\subsection{Summary} 
\label{subsec:discussion_eval}

We observe that, for each case study, the solutions produced by the learning-based solver yield solutions of higher preference compared to the trivial solutions that can be found. By prioritising solutions using the preference definition from Section~\ref{subsec:in_search_for_an_optimal_solution_to_the_rint_problem}, the learning method always produces a solution preferred to the trivial solution(s) discovered. There are two highlights in our results. The first one is the behaviour observed in the Arbiter and Traffic Single case studies, where the baseline trivial solution would require guarantee weakening, yet due to the fine grained learning-based repair method we employ, we discovered guarantee weakening can be skipped, allowing the system to resume its behaviour safely. The second highlight is the final solution in Minepump, which surpasses both trivial solutions in terms of preference, bypassing the necessity of stakeholder involvement by avoiding the need for making a choice.

Although the learning-based algorithm reliably discovers a more preferred solution compared to the trivial solution(s) for every case study, this comes at the cost of exploring a large space of candidate specifications, as shown in Table~\ref{tab:total}. For the more complex case studies in particular, the number of realisable specifications explored before converging to a final solution is substantial. The current methodology is best suited to scenarios where there are no hard time constraints for computing an adaptation. More constraining scenarios motivate the need for heuristics capable of guiding the search towards preferred solutions more directly.

%% file: sections/6-evaluation/6.z-conclusion.tex

%% file: sections/7-related_work/index.tex

Our work is rooted in behaviour-based adaptation as opposed to configuration-based~\cite{braberman2015morph} since we focus on modifying statements that refer to the intended behaviour of the system or the assumed behaviour of the  environment, in terms of safety and liveness properties, in order to salvage as much of the original requirements as possible. 

We continue in the line of work of syntactic degradation-based adaptations of specifications, earlier explored in \cite{buckworth2023adapting,zhang2025behavioral}. Both of their works enable controllers synthesised from a GR(1) specification to function in the context of encountering an environment violation. However, to the best of our knowledge, we are the first to investigate what makes a specification preferred in the context of this problem. 

Although the quality of assume-guarantee specifications has been analysed from a robustness-during-deployment perspective \cite{shalom2023my,zhang2025behavioral}, which correctly suggests that weakest environment assumptions are always preferred \cite{cobleigh2003learning,lomuscio2013assume,nam2006learning} with respect to the logic that the specified system may operate in as broad a set of environment dynamics as possible, the adaptation community seems to have overlooked the potential for over-fitting (too weak) environmental assumptions, and how online adaptation can lead to robustness vulnerabilities. In the avenue where formal methods literature has started learning environment model dynamics from data \cite{keegan2020control,gaaloul2020mining}, and generating assumptions (through learning \cite{cobleigh2003learning}, or by hand \cite{ma2023using}) to enable specific guarantees (rather than the ideal methodology, where guarantees of the system/machine are based on assumptions of the environment/world \cite{jackson1995world}), the chances of faults are increased. And as we illustrated in Section~\ref{subsec:in_search_for_an_optimal_solution_to_the_rint_problem} and Figure~\ref{fig:overfitting_assumption_weakening}, such faults lead to divergences in the perceived model of the world and the actual. These divergences can cascade towards a representation that, although it becomes inclusive of (most) observed environment behaviour, may lead to degrading the system guarantees unnecessarily.

Other than Buckworth's work built on the OGIS framework \cite{jha2017theory}, Cavezza et al.'s work on assumption refinement has been a strong influence for the specification preference criteria \cite{cavezza2017interpolation}, as well as their methodology of search space exploration \cite{cavezza2020minimal}. The refinement algorithm they propose also searches for the minimal refinements, explored in a BFS manner to discover the best solutions early during the run of the algorithm. We observe that our preference criteria for degraded assumptions is the symmetric of their work, and extend the minimal BFS exploration to cascade to guarantee degradation (when necessary).

Work in the evaluation of the strength of a solution for refinement and degradation problem has always dealt with the following context-related issue: there is no trivial way of representing the quality of a specification numerically. There have been several avenues and angles for the fair and complete analysis of such approaches.

One such avenue is the analysis of \textit{robustness} under assumption violations. Bloem et al.~\cite{bloem2014synthesizing} have developed the notion of k-robustness, that attempts to minimise the cost of failure of the guarantees of the system by transforming the specification into an \textit{error specification}, where each formula is given a cost function, and the overall cost of system errors needs to be smaller than k times the cost of environment errors. This design has the potential of generating resilient controllers in the face of failures, that are able to rebound should failures occur. Our work differs from this method, in the sense that we use the failures to learn an explicit new specification, never allowing or expecting deviations from the written specification to occur. Future work could look into leveraging the controller encoding in situations where learning takes longer, and potentially using the cost automata formulation for our own formulas, in order to only weaken behaviour that yields higher error cost. We believe cost automata has the potential to relay deeper meaning during the process of learning a weakening, as opposed to providing a rank to each individual requirement \cite{alur2008ranking}.

An alternative controller resilience problem in the face of assumption violations has been proposed in Ehlers et al.~\cite{ehlers2014resilience}, where $(k,b)$-resilience defines the concept that a controller of a GR(1) specification is resilient to glitches (environment failures) for at most k steps, separated by periods of at least b steps of no glitches, and the amount of glitches is finite. Since we address environment errors that may occur infinitely often, we cannot define our solutions in terms of $(k,b)$-sanity.

The closest approach to the quantitative measurement angle we consider is the \textit{weakness measurement} of Cavezza et al.~\cite{cavezza2021weakness}. They express a specification's set of assumptions as a triple of elements: the entropy of the invariants, the Hausdorff dimension of the invariants, and the Hausdorff dimension of the fairness complement of the assumptions. This measurement has the quality that it preserves the ordering between two (or more) specifications, in terms of which formula is weaker/stronger than the other. The problem with this measurement comes in its structure and practicality. Unfortunately, the format of a lexicographical ordering on a triple makes it impossible to accurately aggregate information on multiple solutions at once. Furthermore, it has not been shown that the valuation would have any correlation to an order of magnitude of the weakness between solutions, although a case could be made that the usage of entropy and/or Hausdorff dimensions can be used as a metric over nested LTL formulas (i.e. $\varphi\rightarrow \psi$ or $\psi\rightarrow \varphi$). Nonetheless, a meaningful quantitative metric to represent and compare non-nested LTL formulas is still a literature gap, which consequently, prevents the creation of a quantitative metric for assume-guarantee specifications based on LTL, like GR(1).

The field of runtime adaptation of systems is well explored in the space of requirement engineering \cite{calinescu2012self, calinescu2017engineering,chu2024integrating}, where quantitative model checking techniques \cite{filieri2011run, kwiatkowska2007quantitative} are employed to verify models that satisfy a set of requirements written in probabilistic computational-tree logic (PCTL) or similar. 

%% file: sections/8-conclusion/index.tex
Reactive systems synthesised from GR(1) specifications depend on the behaviour of the environment to hold. During deployment, a failure in environment behaviour has the potential to silently break most guaranteed behaviour of the system. We highlight the tuple consisting of an assume/guarantee specification and a behaviour trace violating the assumptions of the specification as an important adaptation problem. We propose a preference-based ordering with respect to what makes a good solution to this problem, then propose two algorithms for the discovery of solutions. The first algorithm represents a simple baseline methodology, where each individual formula is treated like a component in a set-theoretical problem, and a solution can be found through the removal of formulas from the specification. We show this algorithm removes the smallest amount of formulas from the original specification. The second algorithm is a finer grained solution, that leverages syntactic manipulation and inductive learning to iteratively degrade the violating formulas until they reach a realisable specification that no longer violates the observed environment behaviour. We show the second algorithm discovers solutions of higher quality than the baseline algorithm.

This work has made a case in Section~\ref{subsec:in_search_for_an_optimal_solution_to_the_rint_problem} for the desirability of the strongest degraded specification as the result of adaptation. To the best of our knowledge, the suggestion that the strongest possible set of assumptions would be preferred has not been previously proposed. Moreover, it appears to contradict the logic of robustness: that a system that is able to achieve its goals (strongest system guarantees) in the highest amount of scenarios (weakest environment assumptions) is more robust. We propose a finer grained definition to robustness, that considers partial failure as a certainty that a system needs to be ready for: an adaptive system that is able to salvage the widest set of goal-related behaviours (strongest possible system guarantees) in \textit{any} scenario (any environment condition) is more robust.

Future work should further analyse the notion of robustness considered in this work. This work showcased adapted results as being favourable under the preference criteria, therefore merely hinting at an extended desirability in the wider context of adaptive systems. Further study should focus on the development of realistic scenarios that systems can leverage to validate their ability to adapt, and their ability to retain their operational guarantees.

One thing to note regarding the preference criteria is how maximal system guarantees are equally preferred. This was intentional, since discriminating between equally preferred guarantees requires additional information beyond the logical formalism alone. Future work will consider ways for a stakeholder to provide discrimination methods for system guarantees when adaptation is necessary. Criticality of guarantees has been explored in the context of synthesis under uncertainty \cite{yang2024synthesizing}, and in the context where all guarantees are equally critical, a prioritisation or weighting method could be employed.

The problem with priorities and weights when applied to guarantees is that, in more complex scenarios, specifically industrial ones, the space of guarantees is vast, and the stakeholders involved in the process may not be sufficiently aware of the low level details of the adaptation process to propose a satisfying weighting or ordering of the guarantees. A relatively novel branch of work has the potential to circumvent that by introducing semantic meaning to the formulas within the specification \cite{bolt2019interacting} with respect to a set of principles that can be translated from the natural language of the stakeholders \cite{kaushik2025universal} (i.e. the client of the system, the manufacturer's priorities, the laws of the country etc.). This aggregated knowledge, including the current specification and the violation trace, creates a complete context for an adaptive system to generate its own set of priorities, ground in the priorities of the relevant stakeholders.

Once a priority between system guarantees is in place, the next step will be to use it to guide the search towards only such solutions that follow the priority system. The learning-based algorithm, in its current state, generates and explores a vast space of potential adaptation candidates, which takes a significant amount of time. This inhibits the application of our methodology at runtime, meaning it can only be used proactively in practical scenarios, by estimating future failures, and pre-computing suitable adapted specifications before there is a need to adapt. To allow this level of controlled adaptation to occur reactively, rules and heuristics need to be developed, with the purpose of salvaging the most important guarantees, and finding a realisable solution in the least amount of steps. We have already discovered that the choice of counter-trace and counter-strategy allows us to manipulate in advance what system behaviour we want to salvage, and therefore, which guarantee is considered for degradation by the learner during OGIS. This methodology, alongside others from classical graph traversal theory (e.g. Monte-Carlo Tree Search), may be used to enable our adaptation methodology to function in reactive contexts, and therefore, enable systems to act in a principled manner under any circumstance.

%% file: appendices/A-proofs/A.1-counter_example_disjunction_of_assumptions.tex
\subsection{Counter Example Disjunction of Assumptions}
\label{apx:counter_example}

We present below a counter-example to the following claim.

\begin{quote}
For any two solutions to an adaptation problem $(\varphi,\tau)$, $\varphi_1 = \langle \mathcal{A}_1, \mathcal{G} \rangle$ and $\varphi_2 = \langle \mathcal{A}_2, \mathcal{G} \rangle$, if their guarantees are equivalent, then there exists a third solution $\varphi' = \langle \mathcal{A}', \mathcal{G} \rangle$ such that $\mathcal{A}' \equiv \mathcal{A}_1 \lor \mathcal{A}_2$, and $\varphi'$ is also a solution to the same adaptation problem.
\end{quote}

Formally, the claim can be stated as follows:
\[
\varphi_1, \varphi_2 \in \mathrm{AllSolutions}(Adapt(\varphi, \tau)) \;\Rightarrow\;
\varphi' \in \mathrm{AllSolutions}(Adapt(\varphi, \tau)),
\]
where $\varphi' = \langle \mathcal{A}_1 \lor \mathcal{A}_2, \mathcal{G} \rangle$.

\begin{proof}

We take note that we only need to show that there is such a case where $\mathcal{A}_1\lor \mathcal{A}_2$ does not have an equivalent for the syntactical structure of the set of assumptions of a GR(1) specification.

We propose the following specification trace pair as our $\varphi$ and $\tau$ respectively:
\label{apx:counter_example_spec_v0}
\input{code/spectra_specs/counter_example/v0}
\fbox{%
\begin{minipage}{\textwidth}
\[
\tau = [(\textcolor{red}{ \neg a,\neg b},\textcolor{green!60!black}{ c})]
\]%
\end{minipage}
}
\captionof{figure}{Counter example specification and trace}

We do highlight that, naturally, this spec by itself has very little practicality, but it would not be outside of the real of possibility for this spec to belong to a more complex specification, handling a larger amount of reactive interactions. This is the smallest self contained example we envisioned at the time of writing this paper. We note the specification above is realisable, a fact that has been verified with the usage of the SYNTECH toolbox plugins within the Eclipse development environment. We further note that assumption a1 obviously violates the trace $\tau$.

There are two immediate solutions that may be proposed without relying on guarantee degradation. $\varphi_1$ is presented below:

\label{apx:counter_example_spec_v1}
\input{code/spectra_specs/counter_example/v1}
\captionof{figure}{Counter example specification repair 1}

And $\varphi_2$ is presented below:

\label{apx:counter_example_spec_v2}
\input{code/spectra_specs/counter_example/v2}
\captionof{figure}{Counter example specification repair 2}

We have verified again realisability using the SYNTECH toolbox. We identify the two proposed assumption degradations are $G(a\lor b)$ and $GF(a)$ respectively, which can be visualised in Figure~\ref{fig:counterexample-structure}, where we highlight the relationship between the different repair proposals.

We now return to the claim we started this section with, that there is always a disjunction between the assumptions. We clearly see though that $G(a\lor b)\lor GF(a)$ cannot be meaningfully decomposed or rewritten to satisfy the GR(1) syntactic structure. Therefore, we have a contradiction to the claim, invalidating it.

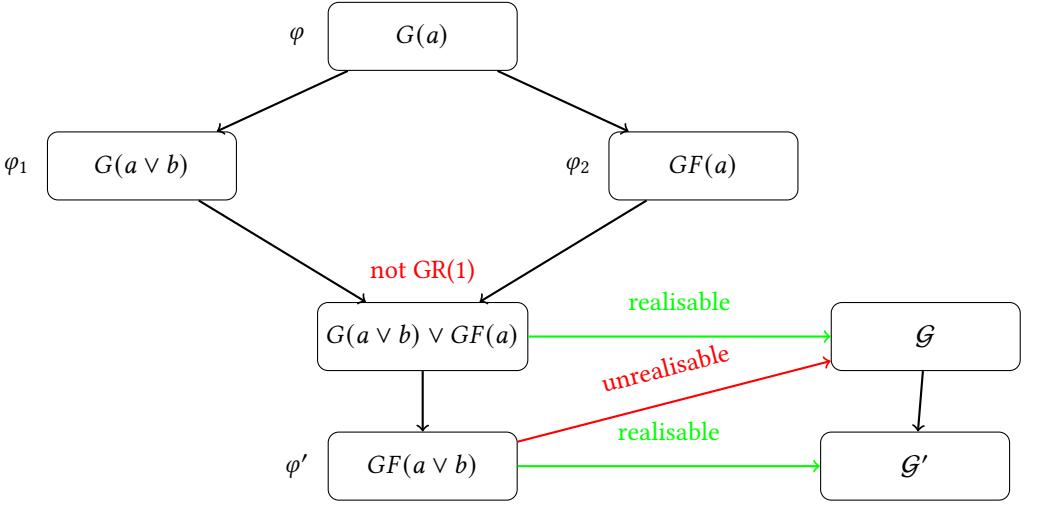
\begin{figure}[t]
\centering
\input{diagrams/counterexample-diagram-tikz}

\caption{Implication-based ordering of the counter-example specifications.}
\label{fig:counterexample-structure}
\end{figure}

As a bonus, let us try to find some formula weaker than both $\mathcal{A}_1$ and $\mathcal{A}_2$. The strongest one we may find that is still GR(1) compliant is $GF(a\lor b)$. The relationship is visualised in Figure~\ref{fig:counterexample-structure}, and the resulting specification $\varphi'$ is illustrated below:

\label{apx:counter_example_spec_v3}
\input{code/spectra_specs/counter_example/v3}
\captionof{figure}{Counter example specification repair 3 - UNREALISABLE}

We find, through verification, that this specification is no longer realisable, therefore requiring guarantee weakening for a complete repair. Therefore, we highlight using a clear example the danger of attempting to degrade our assumption set prematurely and excessively.

\end{proof}

%% file: code/spectra_specs/counter_example/v0.tex
\begin{lstlisting}[style=SpectraStyle]
module CounterExample_v0

env boolean a;
env boolean b;
sys boolean c;

assumption -- a1
	G(a);
guarantee -- g1.1
	G((!a&!b)|(!a&PREV(!b))->!c);
guarantee -- g1.2
	G(a->c);
guarantee -- g2
	G(PREV(b)&b->c);
guarantee -- g3
	GF(c);
\end{lstlisting}

%% file: code/spectra_specs/counter_example/v1.tex
\begin{lstlisting}[style=SpectraStyle]
module CounterExample_v1

env boolean a;
env boolean b;
sys boolean c;

assumption -- a1
	G(a|b);
guarantee -- g1.1
	G((!a&!b)|(!a&PREV(!b))->!c);
guarantee -- g1.2
	G(a->c);
guarantee -- g2
	G(PREV(b)&b->c);
guarantee -- g3
	GF(c);
\end{lstlisting}

%% file: code/spectra_specs/counter_example/v2.tex
\begin{lstlisting}[style=SpectraStyle]
module CounterExample_v2

env boolean a;
env boolean b;
sys boolean c;

assumption -- a1
	GF(a);
guarantee -- g1.1
	G((!a&!b)|(!a&PREV(!b))->!c);
guarantee -- g1.2
	G(a->c);
guarantee -- g2
	G(PREV(b)&b->c);
guarantee -- g3
	GF(c);
\end{lstlisting}

%% file: diagrams/counterexample-diagram-tikz.tex
\begin{tikzpicture}[
    node distance=0.8cm and 1.2cm,
    every node/.style={draw, rounded corners, minimum width=2.5cm, minimum height=0.9cm, align=center},
    arr/.style={->, thick},
    unrealisable/.style={->, thick,color=red},
    realisable/.style={->, thick,color=green},
    labelnode/.style={draw=none, rectangle, inner sep=1pt}
]

\node (Ap) {$G(a)$};

\node (A1pp) [below left=of Ap] {$G(a\lor b)$};
\node (A2pp) [below right=of Ap] {$GF(a)$};

\node (Appp) [below=1.8cm of $(A1pp)!0.5!(A2pp)$] {$G(a\lor b)\lor GF(a)$};

\node (Apppp) [below=of Appp] {$GF(a\lor b)$};

\node (G) [right=4cm of Appp] {$\mathcal{G}$};
\node (Gp) [right=4cm of Apppp] {$\mathcal{G}'$};

\draw[arr] (Ap) -- (A1pp);
\draw[arr] (Ap) -- (A2pp);
\draw[arr] (A1pp) -- (Appp);
\draw[arr] (A2pp) -- (Appp);
\draw[arr] (Appp) -- (Apppp);

\draw[arr] (G) -- (Gp);

\draw[realisable] (Appp) -- node[labelnode, above] {realisable} (G);
\draw[unrealisable] (Apppp) -- node[labelnode, sloped, above, allow upside down] {unrealisable} (G);
\draw[realisable] (Apppp) -- node[labelnode, above] {realisable} (Gp);

\node[labelnode, left=0.4cm of $(Ap)$] {$\varphi$};
\node[labelnode, left=0.4cm of $(A1pp)$] {$\varphi_1$};
\node[labelnode, left=0.4cm of $(A2pp)$] {$\varphi_2$};
\node[labelnode, above=0.4cm of $(Appp)$] {\textcolor{red}{not GR(1)}};
\node[labelnode, left=0.4cm of $(Apppp)$] {$\varphi'$};


\end{tikzpicture}

%% file: code/spectra_specs/counter_example/v3.tex
\begin{lstlisting}[style=SpectraStyle]
module CounterExample_v3_unrealisable

env boolean a;
env boolean b;
sys boolean c;

assumption -- a1
	GF(a|b);
guarantee -- g1.1
	G((!a&!b)|(!a&PREV(!b))->!c);
guarantee -- g1.2
	G(a->c);
guarantee -- g2
	G(PREV(b)&b->c);
guarantee -- g3
	GF(c);
\end{lstlisting}

%% file: appendices/A-proofs/A.2-proof_of_soundness_of_baseline_algorithm.tex
\subsection{Proof of soundness of baseline algorithm}
\label{apx:proof_soudness_trivial_algorithm}
We prove Theorem~\ref{thm:soundness_minimality_optimal_trivial_solution} through the intermediate results stated in Lemmas~\ref{lemma:trivial_a} and~\ref{lemma:trivial_g} below.

\begin{lemma}[Minimal Correction of Assumptions]
\label{lemma:trivial_a}
Let $A_\tau\subseteq \mathcal{A}$ be the violation set induced by $\tau$. Then:
\begin{enumerate}
\renewcommand{\labelenumi}{(\roman{enumi})}
    \item $\tau \models \mathcal{A}\setminus A_\tau$;

    \item $A_\tau$ is minimal with respect to this property.
\end{enumerate}
\end{lemma}

\begin{proof}
\label{proof:trivial_a}
Let $\mathcal{A}$ be a set of assumptions and let $\tau$ be a violation trace such that $\tau\not\models\mathcal{A}$. We show that the violation set $A_\tau\subseteq\mathcal{A}$ is the minimal set of assumptions whose removal yields a specification satisfied by $\tau$.

\begin{enumerate}
\renewcommand{\labelenumi}{(\roman{enumi})}

\item Show that $\tau\models\mathcal{A}\setminus A_\tau$.

Assume, for contradiction, that $\tau\not\models\mathcal{A}\setminus A_\tau$. Then there exists some assumption $a_i\in\mathcal{A}\setminus A_\tau$ such that $\tau\not\models a_i$. By definition of the violation set, $a_i\in A_\tau$, which contradicts the assumption that $a_i\in\mathcal{A}\setminus A_\tau$. Therefore, $\tau\models\mathcal{A}\setminus A_\tau$.

\item Show that $A_\tau$ is minimal.

Let $A'\subset A_\tau$. Then there exists some assumption $a_i\in A_\tau$ such that $a_i\notin A'$. Since $a_i\in A_\tau$, by definition of the violation set we have $\tau\not\models a_i$. Consequently, $\tau\not\models\mathcal{A}\setminus A'$, since $a_i$ remains in the degraded assumption set. Hence no strict subset of $A_\tau$ satisfies the required property, proving minimality.
\end{enumerate}
\end{proof}

\begin{lemma}[Minimal Correction of Guarantees]
\label{lemma:trivial_g}
Let $\mathcal{C}=\{\mathcal{G}^{uc}_1,\ldots,\mathcal{G}^{uc}_k\}$ be the set of unrealisable cores of $\langle \mathcal{A}\setminus A_\tau,\mathcal{G}\rangle$, and let $\mathcal{H}_{min}(\mathcal{C})$ be a minimal hitting set of $\mathcal{C}$. Then:
\begin{enumerate}
\renewcommand{\labelenumi}{(\roman{enumi})}
    \item the specification $\langle \mathcal{A}\setminus A_\tau,\mathcal{G}\setminus\mathcal{H}_{min}(\mathcal{C})\rangle$ is realisable;

    \item $\mathcal{H}_{min}(\mathcal{C})$ is minimal with respect to this property.
\end{enumerate}
\end{lemma}

\begin{proof}
\label{proof:trivial_g}
\begin{enumerate}
\renewcommand{\labelenumi}{(\roman{enumi})}

\item Show that $\langle \mathcal{A}\setminus A_\tau,\mathcal{G}\setminus\mathcal{H}_{min}(\mathcal{C})\rangle$ is realisable.

Assume, for contradiction, that
$\langle \mathcal{A}\setminus A_\tau,\mathcal{G}\setminus\mathcal{H}_{min}(\mathcal{C})\rangle$
is unrealisable. Then there exists some unrealisable subset
$\mathcal{G}'\subseteq\mathcal{G}\setminus\mathcal{H}_{min}(\mathcal{C})$
such that
$\langle \mathcal{A}\setminus A_\tau,\mathcal{G}'\rangle$
is unrealisable.

Since $\mathcal{G}'$ is unrealisable, it must contain some unrealisable core
$\mathcal{G}^{uc}_i\in\mathcal{C}$.
However, by definition of a hitting set,
$\mathcal{H}_{min}(\mathcal{C})\cap\mathcal{G}^{uc}_i\neq\emptyset$.
Therefore,
$\mathcal{G}^{uc}_i\nsubseteq\mathcal{G}\setminus\mathcal{H}_{min}(\mathcal{C})$,
contradicting the assumption that
$\mathcal{G}'\subseteq\mathcal{G}\setminus\mathcal{H}_{min}(\mathcal{C})$.
Hence,
$\langle \mathcal{A}\setminus A_\tau,\mathcal{G}\setminus\mathcal{H}_{min}(\mathcal{C})\rangle$
must be realisable.

\item Show that $\mathcal{H}_{min}(\mathcal{C})$ is minimal.

Let
$\mathcal{H}'\subset \mathcal{H}_{min}(\mathcal{C})$.
Since $\mathcal{H}_{min}(\mathcal{C})$ is a minimal hitting set, there exists some unrealisable core
$\mathcal{G}^{uc}_i\in\mathcal{C}$
such that
$\mathcal{G}^{uc}_i\cap\mathcal{H}'=\emptyset$.
Consequently,
$\mathcal{G}^{uc}_i\subseteq\mathcal{G}\setminus\mathcal{H}'$.

Since
$\langle \mathcal{A}\setminus A_\tau,\mathcal{G}^{uc}_i\rangle$
is unrealisable, it follows that
$\langle \mathcal{A}\setminus A_\tau,\mathcal{G}\setminus\mathcal{H}'\rangle$
is also unrealisable. Therefore, no strict subset of
$\mathcal{H}_{min}(\mathcal{C})$
is sufficient to restore realisability, proving minimality.
\end{enumerate}

Together, Properties~(i) and~(ii) establish that
$\mathcal{H}_{min}(\mathcal{C})$
is the minimal set of guarantees whose removal restores realisability of
$\langle \mathcal{A}\setminus A_\tau,\mathcal{G}\setminus\mathcal{H}_{min}(\mathcal{C})\rangle$.
\end{proof}

From Lemmas~\ref{lemma:trivial_a} and~\ref{lemma:trivial_g}, Theorem~\ref{thm:soundness_minimality_optimal_trivial_solution} follows directly.

%% file: appendices/A-proofs/A.3-finiteness_of_semantic_equivalence_classes.tex
\subsection{Finiteness of Semantic Equivalence Classes}
\label{sec:finiteness_of_semantic_equivalence_classes}
As part of our proof that the adaptation procedure in Algorithm~\ref{alg:learningRInt} terminates, we propose the following theorem:

\begin{theorem}\label{thm:finite-semantic-classes}
The set $\mathcal{F}$ of semantic equivalence classes of assertions over a finite set of Boolean variables $\mathcal{X}$ is finite.
\end{theorem}

\begin{proof}
    A formula $\varphi$ over $\mathcal{X}$ defines a Boolean function:
    \[
        f_\varphi : 2^{\mathcal{X}} \to \{0,1\}.
    \]
    Since $|\mathcal{X}| = n$, there are $2^n$ possible truth assignments, so there are:
    \[
        |\{ f : 2^{\mathcal{X}} \to \{0,1\} \}| = 2^{2^n}
    \]
    Boolean functions. Two formulas $\varphi, \psi$ are semantically equivalent if and only if they define the same Boolean function, so:
    \[
        |\mathcal{F}| = 2^{2^n},
    \]
    which is finite.
\end{proof}

\subsubsection{Finiteness of Weaker Assertions}

\begin{theorem}[Finiteness of Weaker Assertions]\label{thm:finite_assertions}
Given a GR(1) specification $\langle A, G \rangle$ over a finite set of Boolean variables $\mathcal{X}$,
the sets of semantically distinct weaker environment assumptions than $A$ and weaker system guarantees than $G$ are finite.
\end{theorem}

\begin{proof}
    A GR(1) assumption module $A$ and guarantee module $G$ each consist of:
    \begin{itemize}
        \item Initial conditions: $\{\varphi^{init}_i\}$
        \item Safety constraints: $\{\mathbf{G} \varphi^{safe}_j\}$
        \item Liveness constraints: $\{\mathbf{GF} \varphi^{live}_k\}$
    \end{itemize}
    where each $\varphi$ is a propositional formula over the finite variable set $\mathcal{X}$.

    By Theorem~\ref{thm:finite-semantic-classes}, the set of semantically distinct propositional formulas over $\mathcal{X}$ is finite.

    Weaker assumptions correspond to replacements:
    \[
        \forall \sigma \in \Sigma^\omega: \sigma \models A \implies \sigma \models A',
    \]
    and weaker guarantees correspond to:
    \[
        \forall \sigma \in \Sigma^\omega: \sigma \models G' \implies \sigma \models G.
    \]
    These correspond to semantic weakening of each $\varphi$.

    Since:
    \begin{itemize}
        \item There are finitely many semantically distinct propositional formulas.
        \item Each GR(1) module has a finite number of formula slots.
    \end{itemize}
    The total number of semantically distinct weakenings of $A$ and $G$ is finite.
\end{proof}

%% file: appendices/A-proofs/A.4-soundness_of_algorithm.tex
\subsection{Soundness of Algorithm}
\label{apx:proof_soundness_cegis}

To show our algorithm is sound, we prove the following theorem:

\begin{theorem}[Termination of Algorithm~\ref{alg:learningRInt}]
    \label{th:rint}
    Given the tuple $(\varphi,\tau)$, where $\varphi=\langle \mathcal{A},\mathcal{G}\rangle$ is a realisable specification, and $\tau\in \Sigma^*$ is a violation trace of the specification, the $\Call{adapt}{\varphi,\tau}$ method will always terminate, with the solution $\varphi'=\langle\mathcal{A}',\mathcal{G}'\rangle$, conforming to the properties at Definition~\ref{def:rca}.
\end{theorem}

We show in Section~\ref{sec:learning_adaptations} that an immediate solution exists for the discovery of a specification $\varphi'=\langle \mathcal{A}',\mathcal{G}\rangle$ as a solution to the assumption-only reality integration problem. Therefore, to prove Theorem~\ref{th:rint}, we only need to prove the following lemma:

\begin{lemma}[Termination for the second half of Algorithm~\ref{alg:learningRInt} using Guarantee Weakening]
    \label{lem:rint-gw}
    Given the tuple $(\varphi',\tau)$, where $\varphi'=\langle \mathcal{A}',\mathcal{G}\rangle$ is an unrealisable specification, and $\tau\in \Sigma^*$ is a violation trace of specification $\varphi$, the Guarantee Weakening task of the $\Call{adapt}{\varphi,\tau}$ method will always terminate, with the solution $\varphi''=\langle\mathcal{A}',\mathcal{G}'\rangle$, conforming to the properties at Definition~\ref{def:rca}.
\end{lemma}

\begin{proof}
    We consider, without loss of generality, an intermediary step within the Guarantee Weakening problem, and use the structure of the Counter-Play Guided GW problem at \ref{def:rint_g_ct}:
    \[
        (\varphi',\tau,cts)
    \]
    Since our solution space is composed of specifications whose guarantees are strictly weaker than $\varphi'$, we know that our solution space is always finite, as proven for Theorem~\ref{thm:finite_assertions}. Therefore, in order to show termination, we will prove the following lemma:
    \begin{lemma}
        \label{lem:rint_gw2}
        the set of solutions of the intermediate problem is always strictly larger than the set of solutions for any future intermediate problem:
        \begin{align*}
            \forall \varphi'\in Adapt_{GW}(\varphi,\tau,cts).\forall ct'\in\Sigma^* \text{ s.t. } ct'\models\mathcal{A}' \text{ and } ct'\not\models\mathcal{G}'\\Adapt_{GW}(\varphi',\tau,cts\doubleplus ct')\subset Adapt_{GW}(\varphi,\tau,cts)
        \end{align*}
    \end{lemma}

    \begin{proof}
        Take an arbitrary $\varphi''=\langle\mathcal{A}',\mathcal{G}''\rangle$ such that:
        $$\varphi''\in Adapt_{GW}(\varphi',\tau,cts\doubleplus ct')$$
        If $\varphi''$ is realisable, the solution is found and there is no need for further iteration, so without loss of generality, let us assume $\varphi'$ is unrealisable. Then, since it is a solution of a CTGW problem, the following properties hold, as per Definition~\ref{def:rint_g_ct}:
        \begin{enumerate}
            \item $\tau \models \mathcal{A}'$
            \item $\forall \sigma \in \Sigma^\omega$, $\sigma\models\mathcal{G}'\rightarrow\sigma\models\mathcal{G}''$
            \item $\forall \tau^c \in cts \doubleplus ct$, $\tau^c\models \mathcal{G}''$
        \end{enumerate}
        Due to property 3, it also holds that
        $$
        \forall \tau^c \in cts, \tau^c\models \mathcal{G}''
        $$
        Furthermore, due to property 2, and the equivalent property of $\varphi'\in Adapt_{GW}(\varphi',\tau,cts\doubleplus ct')$, we can extract the following property:
        $$
        \forall \sigma \in \Sigma^\omega\text{, }\sigma\models\mathcal{G}\rightarrow\sigma\models\mathcal{G}''
        $$
        Therefore, alongside property 1, we have shown that
        $$\varphi''\in Adapt_{GW}(\varphi,\tau,cts)$$
        Which proves that
        $$
        Adapt_{GW}(\varphi',\tau,cts\doubleplus ct')\subseteq Adapt_{GW}(\varphi,\tau,cts).
        $$
        However, since we know $ct'\not\models \mathcal{G}'$, but  $ct'\models \mathcal{G}''$, we conclude that
        $$
        \varphi'\not\in Adapt_{GW}(\varphi',\tau,cts\doubleplus ct'),
        $$
        which proves Lemma~\ref{lem:rint_gw2}
    \end{proof}

The proof above, in turn, proves Lemma~\ref{lem:rint-gw}.
\end{proof}

%% file: appendices/B-code/B.1-helper_functions.tex
\subsection{Helper Functions}
\label{apx:helper_functions}
\input{code/filter_minimal_formulas}
\input{code/merge_rint_algorithm}

%% file: code/filter_minimal_formulas.tex
\begin{algorithm}[H]
\caption{Filter Minimal Formulas}
\label{alg:filterMinimalFormulas}
\begin{algorithmic}[1]
\Function{filterMinimal}{$formulas$}
\State $filteredFormulas \gets formulas[0]$
\For{$formula \in formulas[1:]$}
    \State $isMinimal \gets \top$
    \For{$filteredFormula \in filteredFormulas$}
        \If{$formula \prec filteredFormula$}
            \State $\Call{remove}{filteredFormulas, filteredFormula}$
            \State $isMinimal \gets \bot$
            \State break;
        \EndIf
    \EndFor
    \If{$isMinimal$}
        \State $\Call{add}{filteredFormulas, formula}$
    \EndIf
\EndFor
\State \Return $filteredFormulas$
\EndFunction
\end{algorithmic}
\end{algorithm}

%% file: code/merge_rint_algorithm.tex
\begin{algorithm}
\caption{Merging Preferred Specifications}
\label{alg:mergeSpecs}
\begin{algorithmic}[1]
\Function{merge}{$specs$}

\State $mergedSpec \gets \langle\top,\top\rangle$

\For{$spec \in specs$}
    \State $mergedSpec.{\mathcal{A}} \gets mergedSpec.{\mathcal{A}} \cup spec.{\mathcal{A}}$
    \State $mergedSpec.{\mathcal{G}} \gets mergedSpec.{\mathcal{G}} \cup spec.{\mathcal{G}}$
\EndFor

\If{$\Call{isRealisable}{mergedSpec}$}
    \State \Return $[mergedSpec]$
\Else
    \State \Return $\Call{trivialGW}{mergedSpec}$
\EndIf
\EndFunction
\end{algorithmic}
\end{algorithm}

%% file: appendices/C-case_studies/C.1-size_of_specifications.tex
\subsection{Size of Specifications}
\label{apx:size_of_specifications}
There are 5 specifications we use in our evaluation, each specification is realisable, and has a pairing of a violation trace. We present in Table~\ref{tab:spec_sizes} the sizes of the specifications in a similar way to \cite{firman2020performance}. The sizes are given as follows:
\begin{itemize}
    \item $|\mathcal{X}|$ - the number of environment variables
    \item $|I^e|$ - the number of environment invariants
    \item $|L^e|$ - the number of justice assumptions
    \item $|\mathcal{Y}|$ - the number of system variables
    \item $|I^s|$ - the number of system invariants
    \item $|L^s|$ - the number of justice guarantees
\end{itemize}

\begin{table}[h]
\centering
\begin{tabularx}{\textwidth}{l|RRR|RRR}
\toprule

\multirow{2}{*}{Case Study} 
& \multicolumn{3}{l|}{Assumptions} 
& \multicolumn{3}{l}{Guarantees} \\

& $|\mathcal{X}|$ & $|I^e|$ & $|L^e|$
& $|\mathcal{Y}|$ & $|I^s|$ & $|L^s|$ \\

\midrule

arbiter & 3 & 1 & 0 & 2 & 4 & 0 \\
lift & 4 & 11 & 0 & 3 & 6 & 0 \\
minepump & 2 & 2 & 0 & 1 & 2 & 0 \\
traffic single & 3 & 2 & 1 & 1 & 1 & 1 \\
traffic updated & 3 & 4 & 1 & 2 & 4 & 0 \\
\bottomrule
\end{tabularx}
\caption{Specification sizes}
\label{tab:spec_sizes}
\end{table}

%% file: appendices/C-case_studies/C.2-specifications_and_violation_traces.tex
\subsection{Specifications \& Violation Traces}
\label{apx:specifications_violation_traces}

We showcase below the exact case studies evaluated in this work in pairs of:
\begin{enumerate}
    \item The GR(1) specification, written in Spectra, using Dwyer patterns to represent GR(1) equivalent patterns that are not syntactically allowed by Spectra (e.g. $G(a\rightarrow F(b))$ will be represented in Spectra as $pRespondsToS(a,b)$);
    \item The violation trace that intentionally violates one (or multiple) environment assumptions.
\end{enumerate}

\subsubsection{Arbiter}
\label{apx:specification_violation_arbiter}
\input{code/spectra_specs/arbiter_spec}
\fbox{%
\begin{minipage}{\textwidth}
\[
\tau = [(\textcolor{red}{\neg a,\neg r1,\neg r2},
\textcolor{green!60!black}{\neg g1,\neg g2})]
\]%
\end{minipage}
}
\captionof{figure}{Arbiter case study}

\subsubsection{Lift}
\label{apx:specification_violation_lift}
\input{code/spectra_specs/lift_spec}
\fbox{%
\begin{minipage}{\textwidth}
\[
\tau = [(\textcolor{red}{\neg b1,\neg b2,\neg b3,c},
\textcolor{green!60!black}{f1,\neg f2,\neg f3}),
(\textcolor{red}{b1, b2, b3,\neg c},
\textcolor{green!60!black}{f1,\neg f2,\neg f3})]
\]%
\end{minipage}
}
\captionof{figure}{Lift case study}

\subsubsection{Minepump}
\label{apx:specification_violation_minepump}
\input{code/spectra_specs/minepump_spec}
\fbox{%
\begin{minipage}{\textwidth}
\[
\tau = [(\textcolor{red}{\neg highwater,\neg methane},\textcolor{green!60!black}{\neg pump}),(\textcolor{red}{highwater,methane},\textcolor{green!60!black}{\neg pump})]
\]%
\end{minipage}
}
\captionof{figure}{Minepump case study}

\subsubsection{Traffic Single}
\label{apx:specification_violation_traffic_single}
\input{code/spectra_specs/traffic-single_spec}
\fbox{%
\begin{minipage}{\textwidth}
\[
\tau = [(\textcolor{red}{ car,\neg emergency, police},\textcolor{green!60!black}{ green}),(\textcolor{red}{ car,\neg emergency, \neg police},\textcolor{green!60!black}{\neg green})]
\]%
\end{minipage}
}
\captionof{figure}{Traffic single case study}

\subsubsection{Traffic Updated}
\label{apx:specification_violation_updated}
\input{code/spectra_specs/traffic-updated_spec}
\fbox{%
\begin{minipage}{\textwidth}
\begin{align*}
\tau = [&(\textcolor{red}{ carA,carB, emergency},\textcolor{green!60!black}{ \neg greenA, \neg greenB}),\\ &(\textcolor{red}{ \neg carA,\neg carB, \neg emergency},\textcolor{green!60!black}{ \neg greenA, \neg greenB})]
\end{align*}
\end{minipage}
}
\captionof{figure}{Traffic updated case study}

%% file: code/spectra_specs/arbiter_spec.tex
\begin{lstlisting}[style=SpectraStyle]
module Arbiter

env boolean a;
env boolean r1;
env boolean r2;
sys boolean g1;
sys boolean g2;

assumption -- a_always
	alw(a);
guarantee -- guarantee1
	pRespondsToS(r1,g1); -- equiv to G(r1->F(g1));
guarantee -- guarantee2
	pRespondsToS(r2,g2); -- equiv to G(r2->F(g2));
guarantee -- guarantee3
	alw(!a->!g1&!g2);
guarantee -- guarantee4
	alw(!g1 | !g2);

pattern pRespondsToS(s, p) {
  var { S0, S1} state;

  // initial assignments: initial state
  ini state=S0;

  // safety this and next state
  alw ((state=S0 & ((!s) | (s & p)) & next(state=S0)) |
  (state=S0 & (s & !p) & next(state=S1)) |
  (state=S1 & (p) & next(state=S0)) |
  (state=S1 & (!p) & next(state=S1)));

  // equivalence of satisfaction
  alwEv (state=S0);
}
\end{lstlisting}

%% file: code/spectra_specs/lift_spec.tex
\begin{lstlisting}[style=SpectraStyle]
module Lift

env boolean b1;
env boolean b2;
env boolean b3;
env boolean c;
sys boolean f1;
sys boolean f2;
sys boolean f3;

assumption -- initial_assumptions
	!b1 & !b2 & !b3;
assumption -- button1_off_at_floor1
	alw(b1 & f1 -> next(!b1));
assumption -- button2_off_at_floor2
	alw(b2 & f2 -> next(!b2));
assumption -- button3_off_at_floor3
	alw(b3 & f3 -> next(!b3));
assumption -- button1_stays_on
	alw(b1 & !f1 & !c -> next(b1));
assumption -- button2_stays_on
	alw(b2 & !f2 & !c -> next(b2));
assumption -- button3_stays_on
	alw(b3 & !f3 & !c -> next(b3));
assumption -- one_floor
	alw(f1 | f2 | f3);
assumption -- one_floor_only1
	alw(f1 -> !f2 & !f3);
assumption -- one_floor_only2
	alw(f2 -> !f1 & !f3);
assumption -- one_floor_only3
	alw(f3 -> !f1 & !f2);
assumption -- not_all_buttons_on
	alw(!b1 | !b2 | !b3);
    
guarantee -- initial_guarantees
	f1 & !f2 & !f3;
guarantee -- move_one_max1
	alw(f1 -> (next(f1) & !b2 & !b3) | (next(f2) & b2) | (next(f2) & b3) | next(f1) & c);
guarantee -- move_one_max2
	alw(f2 -> (next(f1) & b1) | (next(f2) & !b1 & !b3) | (next(f3) & b3) | next(f2) & c);
guarantee -- move_one_max3
	alw(f3 -> (next(f3) & !b1 & !b2) | (next(f2) & b1) | (next(f2) & b2) | next(f3) & c);
guarantee -- button1_answered
	pRepondsToS(b1, (f1 | c)); -- equivalent to G(b1 -> F(f1 | c));
guarantee -- button2_answered
	pRepondsToS(b2, (f2 | c)); -- equivalent to G(b2 -> F(f2 | c));
guarantee -- button3_answered
	pRepondsToS(b3, (f3 | c)); -- equivalent to G(b3 -> F(f3 | c));

-- pattern omitted
\end{lstlisting}

%% file: code/spectra_specs/minepump_spec.tex
\begin{lstlisting}[style=SpectraStyle]
module Minepump

env boolean highwater;
env boolean methane;
sys boolean pump;

assumption -- initial_assumption
    !highwater & !methane;
assumption -- assumption1
	alw(PREV(pump)&pump->next(!highwater));
assumption -- assumption2
	alw(!highwater|!methane);
guarantee -- initial_guarantee
    !pump;
guarantee -- guarantee1
	alw(highwater->next(pump));
guarantee -- guarantee2
	alw(methane->next(!pump));
\end{lstlisting}

%% file: code/spectra_specs/traffic-single_spec.tex
\begin{lstlisting}[style=SpectraStyle]
module TrafficE2

env boolean car;
env boolean emergency;
env boolean police;
sys boolean green;

assumption -- car_idle_when_red
	alw(car & !green -> next(car));
assumption -- car_moves_when_green
	alw(car & green -> next(!car));
assumption -- not_police_often
	alwEv(!police);
guarantee -- green_often
	pRepondsToS(car,green); -- equivalent to G(car -> F(green));
guarantee -- no_car_often
	alwEv(!car);

-- pattern omitted
\end{lstlisting}

%% file: code/spectra_specs/traffic-updated_spec.tex
\begin{lstlisting}[style=SpectraStyle]
module TrafficE2

env boolean carA;
env boolean carB;
env boolean emergency;
sys boolean greenA;
sys boolean greenB;

assumption -- carA_idle_when_red
	alw(carA & !greenA -> next(carA));
assumption -- carB_idle_when_red
	alw(carB & !greenB -> next(carB));
assumption -- carA_moves_when_green
	alw(carA & greenA -> next(!carA));
assumption -- carB_moves_when_green
	alw(carB & greenB -> next(!carB));
assumption -- no_emergency_often
	alwEv(!emergency);

guarantee -- lights_not_both_red
	alw(!greenA | !greenB);
guarantee -- carA_leads_to_greenA
	pRespondsToS(carA,greenA); -- equiv to G(carA -> F(greenA));
guarantee -- carB_lead_to_greenB
	pRespondsToS(carB,greenB); -- equiv to G(carB -> F(greenB));
guarantee -- red_when_emergency
	alw(emergency -> !greenA & !greenB);

-- pattern omitted
\end{lstlisting}

%% file: appendices/C-case_studies/C.3-trivial_solutions.tex
\subsection{Trivial Solutions}
\label{apx:trivial_solutions}

We showcase below the trivial solutions for each case study.

\subsubsection{Arbiter}
\label{apx:trivial_solution_arbiter}
\input{code/spectra_specs/trivial/arbiter_0}

\subsubsection{Minepump}
\label{apx:trivial_solution_minepump}
\input{code/spectra_specs/trivial/minepump_0}
\input{code/spectra_specs/trivial/minepump_1}

\subsubsection{Lift}
\label{apx:trivial_solution_lift}
\input{code/spectra_specs/trivial/lift_0}

\subsubsection{Traffic Single}
\label{apx:trivial_solution_traffic_single}
\input{code/spectra_specs/trivial/traffic_single_0}

\subsubsection{Traffic Updated}
\label{apx:trivial_solution_traffic_updated}
\input{code/spectra_specs/trivial/traffic_updated_0}

%% file: code/spectra_specs/trivial/arbiter_0.tex
\begin{lstlisting}[style=SpectraStyle]
module Arbiter

env boolean a;
sys boolean g1;
sys boolean g2;
env boolean r1;
env boolean r2;

guarantee -- guarantee1_1
	pRespondsToS(r1=true,g1=true); -- equivalent to G((r1->F(g1)));

guarantee -- guarantee2_1
	pRespondsToS(r2=true,g2=true); -- equivalent to G((r2->F(g2)));

guarantee -- guarantee4
	alw((g1=false|g2=false));

-- pattern omitted
\end{lstlisting}

%% file: code/spectra_specs/trivial/minepump_0.tex
\begin{lstlisting}[style=SpectraStyle]
module Minepump

env boolean highwater;
env boolean methane;
sys boolean pump;

assumption -- initial_assumption
	(highwater=false&methane=false);

guarantee -- initial_guarantee
	pump=false;

guarantee -- guarantee1_1
	alw((highwater=true->next(pump=true)));

assumption -- assumption1_1
	alw(((PREV(pump=true)&pump=true)->next(highwater=false)));

\end{lstlisting}

%% file: code/spectra_specs/trivial/minepump_1.tex
\begin{lstlisting}[style=SpectraStyle]
module Minepump

env boolean highwater;
env boolean methane;
sys boolean pump;

assumption -- initial_assumption
	(highwater=false&methane=false);

guarantee -- initial_guarantee
	pump=false;

guarantee -- guarantee2_1
	alw((methane=true->next(pump=false)));

assumption -- assumption1_1
	alw(((PREV(pump=true)&pump=true)->next(highwater=false)));

\end{lstlisting}

%% file: code/spectra_specs/trivial/lift_0.tex
\begin{lstlisting}[style=SpectraStyle]
module Lift

env boolean b1;
env boolean b2;
env boolean b3;
env boolean c;
sys boolean f1;
sys boolean f2;
sys boolean f3;

assumption -- initial_assumptions
	((b1=false&b2=false)&b3=false);

assumption -- button1_off_at_floor1
	alw(((b1=true&f1=true)->next(b1=false)));

assumption -- button2_off_at_floor2
	alw(((b2=true&f2=true)->next(b2=false)));

assumption -- button3_off_at_floor3
	alw(((b3=true&f3=true)->next(b3=false)));

assumption -- button1_stays_on
	alw((((b1=true&f1=false)&c=false)->next(b1=true)));

assumption -- button2_stays_on
	alw((((b2=true&f2=false)&c=false)->next(b2=true)));

assumption -- button3_stays_on
	alw((((b3=true&f3=false)&c=false)->next(b3=true)));

guarantee -- initial_guarantees
	((f1=true&f2=false)&f3=false);

assumption -- one_floor
	alw(((f1=true|f2=true)|f3=true));

assumption -- one_floor_only1
	alw((f1=true->(f2=false&f3=false)));

assumption -- one_floor_only2
	alw((f2=true->(f1=false&f3=false)));

assumption -- one_floor_only3
	alw((f3=true->(f1=false&f2=false)));

guarantee -- move_one_max1
	alw((f1=true->(((((next(f1=true)&b2=false)&b3=false)|(next(f2=true)&b2=true))|(next(f2=true)&b3=true))|(next(f1=true)&c=true))));

guarantee -- move_one_max2
	alw((f2=true->((((next(f1=true)&b1=true)|((next(f2=true)&b1=false)&b3=false))|(next(f3=true)&b3=true))|(next(f2=true)&c=true))));

guarantee -- move_one_max3
	alw((f3=true->(((((next(f3=true)&b1=false)&b2=false)|(next(f2=true)&b1=true))|(next(f2=true)&b2=true))|(next(f3=true)&c=true))));

guarantee -- button1_answered
	pRespondsToS(b1=true,(f1=true|c=true)); -- equivalent to G((b1->F((f1|c))));

guarantee -- button2_answered
	pRespondsToS(b2=true,(f2=true|c=true)); -- equivalent to G((b2->F((f2|c))));

guarantee -- button3_answered
	pRespondsToS(b3=true,(f3=true|c=true)); -- equivalent to G((b3->F((f3|c))));

-- pattern omitted
\end{lstlisting}

%% file: code/spectra_specs/trivial/traffic_single_0.tex
\begin{lstlisting}[style=SpectraStyle]
module TrafficE2

env boolean car;
env boolean emergency;
sys boolean green;
env boolean police;

guarantee -- green_often
	pRespondsToS(car=true,green=true); -- equivalent to G((car->F(green)));

assumption -- not_police_often
	alwEv(police=false);

assumption -- car_idle_when_red
	alw(((car=true&green=false)->next(car=true)));

-- pattern omitted
\end{lstlisting}

%% file: code/spectra_specs/trivial/traffic_updated_0.tex
\begin{lstlisting}[style=SpectraStyle]
module TrafficE2

env boolean carA;
env boolean carB;
env boolean emergency;
sys boolean greenA;
sys boolean greenB;

assumption -- no_emergency_often
	alwEv(emergency=false);

guarantee -- lights_not_both_red
	alw((greenA=false|greenB=false));

guarantee -- carA_leads_to_greenA
	pRespondsToS(carA=true,greenA=true); -- equivalent to G((carA->F(greenA)));

guarantee -- carB_lead_to_greenB
	pRespondsToS(carB=true,greenB=true); -- equivalent to G((carB->F(greenB)));

guarantee -- red_when_emergency
	alw((emergency=true->(greenA=false&greenB=false)));

assumption -- carA_moves_when_green
	alw(((carA=true&greenA=true)->next(carA=false)));

assumption -- carB_moves_when_green
	alw(((carB=true&greenB=true)->next(carB=false)));

-- pattern omitted
\end{lstlisting}

%% file: appendices/C-case_studies/C.4-preferred_learned_solutions.tex
\subsection{Preferred Learned Solutions}
\label{apx:preferred_solutions}

We showcase below the solutions discovered when using Algorithm~\ref{alg:learningRInt} for each case study.

\subsubsection{Arbiter}
\label{apx:preferred_learned_arbiter}
\input{code/spectra_specs/preferred_learned/arbiter_0}

\subsubsection{Minepump}
\label{apx:preferred_learned_minepump}
\input{code/spectra_specs/preferred_learned/minepump_0}

\subsubsection{Lift}
\label{apx:preferred_learned_lift}
\input{code/spectra_specs/preferred_learned/lift_0}

\subsubsection{Traffic Single}
\label{apx:preferred_learned_traffic_single}
\input{code/spectra_specs/preferred_learned/traffic_single_0}

\subsubsection{Traffic Updated}
\label{apx:preferred_learned_traffic_updated}
\input{code/spectra_specs/preferred_learned/traffic_updated_0}

%% file: code/spectra_specs/preferred_learned/arbiter_0.tex
\begin{lstlisting}[style=SpectraStyle]
module Arbiter

env boolean a;
sys boolean g1;
sys boolean g2;
env boolean r1;
env boolean r2;

guarantee -- guarantee1
	pRespondsToS(r1=true,g1=true); -- equivalent to G((r1->F(g1)))

guarantee -- guarantee2
	pRespondsToS(r2=true,g2=true); -- equivalent to G((r2->F(g2)))

guarantee -- guarantee3
	alw((a=false->(g1=false&g2=false)));

guarantee -- guarantee4
	alw((g1=false|g2=false));

assumption -- a_always_0
	alwEv(a=true);

assumption -- a_always_1
	alw((a=true|(g1=false&g2=false&r1=false&r2=false)));

\end{lstlisting}

%% file: code/spectra_specs/preferred_learned/minepump_0.tex
\begin{lstlisting}[style=SpectraStyle]
module Minepump

env boolean highwater;
env boolean methane;
sys boolean pump;

assumption -- initial_assumption
	(highwater=false&methane=false);

assumption -- assumption1_1
	alw(((PREV(pump=true)&pump=true)->next(highwater=false)));

assumption -- assumption2_1
	alw((pump=true->(highwater=false|methane=false)));

guarantee -- initial_guarantee
	pump=false;

guarantee -- guarantee1_1_0
	alw((highwater=true->next(pump=true)));

guarantee -- guarantee2_1_0
	alw((methane=true->(next(pump=false)|next(methane=true)&next(highwater=true))));

guarantee -- guarantee2_1_8
	pRespondsToS(methane=true,next(pump=false)); --equivalent to G(methane->X(!pump))
    
-- pattern omitted
\end{lstlisting}

%% file: code/spectra_specs/preferred_learned/lift_0.tex
\begin{lstlisting}[style=SpectraStyle]
module Lift

env boolean b1;
env boolean b2;
env boolean b3;
env boolean c;
sys boolean f1;
sys boolean f2;
sys boolean f3;

assumption -- initial_assumptions
	((b1=false&b2=false)&b3=false);

assumption -- button1_off_at_floor1
	alw(((b1=true&f1=true)->next(b1=false)));

assumption -- button2_off_at_floor2
	alw(((b2=true&f2=true)->next(b2=false)));

assumption -- button3_off_at_floor3
	alw(((b3=true&f3=true)->next(b3=false)));

assumption -- button1_stays_on
	alw((((b1=true&f1=false)&c=false)->next(b1=true)));

assumption -- button2_stays_on
	alw((((b2=true&f2=false)&c=false)->next(b2=true)));

assumption -- button3_stays_on
	alw((((b3=true&f3=false)&c=false)->next(b3=true)));

guarantee -- initial_guarantees
	((f1=true&f2=false)&f3=false);

assumption -- one_floor
	alw(((f1=true|f2=true)|f3=true));

assumption -- one_floor_only1
	alw((f1=true->(f2=false&f3=false)));

assumption -- one_floor_only2
	alw((f2=true->(f1=false&f3=false)));

assumption -- one_floor_only3
	alw((f3=true->(f1=false&f2=false)));

guarantee -- move_one_max1
	alw((f1=true->(((((next(f1=true)&b2=false)&b3=false)|(next(f2=true)&b2=true))|(next(f2=true)&b3=true))|(next(f1=true)&c=true))));

guarantee -- move_one_max2
	alw((f2=true->((((next(f1=true)&b1=true)|((next(f2=true)&b1=false)&b3=false))|(next(f3=true)&b3=true))|(next(f2=true)&c=true))));

guarantee -- move_one_max3
	alw((f3=true->(((((next(f3=true)&b1=false)&b2=false)|(next(f2=true)&b1=true))|(next(f2=true)&b2=true))|(next(f3=true)&c=true))));

guarantee -- button1_answered
	pRespondsToS(b1=true, (f1=true|c=true)); -- equivalent to G((b1->F((f1|c))));

guarantee -- button2_answered
	pRespondsToS(b2=true, (f2=true|c=true)); -- equivalent to G((b2->F((f2|c))));

guarantee -- button3_answered
	pRespondsToS(b3=true, (f3=true|c=true)); -- equivalent to G((b3->F((f3|c))));

assumption -- not_all_buttons_on_0
	alw((c=true->((b1=false|b2=false)|b3=false)));

assumption -- not_all_buttons_on_1
	alw((f1=false->((b1=false|b2=false)|b3=false)));

-- pattern omitted
\end{lstlisting}

%% file: code/spectra_specs/preferred_learned/traffic_single_0.tex
\begin{lstlisting}[style=SpectraStyle]
module TrafficE2

env boolean car;
env boolean emergency;
sys boolean green;
env boolean police;

guarantee -- green_often
	alw((car=true->F(green=true)));

guarantee -- no_car_often
	alwEv(car=false);

assumption -- not_police_often
	alwEv(police=false);

assumption -- car_idle_when_red
	alw(((car=true&green=false)->next(car=true)));

assumption -- car_moves_when_green_0
	alw((((car=true&green=true)&police=false)->next(car=false)));

assumption -- car_moves_when_green_1
	pRespondsToS(car=true&green=true,next(car=false)); -- equivalent to G(((car&green)->F(next(!car))));

-- pattern omitted
\end{lstlisting}

%% file: code/spectra_specs/preferred_learned/traffic_updated_0.tex
\begin{lstlisting}[style=SpectraStyle]
module TrafficE2

env boolean carA;
env boolean carB;
env boolean emergency;
sys boolean greenA;
sys boolean greenB;

assumption -- no_emergency_often
	alwEv(emergency=false);

guarantee -- lights_not_both_red
	alw((greenA=false|greenB=false));

guarantee -- carA_leads_to_greenA
	pRespondsToS(carA=true,greenA=true); -- equivalent to G((carA->F(greenA)));

guarantee -- carB_lead_to_greenB
	pRespondsToS(carB=true,greenB=true); -- equivalent to G((carB->F(greenB)));

guarantee -- red_when_emergency
	alw((emergency=true->(greenA=false&greenB=false)));

assumption -- carA_idle_when_red_0
	alw((((carA=true&greenA=false)&carB=false)->next(carA=true)));

assumption -- carB_idle_when_red_1
	alw((((carB=true&greenB=false)&carA=false)->next(carB=true)));

assumption -- carA_moves_when_green
	alw(((carA=true&greenA=true)->next(carA=false)));

assumption -- carB_moves_when_green
	alw(((carB=true&greenB=true)->next(carB=false)));

assumption -- carB_idle_when_red_3
	alw(((carB=true&greenB=false)->(next(carB=true)|next(carA=false))));

assumption -- carA_idle_when_red_2
	alw(((carA=true&greenA=false)->(next(carA=true)|greenB=false)));

assumption -- carA_idle_when_red_3
	alw(((carA=true&greenA=false)->(next(carA=true)|next(emergency=false))));

assumption -- carB_idle_when_red_4
	alw(((carB=true&greenB=false)->(next(carB=true)|greenA=false)));

assumption -- carB_idle_when_red_5
	pRespondsToS(carB=true&greenB=false,next(carB=true)); -- equivalent to G(((carB&!greenB)->F(next(carB))));

assumption -- carA_idle_when_red_4
	alw(((carA=true&greenA=false)->(next(carA=true)|next(carB=false))));

assumption -- carA_idle_when_red_5
	pRespondsToS(carA=true&greenA=false,next(carA=true)); -- equivalent to G(((carA&!greenA)->F(next(carA))));

assumption -- carA_idle_when_red
	alw((((carA=true&greenA=false)&emergency=false)->next(carA=true)));

-- pattern omitted
\end{lstlisting}